\documentclass[11pt]{article}
\usepackage{graphicx} 
\usepackage{fullpage}
\usepackage{amsthm}
\usepackage{amsmath}
\usepackage{amssymb}
\usepackage{bm}
\usepackage{bbm}
\usepackage{mathtools}
\usepackage{hyperref}
\usepackage{physics}
\usepackage{subcaption}
\usepackage{multirow}
\usepackage{float} 
\usepackage{placeins} 

\usepackage{tikz}
\usetikzlibrary{positioning,arrows.meta,calc}
\usepackage{autonum} 

\usepackage{float} 
\usepackage[linesnumbered, ruled]{algorithm2e}
\SetKwInput{KwModel}{Model input}
\SetKwInput{KwTune}{Tuning input}

\usepackage{natbib}

\newtheorem{theorem}{Theorem}
\newtheorem{proposition}{Proposition}

\theoremstyle{plain}
\newtheorem{condition}{Condition}

\theoremstyle{plain}

\theoremstyle{plain}

\theoremstyle{definition}
\newtheorem{definition}{Definition}

\DeclareMathOperator*{\argmax}{arg\,max}
\DeclareMathOperator*{\argmin}{arg\,min}
\renewcommand{\dd}{\,\mathrm{d}}

\title{Covariate-dependent Nonparametric $g$-modeling for regression \\ via infinite Mixture-of-Expertizing class}
\author{%
  Akira Okazaki and Keisuke Yano\\
  The Institute of Statistical Mathematics, \\ 
  10-3 Midori-cho, Tachikawa, Tokyo, 190-8562, Japan.\\[2.5em]
  Contributing authors: \textcolor{blue}{okazaki@ism.ac.jp}; \textcolor{blue}{yano@ism.ac.jp};
}
\date{}

\begin{document}

\maketitle

\begin{abstract}
    Empirical Bayes $g$-modeling captures unit-level heterogeneity by estimating a latent prior distribution from observed data. In the existing formulations, however, the prior is shared by all units. In this paper, we develop a covariate-dependent $g$-modeling framework for regression in which the entire prior distribution of the regression coefficients is allowed to depend on covariates. We formulate the estimation of the prior
    as nonparametric maximum likelihood estimation (NPMLE) of the covariate-dependent prior, and show that the unrestricted problem is ill-posed. 
    To resolve this, we introduce the infinite Mixture-of-Expertizing class of conditional priors, under which the NPMLE is precisely a softmax-gated Mixture of Experts (MoE) whose number of experts is not fixed in advance but is determined by the data. 
    Building on a first-order optimality condition, we propose two exemplar-based estimation algorithms that select experts automatically, together with a post-hoc aggregation of experts for interpretation. On the theoretical side, we show that every conditional prior in the class is Lipschitz continuous in the covariates, and that aggregated softmax gates can approximate any continuous gate function.  The effectiveness of the proposed NPMLE is shown through application to synthetic datasets and real datasets. 
\end{abstract}

\section{Introduction}
Assume that we observe $n$ independent pairs $\{(x_i,y_i):i=1,\ldots,n\}$,
where $y_i$ is a scalar response and $x_i$ is a $p$-dimensional covariate vector of explanatory variables. 
We consider the unit-level generalized linear regression model (unit-level GLM) formulated as
\begin{equation}
    \label{model: unit_level_GLM_ob}
    \begin{split}
        y_i \sim \mathrm{EF}(\mu_i,\phi), \
        \mu_{i}= g^{-1}(\eta_i),\   \eta_i = x_i^{\top}\beta_i,
        \ \beta_i\overset{\mathrm{i.i.d}}{\sim} G^{\ast}(\beta), \quad  \quad i=1,\ldots,n,
    \end{split}
\end{equation}
where $\beta_i\in\mathbb R^{p}$ is a latent regression coefficient vector for unit $i$ drawn from a prior probability measure $G^{\ast}$, $\eta_i$ is a linear predictor,
$g^{-1}(\cdot)$ is a known inverse link function, and $\mathrm{EF}(\mu_i,\phi)$ denotes an exponential family distribution with conditional mean $\mu_i=\mathbb{E}[Y|x_i,\beta_i]$
and dispersion parameter $\phi$.

In the unit-level GLM \eqref{model: unit_level_GLM_ob}, each unit has its own latent regression coefficients, which allows the linear relationship between the response and the explanatory variables to vary across units.
This setting, in which unit-level parameters are drawn from an unknown prior, is the classical setting of empirical Bayes (EB) and compound decision theory \citep{Robbins1950-fh,Brown2009-th}; see \cite{Ignatiadis2025-nz} for a recent overview.
The paradigm of estimating the prior $G^{\ast}$ from the observations is known as $g$-modeling \cite[e.g.,][]{Efron2014-yr,Efron2016-ep}.
 The estimated prior $\hat{G}$ enables unit-level posterior inference, and provides an interpretation of the heterogeneity across units through $\hat{G}$.

A classical nonparametric approach to $g$-modeling is known as the nonparametric maximum likelihood estimator (NPMLE) \citep[e.g.,][]{Kiefer1956-dj,Laird1978-rf,Lindsay1995-mi} that is obtained by maximizing the marginal likelihood over prior measures.
For the unit-level GLM \eqref{model: unit_level_GLM_ob}, the NPMLE is defined as
\begin{equation}
    \label{def:NPMLE_regression}
    \hat{G} \in \argmax_{G\in\mathcal{G}} \sum_{i=1}^{n}\log f_G(y_i|x_i),
\end{equation}
where $f_{G}(y|x)$ is the conditional marginal likelihood integrated over  $G$ and $\mathcal{G}$ is a set of candidate prior measures.
For the unit-level Gaussian linear model, \cite{Jiang2025-am} recently provided a theoretical justification of NPMLE together with a practical estimation procedure called the exemplar method.

There is, however, 
a structural limitation in the framework \eqref{model: unit_level_GLM_ob}: that is, the same prior $G^{\ast}$ is shared by all units, so that the model assumes
the distribution of the heterogeneity to be identical regardless of the observed attributes of the units. In many applications, however, 
the heterogeneity is structured by covariates.
In housing market analysis, a metropolitan market is considered to consist of \emph{submarkets}, within each of which the relationship between the price
and the housing characteristics is distinct (\cite{Goodman2003-ed}),
and the submarket to which a house belongs is closely tied to its location and characteristics.
In epidemiology and the life sciences, the association between a risk factor and a health outcome, or the effect of a treatment, often differs across subpopulations \citep{Kent2018-hte}, and part of this heterogeneity is systematically explained by demographic and socioeconomic attributes such as age, sex, and income, while the remainder is idiosyncratic to each individual.
Capturing such covariate-structured heterogeneity requires the prior to vary with the covariates.

In this paper, we develop a covariate-dependent $g$-modeling framework in which the prior of the latent regression coefficients depends on covariates, as in the model
\begin{equation}
    \label{model: unit_level_GLM_proposed}
    \begin{split}
        y_i \sim \mathrm{EF}(\mu_i,\phi), \
        \mu_{i}= g^{-1}(\eta_i),\   \eta_i = x_i^{\top}\beta_i,
        \ \beta_i\overset{\mathrm{ind}}{\sim} G^{\ast}(\beta|z_i), \quad  \quad i=1,\ldots,n,
    \end{split}
\end{equation}
where $z_i$ is an observed $d$-dimensional covariate vector for unit $i$.
The covariate $z_i$ used in the prior may not be identical to the covariate $x_i$ in the linear regression term. 
For instance, geographic coordinates such as latitude and longitude serve as $z_i$ when spatial variation of the coefficients is of interest. 
This extension decomposes the heterogeneity across units into a component explained by the covariates captured by the movement of the conditional prior along $z$
and an irreducible component among units with similar covariates captured by the conditional prior at each $z$. 

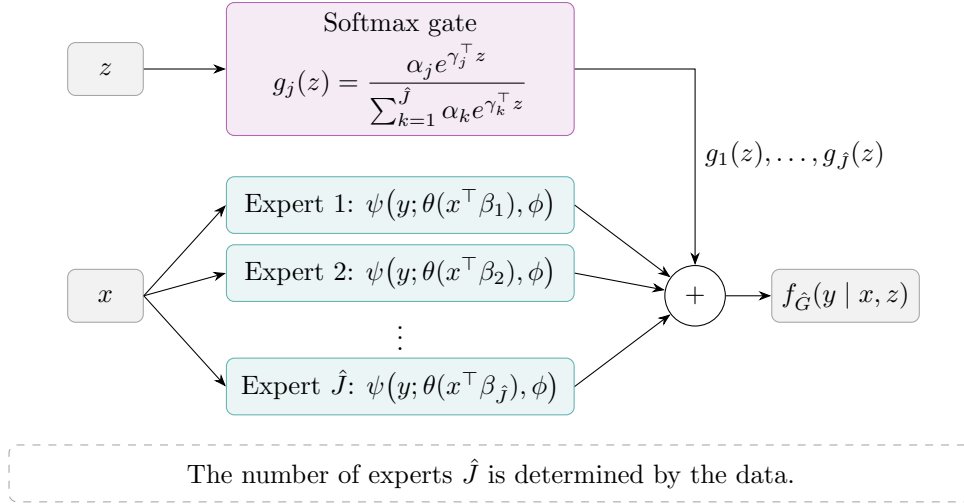
\begin{figure}[htbp]
\centering
\begin{tikzpicture}[
  >=Stealth, font=\small,
  inp/.style={draw=gray!70, fill=gray!10, rounded corners=3pt, minimum width=1.0cm, minimum height=0.7cm, align=center},
  gate/.style={draw=violet!60, fill=violet!8, rounded corners=3pt, minimum width=4.6cm, minimum height=1.5cm, align=center},
  expert/.style={draw=teal!60, fill=teal!8, rounded corners=3pt, minimum width=4.6cm, minimum height=0.75cm},
  sumnode/.style={draw, circle, minimum size=0.8cm, inner sep=0pt},
]
  \node[gate]   (gate) at (0,2.7)  {Softmax gate\\[2pt] $g_j(z)=\dfrac{\alpha_j e^{\gamma_j^{\top}z}}{\sum_{k=1}^{\hat J}\alpha_k e^{\gamma_k^{\top}z}}$};
  \node[expert] (e1)   at (0,0.9)  {Expert $1$: $\psi\bigl(y;\theta(x^{\top}\beta_1),\phi\bigr)$};
  \node[expert] (e2)   at (0,0)    {Expert $2$: $\psi\bigl(y;\theta(x^{\top}\beta_2),\phi\bigr)$};
  \node                at (0,-0.75){$\vdots$};
  \node[expert] (eJ)   at (0,-1.5) {Expert $\hat J$: $\psi\bigl(y;\theta(x^{\top}\beta_{\hat J}),\phi\bigr)$};

  \node[inp] (z) at (-3.9,2.7)  {$z$};
  \node[inp] (x) at (-3.9,-0.3) {$x$};
  \draw[->] (z) -- (gate);
  \foreach \e in {e1,e2,eJ} \draw[->] (x.east) -- (\e.west);

  \node[sumnode] (sum) at (3.9,-0.3) {$+$};
  \draw[->] (gate.east) -| node[pos=0.72, right, align=left] {$g_1(z),\ldots,g_{\hat J}(z)$} (sum.north);
  \foreach \e in {e1,e2,eJ} \draw[->] (\e.east) -- (sum);
  \node[inp, right=0.6cm of sum] (out) {$f_{\hat G}(y\mid x,z)$};
  \draw[->] (sum) -- (out);

  \node[draw=gray!70, dashed, rounded corners=3pt, text width=12.4cm, align=center, inner sep=5pt]
    at (1.2,-2.65)
    {The number of experts $\hat J$ is determined by the data.};
\end{tikzpicture}
\caption{Schematic of the iMoE-NPMLE. The covariate $z$ enters the softmax gate and the explanatory variable $x$ enters the GLM experts. The marginal conditional density \eqref{def:proposed_NPMLE_discrete} is the gate-weighted sum of the expert densities. Unlike the conventional MoE, the number of experts is not fixed in advance but is given by the number of support points of the estimated prior probability measure. Notations follow Section \ref{sec:npmle}.}
\label{fig:schematic_moe}
\end{figure}

\begin{figure}[h]
    \centering
    \begin{subfigure}[b]{0.32\textwidth}\centering\includegraphics[width=\linewidth]{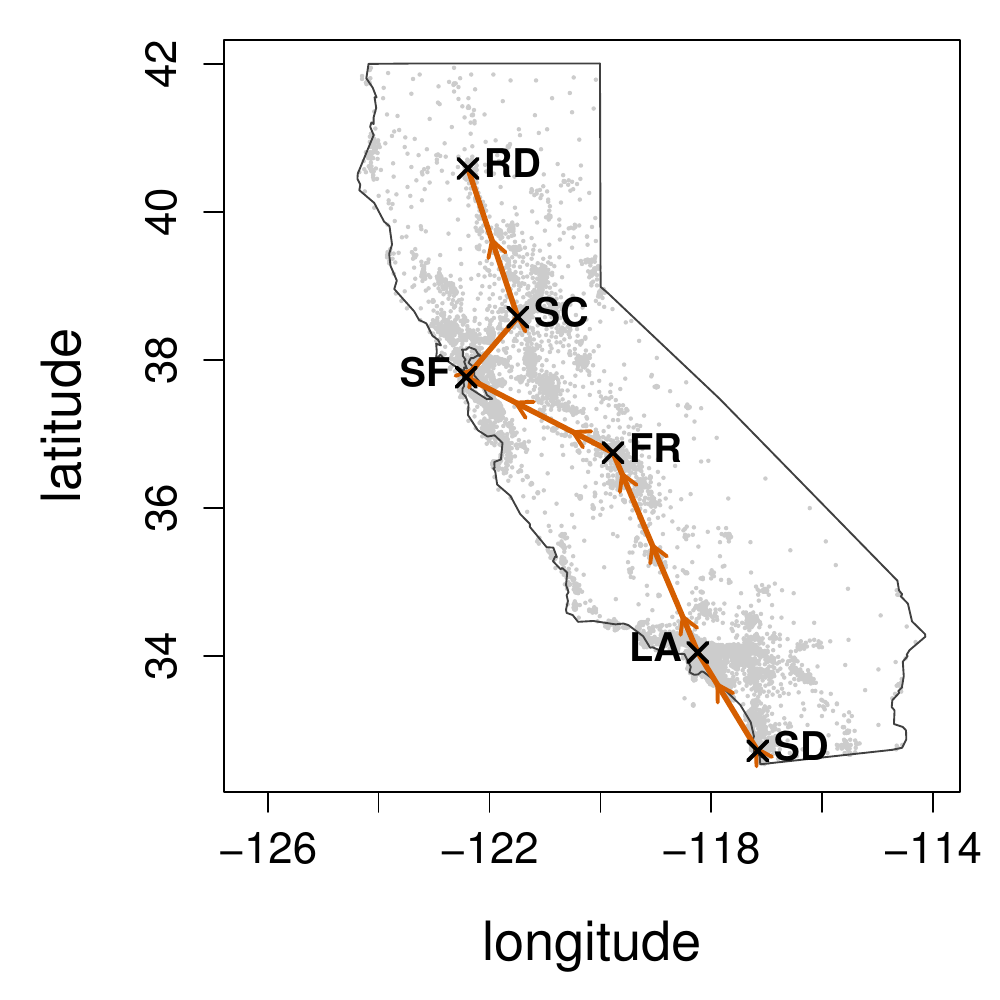}\caption{the path $z(t)$}
    \label{fig:intro_path}
    \end{subfigure}\hfill
    \begin{subfigure}[b]{0.32\textwidth}\centering\includegraphics[width=\linewidth]{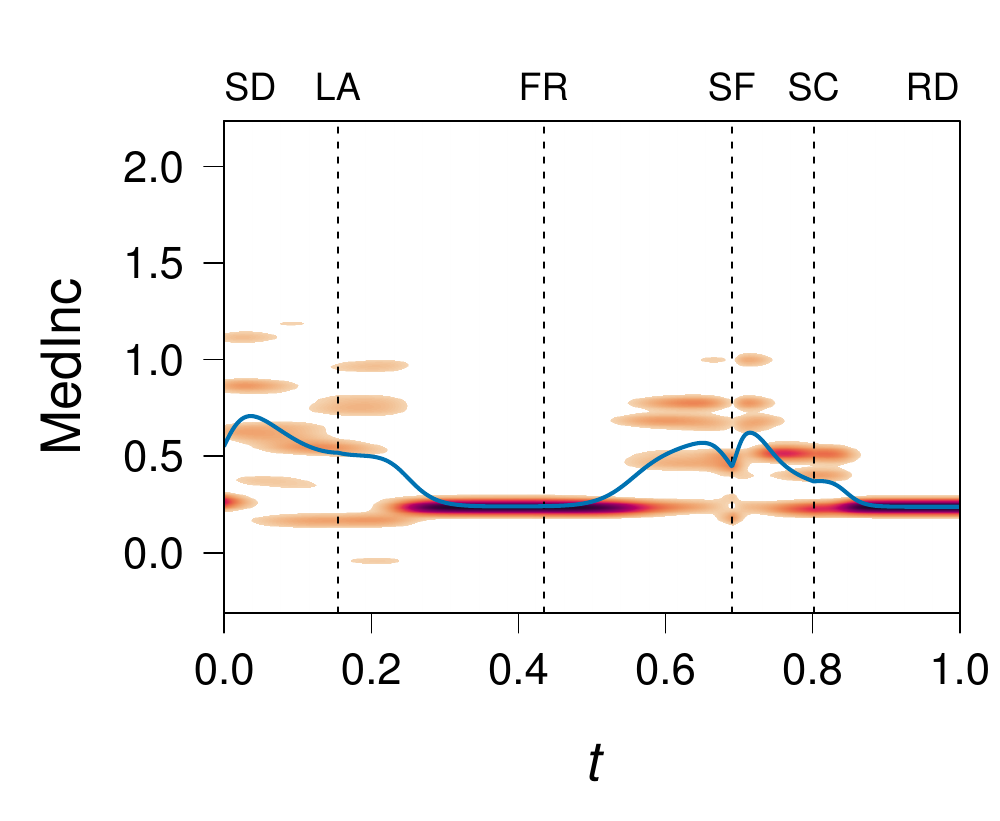}\caption{MedInc}
    \label{fig:intro_medinc}\end{subfigure}\hfill
    \begin{subfigure}[b]{0.32\textwidth}\centering\includegraphics[width=\linewidth]{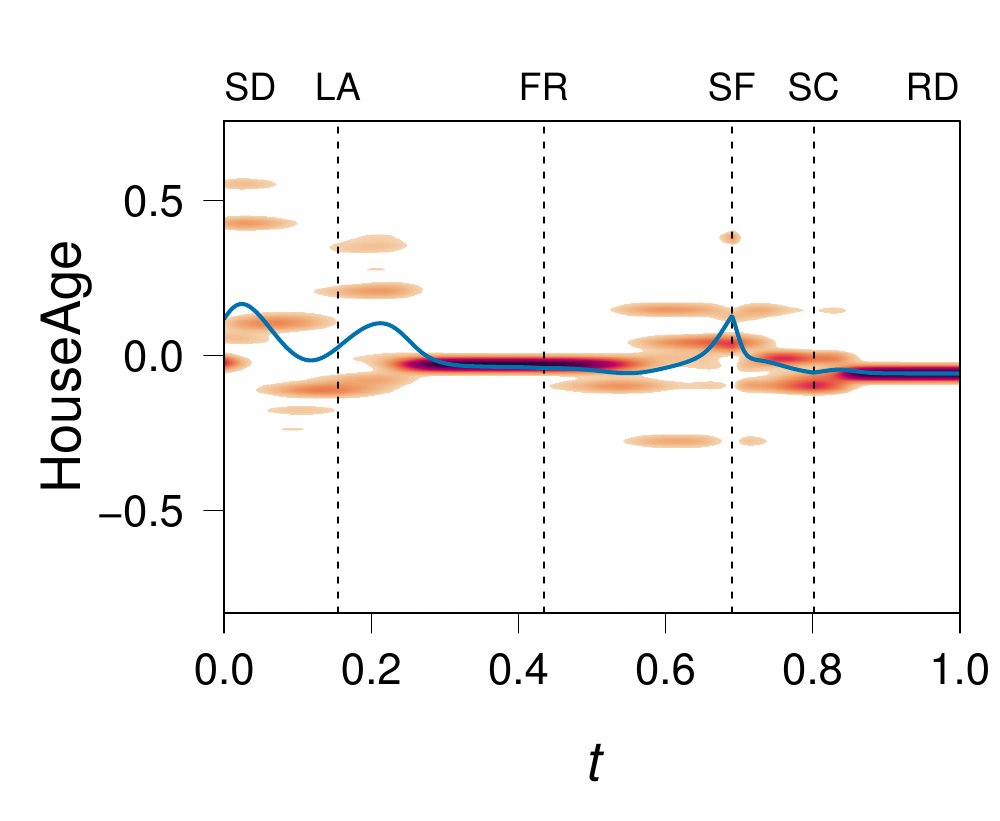}\caption{HouseAge}\label{fig:intro_houseage}\end{subfigure}
    \caption{Visualization of the estimated covariate-dependent prior $\hat{G}(\beta|z)$ from the California housing analysis in Section \ref{sec:california_housing}, where the covariate $z_i$ is the geographic coordinates for each unit. (a): a path $z(t):[0,1]\mapsto \mathbb{R}^2$ through the six major cities with arrows in the direction of increasing $t$, where SD, LA, FR, SF, SC, and RD stand for San Diego, Los Angeles, Fresno, San Francisco, Sacramento, and Redding. (b), (c): the marginal prior $\hat{G}(\beta_c|z(t))$ of the coefficient of the median income (MedInc) and of the median house age (HouseAge) along the path. The horizontal axis is $t$, the vertical axis is the value of the coefficient, and the shading represents the probability mass of the regression coefficient. The blue line is the conditional mean of the coefficient and the dashed vertical lines represent the cities.}
    \label{fig:intro_estimated_prior}
\end{figure}

Estimating $G^{\ast}(\beta|z)$ nonparametrically raises a difficulty. The direct extension of the NPMLE \eqref{def:NPMLE_regression} that maximizes over all conditional priors is ill-posed, since degenerate solutions interpolating each observation attain the maximum (see Section 2). Therefore, we restrict conditional probability measures to a reasonable subclass.
In this paper, we employ the \emph{infinite Mixture-of-Expertizing class} (iMoE class) that consists of an infinite mixture of softmax functions with a normalizing constant. 
The resulting NPMLE (iMoE-NPMLE; Figure \ref{fig:schematic_moe}) is the discrete conditional measure, and the conditional marginal density is precisely a Mixture-of-Experts (MoE) model \citep{Jacobs1991-zl} with a number of experts not specified, in which the covariate-dependent weight functions are called gates and the component densities are experts. Building on the exemplar methods \citep[e.g.,][]{Bohning2000-gq, Lashkari2007-qc,Li2019-mj,Soloff2024-ep,Jiang2025-am}, we propose two estimation algorithms based on a specific parameterization of a \emph{softmax gate}. The parameterization is naturally inherited from our construction of NPMLE, and helps perform support selection automatically during the optimization procedure. To further enhance interpretability, we propose an aggregation procedure that reduces the estimated prior to a few support points equipped with aggregated softmax gates.

Figure \ref{fig:intro_estimated_prior} displays the estimated conditional prior $\hat{G}(\beta|z)$ obtained by applying the proposed method to the California housing data, where the response is the median house value of a census block with six explanatory variables, and the geographic coordinates $z_i$ (with details in Section \ref{sec:california_housing}). The panels show the varying prior along a path $z(t)$ going through six major cities in the state (Figure \ref{fig:intro_path}).
The estimated prior clearly shifts its mass as the covariate values vary (Figure \ref{fig:intro_medinc} and \ref{fig:intro_houseage}). The linear relationship between the response and the explanatory variables differs across covariate-defined subgroups, and this structure cannot be represented by a covariate-independent prior.

On the theoretical side,
we provide theoretical support for the proposed framework. The iMoE class arises naturally as the posterior distribution of $\beta$ when $z$ is random (Proposition \ref{prop:covariate_representation}). Furthermore, the aggregated prior remains in the iMoE class, and its gates can approximate any weight function that varies continuously in $z$ (Theorem \ref{thm:uat_gate}). This shows that the iMoE class can approximate any discrete covariate-dependent prior with continuously varying weights.

\subsection{Literature review and contributions}
Our work relates to three strands of literature.

\emph{NPMLE and $g$-modeling.} The NPMLE goes back to \cite{Kiefer1956-dj} and \cite{Laird1978-rf}, with the geometry of the mixture likelihood developed by \cite{Lindsay1995-mi}. Modern computational and theoretical advances include \cite{Koenker2014-og}, \cite{Saha2020-ms}, and \cite{Soloff2024-ep}. The relative merits of $g$-modeling over $f$-modeling are discussed in \cite{Shen2022-db}. For regression with unit-level coefficients, nonparametric estimation in random-coefficient models has been studied by \cite{Dunker2025-uq}, \cite{Welbaum2025-tz}, and \cite{Jiang2025-am}. In particular, \cite{Jiang2025-am} established the NPMLE theory based on \cite{Lindsay1995-mi} and the exemplar method for the Gaussian unit-level linear model. In all of these works, the prior probability measure is common to all units, whereas we let the entire prior probability measure depend on covariates.

\emph{Covariates and side information in empirical Bayes.} Several recent works incorporate covariates into EB analysis. \cite{Ignatiadis2019-pz} considered the setting where the mean of the prior Gaussian distribution is a function of covariates. \cite{Scott2015-ge} proposed false discovery rate regression in which the null probability of a two-group model is modulated by covariate information. Side-information-based $f$-modeling (\cite{Cohen2013-pi}; \cite{Zhao2022-ev}; \cite{Luo2023-co}) is also studied. In these approaches, the covariate parameterizes a latent location parameter or drives a probability weight for a two-component mixture. In contrast, our framework allows the covariate to drive the full nonparametric prior probability measure.

\emph{Mixture of Experts.} The MoE model was introduced by \cite{Jacobs1991-zl}; refer to \cite{Yuksel2012-jr} an \cite{Masoudnia2014-it} as surveys. Recently, the MoE structure has been a focus of the machine learning community, especially for efficiently learning large language models \citep{shazeer2017outrageously,fedus2022switch,dai2024deepseekmoe}. From the perspective of the statistical model, a large body of work establishes approximation and estimation theory at the level of the conditional mean or conditional density (\cite{Jiang1999-cn}; \cite{Jiang1999-tc}; \cite{Norets2014-gk}; \cite{Nguyen2017-sq}; \cite{Nguyen2021-ev}; \cite{Nguyen2023-hy}), and recent studies address the selection of the number of experts and post-hoc merging for softmax-gated MoE (\cite{Hai2025-jk}; \cite{Tran2026-lb}).  Our results differ in the estimation target: our target is the conditional prior measure $G^{\ast}(\beta|z)$ itself. \cite{Oh2026-gz} also considered NPMLE for MoE, but they focus on nonparametric estimation of the noise distribution rather than the prior probability measure.

\subsection{Organization}

The rest of the paper is organized as follows. 
Section \ref{sec:npmle} formulates the covariate-dependent prior model and its $g$-modeling. The ill-posedness of the unrestricted NPMLE is discussed, and the iMoE class of conditional priors is introduced.
 Section \ref{Sec:estimation_procedure} presents estimation procedures. It describes the generation of exemplars, and exemplar-informed estimation algorithms, and the aggregation procedure for interpretable priors.
  Section \ref{sec:Theory} gives the generative interpretation of the iMoE class, its Lipschitz continuity in the covariate under a bounded set of gate parameters, and the universal approximation theorem for the aggregated softmax gates. Section \ref{sec:Simulation} reports the simulation studies, Section \ref{sec:real_data} the applications to the California housing data, the NHANES depression scores, and Section \ref{sec:discussion} concludes with a discussion of future work.

\section{Covariate-dependent NPMLE}
\label{sec:npmle}

In this paper, we work with
Model \eqref{model: unit_level_GLM_proposed},
where an exponential family distribution $ \mathrm {EF} (\mu, \phi)$ has a probability density function defined by
\begin{equation}
    \psi_{
    }(y;\theta,\phi) \coloneq   \exp\left(\frac{y\theta-b(\theta)}{a(\phi)}+c(y,\phi)\right).
\end{equation}
Here $\theta\in \Theta$ and $\phi\in \Phi$ are a canonical parameter and a dispersion parameter, respectively. $a(\cdot),b(\cdot)$ and $c(\cdot,\cdot)$ are known functions. The canonical parameter of unit $i$ is
determined by the linear predictor through
$\theta_i \coloneq \theta(x_i^{\top}\beta_i)$, where
$\theta(\cdot) := (b')^{-1}\!\left(g^{-1}(\cdot)\right)$.
Accordingly, we write the density of $y$ given $x$ and $\beta$ as
$\psi\!\left(y; \theta(x^{\top}\beta), \phi\right)$.
Throughout the paper, we assume that $\psi(y;\theta,\phi)>0$ for every $\theta\in\Theta$ and every $y$ in the support of the distribution, and that the map $\theta(\cdot)=(b')^{-1}(g^{-1}(\cdot))$ is continuous. Consequently, $\psi(y;\theta(x^{\top}\beta),\phi)$ is continuous and strictly positive in $\beta$.
These conditions are satisfied by the standard exponential families, such as the Gaussian, Bernoulli, Poisson, and gamma distributions, with their canonical or other commonly used link functions.

\subsection{Infinite Mixture-of-Expertizing (iMoE) Class}

We start with discussing the ill-posedness of a naive NPMLE of the conditional probability measure for the model \eqref{model: unit_level_GLM_proposed}.
The naive NPMLE is given
by
\begin{equation}
    \label{def:conditional_NPMLE}
    \begin{split}
        &\hat{G}(\cdot|z) \in \argmax_{G_z\in\mathcal{G}_{z}} \sum_{i=1}^n\log{f_{G_{z}}}(y_i|x_i,z_i), \\
        &\text{with}\quad
        f_{G_z}(y|x,z) \coloneq  \int\psi(y;\theta(x^{\top}\beta),\phi)
        \dd G(\beta|z),
    \end{split}
\end{equation}
where $\mathcal{G}_{z}$ is the class of all conditional distributions. Hereafter we assume that the constant part of $z_i$ is excluded.

The naive NPMLE $\hat{G}(\cdot|z)$ \eqref{def:conditional_NPMLE} with distinct $z_{i}$s is ill-posed, because we have infinitely many degenerate trivial solutions given as
\begin{equation}
    \label{def:ill_posed_solution}
    \hat{G}_{\mathrm{ill}}(\cdot|z) \coloneq
    \begin{cases}
        \delta_{\beta^\circ}\ \text{with any}\ \beta^\circ \in\mathcal{B}_i  & \text{for}\ z= z_i,\ i=1,\ldots,n, \\
        \text{arbitrary probability measure of $\beta$}                                                                                              & \text{for}\ z\neq z_i,\ i=1,\ldots,n,
    \end{cases}
\end{equation}
where 
 \begin{equation}
    \mathcal{B}_i=
\argmax_{\beta\in\mathbb{R}^{p}}\psi\!\left(y_i;\theta(x_i^{\top}\beta),\phi\right).
\end{equation}
 Indeed, for any conditional prior $G(\cdot|z_i)$ we have \[
 f_{G_{z}}(y_i|x_i,z_i)=\int\psi(y_i;\theta(x_i^{\top}\beta),\phi)\dd G(\beta|z_i)\le \max_{\beta}\psi(y_i;\theta(x_i^{\top}\beta),\phi),
 \]
 and the degenerate prior $\delta_{\beta^\circ}$ with $\beta^\circ \in\mathcal{B}_i$ attains this upper bound for each unit separately.

To address this degeneracy, we employ the following \emph{infinite Mixture-of-Expertizing} (iMoE) class of conditional probability measures.
Let $B\subset\mathbb{R}^{p}$ and $\Gamma\subset\mathbb{R}^{d}$ be compact sets, and $\mathcal{P}(B\times\Gamma)$ denote the set of probability measures whose support is contained in $B\times\Gamma$. For a mixing probability measure $H\in\mathcal{P}(B\times\Gamma)$ and a fixed $z\in\mathbb{R}^{d}$, we define the conditional probability measure $G_H(\cdot|z)$ on $B$ by
\begin{equation}
    \label{def:infinite_soft_max_prior}
    \dd G_H(\beta|z)\coloneq\frac{\int_{\Gamma} e^{\gamma^{\top}z}\dd H(\beta,\gamma)}{\int_{B\times\Gamma} e^{\gamma^{\top}z}\dd H(\tilde{\beta},\gamma)}.
\end{equation}
As the collection of these measures over the mixing measure $H$, we define the iMoE class of conditional priors by
\begin{equation}
    \label{def:infinite_soft_max_class}
    \mathcal{G}_{\mathrm{iMoE}}\coloneq\left\{z\mapsto G_H(\cdot|z) : H\in\mathcal{P}(B\times\Gamma)\right\}\subset\mathcal{G}_{z}.
\end{equation}
When $H=\sum_{j=1}^{J}\alpha_j\delta_{(\beta_j,\gamma_j)}$ with $\alpha_j>0$ and $\sum_{j=1}^{J}\alpha_j=1$, the conditional prior \eqref{def:infinite_soft_max_prior} reduces to a discrete prior with a softmax weight
\begin{equation}
    \label{def:discrete_soft_max_prior}
    G_H(\cdot|z)=\sum_{j=1}^{J}\frac{\alpha_j e^{\gamma_j^{\top}z}}{\sum_{k=1}^{J}\alpha_k e^{\gamma_k^{\top}z}}\,\delta_{\beta_j},
\end{equation}
where the support points of $G_{H}$ are shared across $z$, and the input for the softmax weight is modeled by $\gamma_{j}^\top z + \log \alpha_{j}$.
The class $\mathcal{G}_{\mathrm{iMoE}}$ extends this softmax-weighted structure to a general mixing measure $H$ that is not necessarily finite.

We define the iMoE-NPMLE as the maximizer of the conditional marginal likelihood over $\mathcal{G}_{\mathrm{iMoE}}$
\begin{equation}
    \label{def:proposed_NPMLE}
    \begin{split}
    \sum_{i=1}^n\log{f_{\hat{G}(\cdot|z)}}(y_i|x_i,z_i)
          = \max_{G_{H}\in\mathcal{G}_{\mathrm{iMoE}}}\sum_{i=1}^n\log{f_{G_{H}}}(y_i|x_i,z_i) 
          = \max_{H\in \mathcal{P}(B\times \Gamma)}\sum_{i=1}^{n}\log{f_{G_{H}}}(y_i|x_i,z_i),
    \end{split}
\end{equation}
where 
we use the fact that
every conditional prior in $\mathcal{G}_{\mathrm{iMoE}}$ is indexed by its mixing measure $H$.

\subsection{Existence and representation}
\label{sec:Existence}

We next discuss the existence and the representation of iMoE-NPMLE.
Let 
\[u_i(\gamma)= e^{\gamma^{\top}z_i}\quad \text{and}\quad v_i(\beta,\gamma)=e^{\gamma^{\top}z_i}\psi(y_i;\theta(x_i^{\top}\beta),\phi),
\quad i=1,\ldots,n,
\]
respectively.
Let us define the kernel and the normalization constant of the unit-level marginal likelihoods as functions of $H$ by
\begin{equation}
    \label{def:Lambda_Xi_tau}
 \Xi_{i}(H)\coloneq\int_{B\times\Gamma} v_{i}(\beta,\gamma)\dd H(\beta,\gamma),\quad   \Lambda_{i}{(H)}\coloneq\int_{B\times \Gamma} u_{i}(\gamma) \dd H(\beta,\gamma), 
\end{equation}
respectively.
Let
\begin{equation}
    \begin{split}
    \kappa(\beta,\gamma)&\coloneq(v_1(\beta,\gamma),\ldots,v_n(\beta,\gamma),\,
    u_1(\gamma),\ldots,u_n(\gamma))^{\top}\in\mathbb{R}^{2n},\\
    \zeta(H)&\coloneq (\Xi_{1}(H),\ldots,\Xi_{n}(H),\,
    \Lambda_1(H),\ldots,\Lambda_n(H))^{\top}\in\mathbb{R}^{2n},\\
    \ell(H)&\coloneq \sum_{i=1}^{n}\log f_{G_H}(y_i|x_i,z_i),
    \end{split}
\end{equation}
respectively.

We derive the discrete representation of the mixing measure that maximizes the marginal likelihood
\[
\hat{H}=\sum_{j=1}^{\hat{J}}
\hat{\alpha}_j\delta_{(\hat{\beta}_j,
\hat{\gamma}_j)},
\]
where $\hat{J}$ is some integer less than $2n$,
$\hat{\alpha} \in \Delta^{\hat{J}-1}\coloneq
\{\alpha\in\mathbb{R}^{\hat{J}}:\alpha_j\geq 0,\ \sum_{j=1,\ldots,\hat{J}}\alpha_j=1
\}$,
$(\hat{\beta}_{j},
\hat{\gamma}_{j})\in B\times \Gamma$ ($j=1,\ldots,\hat{J}$).
Observe that the marginal log-likelihood $\ell(H)$ has the representation
\begin{equation}
    \begin{split}
            \ell(H)&=\sum_{i=1}^{n}(\log\Xi_{i}(H)-\log\Lambda_{i}(H)) = \sum_{i=1}^{n}(\log \zeta_i(H)-\log\zeta_{i+n}(H)).          
    \end{split}
\end{equation}
Thus, iMoE-NPMLE \eqref{def:proposed_NPMLE} is equivalent to the maximum of the log-likelihood with respect to the vector $\zeta(H)$:
\begin{equation}
    \label{def:proposed_NPMLE_H}
\sum_{i=1}^{n}(\log \hat{\zeta}_i-\log \hat{\zeta}_{i+n})
\coloneq \max_{
\zeta \in \{\zeta=\zeta(H)\,:\, H\in \mathcal{P}(B\times\Gamma)\}} \sum_{i=1}^{n}(\log \zeta_i-\log\zeta_{i+n}).
\end{equation}
Since the vector $\zeta(H)$ is a convex combination of $\kappa(\beta,\gamma)$, 
\begin{equation}
    \label{eq:tau_convex_combination}
    \zeta(H)=\int_{B\times\Gamma}\kappa(\beta,\gamma)\dd H(\beta,\gamma),\quad H\in\mathcal{P}(B\times\Gamma),
\end{equation}
the solution of optimization problem \eqref{def:proposed_NPMLE_H} becomes
\begin{equation}
    \label{def:proposed_NPMLE_H_convex}
    \hat{\zeta}\coloneq \argmax_{\zeta \in \mathrm{conv}\{\kappa(\beta,\gamma)\;:\;(\beta,\gamma)\in B\times\Gamma\}} \sum_{i=1}^{n}(\log \zeta_i-\log\zeta_{i+n}).
\end{equation}
Here the solution $\hat{\zeta}$ is on the boundary of the convex hull of $\kappa(\beta,\gamma)$, because the objective function is increasing with respect to $\zeta_i$ and decreasing with respect to $\zeta_{i+n}$. 
Thus, from the Carathéodory's theorem, the solution $\hat{\zeta}$ is represented as a convex combination of at most $2n$ points in $\{\kappa(\beta,\gamma)\,:\,(\beta,\gamma)\in B\times\Gamma\}$:
\begin{equation}
    \hat{\zeta}=
    \sum_{j=1}^{\hat{J}}\hat{\alpha}_j \kappa(\hat{\beta}_j,
    \hat{\gamma}_j)
    ,\quad 
    \hat{\alpha}\in \Delta^{\hat{J}-1}
    ,\quad 
    (\hat{\beta}_j,
    \hat{\gamma}_j)
    \in B\times\Gamma,\quad \hat{J}\leq 2n.
\end{equation}
From \eqref{eq:tau_convex_combination}, the corresponding mixing measure $\hat{H}\in\mathcal{P}(B\times\Gamma)$ has the discrete representation as
\begin{equation}
    \label{def:solution_of_NPMLE_discrete_mixing_measure}
    \hat{H}=\sum_{j=1}^{\hat{J}}
    \hat{\alpha}_j
    \delta_{(\hat{\beta}_j
    ,
    \hat{\gamma}_j)}.
\end{equation}
Thus, iMoE-NPMLE \eqref{def:proposed_NPMLE} has a discrete form \eqref{def:discrete_soft_max_prior} with at most $2n$ support points.

On the basis of the above discrete representation,
we derive the softmax gating MoE representation of the marginal conditional density.
To do this, 
we prepare a softmax gate.
\begin{definition}[Softmax gate]
\label{def:extended_softmax_gate}
Let $J$ be an integer.
Let $\alpha=(\alpha_1,\ldots,\alpha_J)^{\top}\in\Delta^{J-1}\coloneq\{\alpha\in\mathbb{R}^{J}:\alpha_j\geq 0,\ \sum_{j=1}^{J}\alpha_j=1\}$ and $\gamma=(\gamma_1,\ldots,\gamma_J)$ with $\gamma_j\in\Gamma$. 
The \emph{softmax gate} is the vector-valued function $g(z;\alpha,\gamma)\coloneq(g_1(z;\alpha,\gamma),\ldots,g_J(z;\alpha,\gamma))^{\top}$ with
\begin{equation}
    \label{def:extended_softmax_gate_eq}
    g_j(z;\alpha,\gamma)\coloneq\frac{\alpha_j e^{\gamma_j^{\top}z}}{\sum_{k=1}^{J}\alpha_k e^{\gamma_k^{\top}z}},\qquad j=1,\ldots,J.
\end{equation}
\end{definition}
With this notation, the conditional prior produced by the mixing measure $\hat{H}$ is written as 
\[\hat{G}(\cdot|z)=\sum_{j=1}^{\hat{J}}g_j(z;\hat{\alpha},\hat{\gamma})\,\delta_{\hat{\beta}_j}.\]
The corresponding marginal conditional density is
\begin{equation}
    \label{def:proposed_NPMLE_discrete}
    f_{\hat{G}_{z}}(y|x,z) = \sum_{j=1}^{\hat{J}}g_j(z;\hat{\alpha},\hat{\gamma})
    \,
    \psi(y;\theta(x^{\top}\hat{\beta}_j),\phi)=\sum_{j=1}^{\hat{J}}\frac{\hat{\alpha}_j e^{z^{\top}\hat{\gamma}_j}}{\sum_{k=1}^{\hat{J}}\hat{\alpha}_k e^{z^{\top}\hat{\gamma}_k}}\,\psi(y;\theta(x^{\top}\hat{\beta}_j),\phi),
\end{equation}
and is exactly the softmax gating MoE model with the expert densities $\psi(y;\theta(x^{\top}\beta_j),\phi)$. 
Remark that 
$G_{H}(\cdot\mid z)$ in the above parameterization can be transformed into the standard affine parameterization:
\begin{equation}
    \label{def:discrete_soft_max_prior_intercept}
        G_H(\cdot|z)=\sum_{j=1}^{J}\frac{e^{\gamma_j^{\top}z+ \gamma_{0j}}}{\sum_{k=1}^{J} e^{{\gamma_k^{\top}z}+\gamma_{0k}}}\,\delta_{\beta_j}, \quad \text{with}\;\gamma_{0j}\coloneq \log \frac{\alpha_j}{\alpha^{\dagger}}\in\mathbb{R},\qquad0<\alpha^{\dagger}\leq1.
\end{equation}

\subsection{The first-order optimality condition}
\label{sec:First order optimality}

We proceed to the discussion on the first-order optimality condition of iMoE-NPMLE \eqref{def:proposed_NPMLE}
through the directional derivative with respect to $H$.
Let $H_{\epsilon}\coloneq (1-\epsilon)H+\epsilon\delta_{(\beta,\gamma)}$. The directional derivative of the marginal log-likelihood $\ell(H)$ with respect to $H$ along with point mass $\delta_{(\beta,\gamma)}$ is given by 
\begin{equation}
    \label{cal:directional_derivative}
    \begin{split}
    D(H;(\beta,\gamma)) & \coloneq 
    \frac{\partial \ell(H_{\epsilon})}{\partial \epsilon}\Bigg{|}_{\epsilon=+0}
    = \sum_{i=1}^{n}\left(\frac{v_i(\beta,\gamma)}{\Xi_i(H)}-\frac{u_i(\gamma)}{\Lambda_i(H)}\right). \\ 
    \end{split}
\end{equation}

The following theorem states the first-order optimality condition with respect to $\hat{G}(\cdot|z)$.

\begin{theorem} [The first-order optimality condition of iMoE-NPMLE \eqref{def:proposed_NPMLE}]
    \label{thm:necessary condition}
Let $\hat{H}$ be a mixing measure giving the iMoE-NPMLE
\eqref{def:proposed_NPMLE}. Then, for all $(\beta,\gamma)\in B\times\Gamma$, we have
\begin{equation}
    \label{eq:necessary_condition}
    \begin{split}
    &(i) \quad \sup_{(\beta,\gamma)\in B\times\Gamma}D(\hat{H};(\beta,\gamma))\leq 0, \\
    &(ii) \quad D(\hat{H};(\beta,\gamma)) = 0 \quad \text{for all}\ (\beta,\gamma)\in\mathrm{supp}(\hat{H}).
    \end{split}
\end{equation}
\end{theorem}
The proof is given in Appendix \ref{app:proofs}.
The conditions in the theorem correspond to a covariate-dependent prior version of the sufficient and necessary conditions of the NPMLE in \cite{Jiang2025-am}, which is stated for the Gaussian unit-level linear model with a non-conditional prior. Remark that our result is stated only for the necessary because of the normalizing constant $\Lambda_i(H)$. In the classical NPMLE, which corresponds to the restriction $\gamma=0$, we have $u_i(\gamma)/\Lambda_i(H)=1$, and integrating the directional derivative against any mixing measure $H'$ gives \[\int D(\hat{H};(\beta,\gamma))\dd H'(\beta,\gamma)=\sum_{i=1}^{n}\{\Xi_i(H')/\Xi_i(\hat{H})-1\}.\] Then the inequality $\log x\le x-1$ turns the first-order conditions into 
\[\ell(H')-\ell(\hat{H})=\sum_{i=1}^{n}\log\{\Xi_i(H')/\Xi_i(\hat{H})\}\le0.\] Without the restriction, the integral becomes $\sum_{i=1}^{n}\{\Xi_i(H')/\Xi_i(\hat{H})-\Lambda_i(H')/\Lambda_i(\hat{H})\}$, so that the same argument does not apply.

The following proposition converts
Theorem \ref{thm:necessary condition} into an explicit form in the exponential family.
For $\beta\in B$, let $\mu_i(\beta)\coloneq g^{-1}(x_i^{\top}\beta)$ denote the conditional mean of unit $i$, and let $V(\mu)\coloneq b''\!\left((b')^{-1}(\mu)\right)$ denote the variance function of the exponential family.
\begin{proposition}[The first-order optimality condition in the exponential family]
    \label{prop:necessary_condition}
    Suppose that $b$ is twice differentiable with $b''>0$ on $\Theta$ and that $g$ is differentiable with $g'\neq 0$ on the mean space.
    Let $\hat{H}$ be a mixing measure giving the solution of the iMoE-NPMLE \eqref{def:proposed_NPMLE}.
    Then, for every $(\beta,\gamma)\in\mathrm{supp}(\hat{H})$ in the interior of $B\times\Gamma$, we have
    \begin{align}
        \nabla_{\beta}D(\hat{H};(\beta,\gamma))
        &=\sum_{i=1}^{n}\frac{v_i(\beta,\gamma)}{\Xi_i(\hat{H})}\,
          \frac{y_i-\mu_i(\beta)}{a(\phi)\,V(\mu_i(\beta))\,g'(\mu_i(\beta))}\,x_i
        =0,
        \label{eq:gradient_condition_beta}\\
        \nabla_{\gamma}D(\hat{H};(\beta,\gamma))
        &=\sum_{i=1}^{n}\left(\frac{v_i(\beta,\gamma)}{\Xi_i(\hat{H})}-\frac{u_i(\gamma)}{\Lambda_i(\hat{H})}\right)z_i
        =0.
        \label{eq:gradient_condition_gamma}
    \end{align}
    In particular, for the canonical link $g=(b')^{-1}$, the weight in \eqref{eq:gradient_condition_beta} reduces to $(y_i-\mu_i(\beta))/a(\phi)$.
\end{proposition}
The proof is given in Appendix \ref{app:proofs}.
The condition \eqref{eq:gradient_condition_beta} says that each support point $\beta$ solves a weighted GLM score equation. The weight of unit $i$ is $v_i(\beta,\gamma)/\Xi_i(\hat{H})$, which is the posterior probability that unit $i$ belongs to the support point $(\beta,\gamma)$. Remark that these weights depend on the unknown $\hat{H}$, so the equation is not a standard GLM score equation. 

The following proposition states the usefulness of the first-order optimality condition to design the estimation algorithm of the iMoE-NPMLE.
For a probability vector $w \in \Delta^{n-1}$, let
\begin{equation}
    \label{def:S_beta}
    S_{\beta}(w)\coloneq\left\{\beta\in B:\sum_{i=1}^{n}w_i\,(y_i-\mu_i(\beta))\,x_i=0\right\}
\end{equation}
be the set of solutions of the GLM score equation in which unit $i$ has weight $w_i$.

\begin{proposition}[sparse representation of support points]
    \label{prop:sparse_representation}
    Suppose that the assumptions of Proposition \ref{prop:necessary_condition} hold. Let $(\beta,\gamma)$ be a support point of $\hat{H}$ that lies in the interior of $B\times\Gamma$. Then, there is a probability vector $w\in\Delta^{n-1}$ with at most $p+1$ nonzero entries such that $\beta\in S_{\beta}(w)$. In words, $\beta$ is the solution of a GLM score equation that uses at most $p+1$ units.
\end{proposition}
The proof is given in Appendix \ref{app:proofs}.
Proposition \ref{prop:sparse_representation} ensures that if we know the non-zero entries and the corresponding weights, we do not need to use the whole dataset to find the support point. This part is directly extension of the result in \cite{Jiang2025-am} to the covariate-dependent prior setting and GLM.

In contrast, no such sparse representation holds for $\gamma$. In the condition \eqref{eq:gradient_condition_gamma}, the coefficient of $z_i$ is $e^{z_i^{\top}\gamma}\{\psi(y_i;\theta(x_i^{\top}\beta),\phi)/\Xi_i(\hat{H})-1/\Lambda_i(\hat{H})\}$, whose sign varies with $i$, so the argument based on Carath\'{e}odory's theorem does not apply. Consequently, the support point $\gamma$ cannot be determined from a small subset of units, and has to be estimated from all the units through \eqref{eq:gradient_condition_gamma}.

\section{Estimation procedure}
\label{Sec:estimation_procedure}

In this section, we propose estimation procedures for iMoE-NPMLE \eqref{def:proposed_NPMLE}. Since iMoE-NPMLE has non-convexity arising from the normalizing constant, our goal is to find a local maximizer of the marginal likelihood satisfying the first-order optimality condition shown in Theorem \ref{thm:necessary condition} and Proposition \ref{prop:necessary_condition}.

Our estimation target $\hat{H}$ is a discrete measure characterized by at most $2n$ support points $(\hat{\beta}_j,\hat{\gamma}_j)$ and the weights $\hat{\alpha}_j$ as in \eqref{def:solution_of_NPMLE_discrete_mixing_measure}. Thus, the estimation is achieved by fitting the MoE model whose number of experts $\hat{J}$ is potentially large and unknown. Although the EM algorithm is a popular estimation procedure for the MoE model, we need to sweep the number of experts $\hat{J}$ to find the optimal number of experts, which is computationally infeasible for large $n$.

To avoid the sweep for $\hat{J}$, we propose two exemplar-based algorithms for fitting the MoE, Free-EXMoE and Grid-EXMoE, which are collectively referred to as EXMoE (EXemplar-based algorithms for fitting the Mixture of Experts). 
Based on the sparse representation of the support points of $\beta$ in Proposition \ref{prop:sparse_representation}, both algorithms start from \emph{exemplars}, candidates for the support points computed from randomly selected small subsamples.
 Free-EXMoE uses the exemplars as the initial values and treats the support points as free parameters, updating the exemplars and the gate parameters together with the weights $\alpha$. Grid-EXMoE fixes the support points on the product of the exemplars and a grid of the gate parameters and updates only $\alpha$, as in the exemplar method \citep{Bohning2000-gq,Lashkari2007-qc,Soloff2024-ep,Jiang2025-am}.

These two algorithms differ mainly in the trade-off between the accuracy of $\hat{\beta}$ and the computational costs. Free-EXMoE is suitable in high dimensions, where the number of fixed exemplars needed to cover the true points is too large, and the grid of $\gamma$ grows exponentially with the dimension of $z$. Grid-EXMoE is computationally cheaper per candidate point, so that we can take many candidate points beyond the $2n$ bound.

\begin{figure}[htbp]
\centering
\begin{tikzpicture}[
  >=Stealth, font=\small,
  cand/.style={draw=teal!60, fill=teal!8, rounded corners=2pt, minimum width=1.8cm, minimum height=0.42cm, inner sep=1pt},
  pruned/.style={draw=gray!60, dashed, rounded corners=2pt, minimum width=1.8cm, minimum height=0.42cm, inner sep=1pt, text=gray},
  agg/.style={draw=violet!60, fill=violet!8, rounded corners=3pt, minimum width=2.3cm, minimum height=0.95cm, align=center},
  foot/.style={text width=3.7cm, align=center, font=\footnotesize},
]
  \node[font=\small\bfseries] at (0,0) {$M$ candidate experts};
  \node[font=\small\bfseries] at (4,0) {$\hat J$ active experts};
  \node[font=\small\bfseries] at (8,0) {$K$ aggregated experts};

  \foreach \k in {1,...,8} \node[cand] (c\k) at (0,-0.55*\k) {$E_{\k}$};

  \foreach \k in {1,3,4,6,8} \node[cand]   (a\k) at (4,-0.55*\k) {$E_{\k}$};
  \foreach \k in {2,5,7}     \node[pruned] (a\k) at (4,-0.55*\k) {$\alpha_{\k}=0$};
  \draw[->, thick] (1.2,-2.475) -- (2.8,-2.475);

  \node[agg] (g1) at (8,-1.10) {$(\bar\beta_1,g^{\dagger}_1)$\\[-1pt]{\scriptsize $E_1+E_3$}};
  \node[agg] (g2) at (8,-2.75) {$(\bar\beta_2,g^{\dagger}_2)$\\[-1pt]{\scriptsize $E_4+E_6$}};
  \node[agg] (g3) at (8,-4.40) {$(\bar\beta_3,g^{\dagger}_3)$\\[-1pt]{\scriptsize $E_8$}};
  \foreach \a/\g in {1/1,3/1,4/2,6/2,8/3} \draw[->] (a\a.east) -- (g\g.west);

  \node[foot] at (0,-5.7) {Algorithm~\ref{alg:exemplar}:\\ exemplars $\beta^{(j)}$ from local subsamples};
  \node[foot] at (4,-5.7) {Algorithm~\ref{alg:exmoe} or \ref{alg:fixed_dictionary}:\\ pairwise Frank--Wolfe sets $\alpha_j=0$};
  \node[foot] at (8,-5.7) {Section~\ref{sec:experts_aggregation}:\\ merge experts with similar $\hat\beta_j$, sum gates};
\end{tikzpicture}
\caption{Schematic of the estimation procedure in terms of experts. Exemplars provide $M$ candidate experts; Free-EXMoE or Grid-EXMoE removes redundant experts by setting their weights exactly to zero, without specifying the number of experts; the aggregation merges experts with similar regression coefficients into $K$ interpretable experts with aggregated gates.}
\label{fig:schematic_pipeline}
\end{figure}
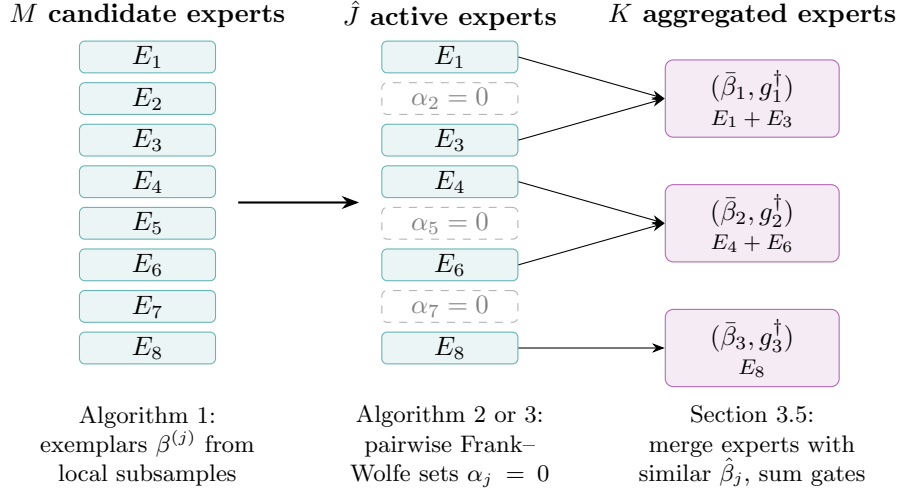

The rest of this section follows the steps of the procedure described in Figure \ref{fig:schematic_pipeline}. Both algorithms first need a collection of exemplars, and Section \ref{subsec:anchored_exemplars} describes how they are generated from small subsamples. Starting from these exemplars, Free-EXMoE updates the weights, the gate parameters, and the regression coefficients by block coordinate descent (Section \ref{subsec:exmoe}), whereas Grid-EXMoE fixes the support points on the exemplars and a grid of the gate parameters, and estimates only the weights (Section \ref{subsec:fixed_dictionary}). Since both algorithms return a potentially local solution, Section \ref{subsec:posterior_check} gives a diagnostic that checks the first-order optimality condition for the output. Finally, the estimated measure typically has many support points, and Section \ref{sec:experts_aggregation} gives a post hoc procedure that aggregates them into a few interpretable clusters.

\subsection{Generation of exemplars}
\label{subsec:anchored_exemplars}

We first discuss how we generate exemplars for the candidates of the support points of $\beta$.
Motivated by Proposition \ref{prop:sparse_representation}, we generate exemplars from the solutions of GLM score equations on small subsets of units. Although the non-zero entries and the corresponding weights in Proposition \ref{prop:sparse_representation} are unknown, the solution of $p+1$ units with uniform weights is a good empirical approximation of the support point $\hat{\beta}_j$, as shown in \cite{Jiang2025-am} for classical NPMLE. However, our estimation target, the conditional covariate-dependent prior, changes the prior of the latent regression coefficients $\beta_i$ across covariates $z_i$, which means that the different units are possibly generated from distributions not sharing the same support points. Eventually, the noisy exemplars from different support points would be generated. Therefore, we need to generate exemplars from units that are more likely from the same support point. 

To incorporate this intuition, we consider drawing a subset around an \emph{anchor} unit. We pick a unit $i_0$ at random. Then, we collect the $P$ units closest to $i_0$ in the covariate space of $z$, and draw $m$ units from them, where $m\geq p+1$. The anchor makes the exemplars cover the local range of $z$ evenly, and the pool size $P$ and the subsample size $m$ control how local and how noisy each exemplar is. The procedure is summarized in Algorithm \ref{alg:exemplar}.

\begin{algorithm}[H]
    \caption{Anchored generation of exemplars}
    \label{alg:exemplar}
    \KwIn{Data $\{(x_i,y_i,z_i)\}_{i=1}^{n}$; support bound $B$; number of exemplars $M$; pool size $P\leq n$ and subsample size $m$ with $p+1\leq m\leq P$}
    \KwOut{Exemplars $\beta^{(1)},\ldots,\beta^{(M)}\in B$}
    \BlankLine
    \For{$j=1,\ldots,M$}{
        Draw an anchor unit $i_0$ uniformly from $\{1,\ldots,n\}$\;
        Let $\mathcal{P}_j\subset\{1,\ldots,n\}$ be the $P$ units closest to $i_0$ in $z$ (with $\mathcal{P}_j=\{1,\ldots,n\}$ when $P=n$)\;
        Draw a subset $S_j\subset\mathcal{P}_j$ with $|S_j|=m$ uniformly at random without replacement\;
        Solve the GLM score equation with uniform weights on $S_j$, $\sum_{i\in S_j}\bigl(y_i-\mu_i(\beta)\bigr)x_i=0$, and set the solution to $\beta^{(j)}$\;
        \If{$\beta^{(j)}\notin B$ or the score equation has no solution}{
            Discard $S_j$ and redraw\;
        }
    }
\end{algorithm}

The pool size $P$ and the subsample size $m$ in Algorithm \ref{alg:exemplar} control globality and locality of the generated exemplars. With $P=n$ and $m=p+1$, the anchor plays no role and $S_j$ is a uniformly random subset of $p+1$ units, so that $\beta^{(j)}$ is the exact solution of the score equation on $S_j$. This is the exemplar generation of \cite{Jiang2025-am}. With $P<n$ and $m>p+1$, each exemplar is a local GLM estimate around the anchor, which is less variable than the exact solution of $p+1$ units and remains close to the support points whose gates are concentrated near $z_{i_0}$. This strategy is suitable for the situation, where $p$ is larger and the responses are noisier such as in the real data analyses.

\subsection{Free-EXMoE: exemplars as initial values of free support points}
\label{subsec:exmoe}

We next present the first estimation algorithm for the mixing measure on the basis of examplers.
For the exemplars $\beta^{(1)},\ldots,\beta^{(M)}$ generated by Algorithm \ref{alg:exemplar}, the mixing measure is restricted to the form $H=\sum_{j=1}^{M}\alpha_j\delta_{(\beta_j,\gamma_j)}$ with $M$ support points, where the exemplars serve as the initial values $\beta_j=\beta^{(j)}$ of the regression coefficients and are updated during the estimation. Writing $\psi_{ij}(\beta_j)\coloneq\psi(y_i;\theta(x_i^{\top}\beta_j),\phi)$ and $\beta\coloneq(\beta_1,\ldots,\beta_M)$, we estimate $(\alpha,\gamma,\beta)$ by minimizing the penalized negative log-likelihood
\begin{equation}
    \label{def:exmoe_objective}
    \begin{split}
    &F_{\lambda}(\alpha,\gamma,\beta)\coloneq\mathcal{L}(\alpha,\gamma,\beta)+\frac{\lambda}{2}\sum_{j=1}^{M}\|\gamma_j\|_2^{2},
    \qquad \alpha\in\Delta^{M-1},\quad \beta_j\in B,
    \\
    &\quad\text{where}\quad
    \mathcal{L}(\alpha,\gamma,\beta)\coloneq-\ell(H)=\sum_{i=1}^{n}\left\{\log\sum_{j=1}^{M}\alpha_j e^{z_i^{\top}\gamma_j}-\log\sum_{j=1}^{M}\alpha_j e^{z_i^{\top}\gamma_j}\psi_{ij}(\beta_j)\right\}
    \end{split}
\end{equation}
with a regularization parameter $\lambda > 0$. The ridge term keeps the minimizer of $\gamma_j$ bounded, which replaces the explicit constraint $\gamma_j\in\Gamma$ in the implementation. 
Letting
\begin{equation}
    \label{def:exmoe_weights}
    \pi_{ij}\coloneq\frac{\alpha_j e^{z_i^{\top}\gamma_j}}{\sum_{k=1}^{M}\alpha_k e^{z_i^{\top}\gamma_k}},\qquad
    \text{and}\qquad
    w_{ij}\coloneq\frac{\alpha_j e^{z_i^{\top}\gamma_j}\psi_{ij}(\beta_j)}{\sum_{k=1}^{M}\alpha_k e^{z_i^{\top}\gamma_k}\psi_{ik}(\beta_k)},
\end{equation}
the gradients of $F_{\lambda}$ are given by
\begin{equation}
    \label{eq:exmoe_gradients}
    \begin{split}
    \frac{\partial F_{\lambda}}{\partial\alpha_j}&=-\frac{1}{\alpha_j}\sum_{i=1}^{n}(w_{ij}-\pi_{ij})=-D_j(\alpha,\gamma,\beta)\coloneq -D(H;(\beta_j,\gamma_j)),\\
    \nabla_{\gamma_j}F_{\lambda}&=-\sum_{i=1}^{n}(w_{ij}-\pi_{ij})\,z_i+\lambda\gamma_j,\\
    \nabla_{\beta_j}F_{\lambda}&=-\sum_{i=1}^{n}w_{ij}\,\nabla_{\beta_j}\log\psi\bigl(y_i;\theta(x_i^{\top}\beta_j),\phi\bigr)=-\alpha_j\nabla_{\beta_j}D(H;(\beta_j,\gamma_j)).
    \end{split}
\end{equation}

Then, given the exemplars, we minimize $F_{\lambda}$ by a block coordinate descent over $\alpha$, $\gamma$, and $\beta$ via the above gradient formula. The process is summarized in Algorithm \ref{alg:exmoe} and discussed in the following.
\begin{algorithm}[!htp]
    \caption{Free-EXMoE}
    \label{alg:exmoe}
    \KwModel{Data $\{(x_i,y_i,z_i)\}_{i=1}^{n}$; hyperparameters  $\phi$, $B$, and $\lambda$; exemplars $\beta^{(1)},\ldots,\beta^{(M)}$;}
    \KwTune{Update rule $\mathcal{U}_{\Delta}$ ;  $T_{\max}$,$\epsilon_{\mathrm{tol}}$, $R$, and $s_{\gamma}$;}
    \KwOut{$\hat{G}(\cdot|z)=\sum_{j\in\hat{\mathcal{S}}}\hat{g}_j(z)\delta_{\hat{\beta}_j}$ with $\hat{g}_j(z)=\hat{\alpha}_j e^{z^{\top}\hat{\gamma}_j}/\sum_{k\in\hat{\mathcal{S}}}\hat{\alpha}_k e^{z^{\top}\hat{\gamma}_k}$;}
    \BlankLine
    \tcp{Step 1: block coordinate descent for $(\alpha,\gamma,\beta)$ with $R$ restarts}
    \For{$r=1,\ldots,R$}{
        Initialize $\alpha=(1/M,\ldots,1/M)^{\top}$ and $\beta_j=\beta^{(j)}$ for all $j$\; 
        Set $\gamma_j=0$ for all $j$ if $r=1$, and draw $\gamma_j\sim N(0,s_{\gamma}^{2}I)$ independently otherwise\;         \For{$t=1,\ldots,T_{\max}$}{
            Compute $F_t\coloneq F_{\lambda}(\alpha,\gamma,\beta)$; \textbf{break} if  $F_{t-1}-F_t<\epsilon_{\mathrm{tol}}\,|F_t|$\;
            \tcp{$\alpha$-update: descent step on the simplex}$\alpha\leftarrow\mathcal{U}_{\Delta}(\alpha;D,\gamma,\beta)$, an update on $\Delta^{M-1}$ \;
            \tcp{$\gamma$-update: minimization over the active support points}
            Set $\mathcal{A}\leftarrow\{j:\alpha_j>0\}$\;
                $\gamma_{\mathcal{A}}=(\gamma_j)_{j\in\mathcal{A}}\leftarrow\argmin_{\gamma_{\mathcal{A}}}F_{\lambda}(\alpha,\gamma,\beta)$ by L-BFGS-B with the gradient in \eqref{eq:exmoe_gradients}\;
            Set $\gamma_j\leftarrow 0$ for $j\notin\mathcal{A}$\;
            \tcp{$\beta$-update: weighted GLM fits on the active support points}
            \For{$j\in\mathcal{A}$}{
                Let $\tilde{\beta}_j\leftarrow\argmin_{b}Q_j(b)$ in \eqref{eq:exmoe_beta_majorization}, computed by IRLS started at $\beta_j$ \;
                \If{$\tilde{\beta}_j\in B$ }{
                    $\beta_j\leftarrow\tilde{\beta}_j$\;
                }
            }
        }
        Set $\alpha\leftarrow\mathcal{U}_{\Delta}(\alpha;D,\gamma,\beta)$\;
        Record $(\alpha^{[r]},\gamma^{[r]},\beta^{[r]})\leftarrow(\alpha,\gamma,\beta)$, $\mathcal{L}^{[r]}\leftarrow\mathcal{L}(\alpha,\gamma,\beta)$, and $\hat{\mathcal{S}}^{[r]}\leftarrow\{j:\alpha_j>0\}$\;
    }
    \BlankLine
    \tcp{Step 2: selection of the restart}
    Set $\hat{r}\leftarrow\argmin_{r}\mathcal{L}^{[r]}$, $\hat{\mathcal{S}}\leftarrow\hat{\mathcal{S}}^{[\hat{r}]}$, $\hat{\alpha}_j\leftarrow\alpha_j^{[\hat{r}]}/\sum_{k\in\hat{\mathcal{S}}}\alpha_k^{[\hat{r}]}$, $\hat{\gamma}_j\leftarrow\gamma_j^{[\hat{r}]}$, and $\hat{\beta}_j\leftarrow\beta_j^{[\hat{r}]}$ for $j\in\hat{\mathcal{S}}$\;
\end{algorithm}

\textbf{Update of $\alpha$ given $(\gamma,\beta)$}.
This update is an update rule $\mathcal{U}_{\Delta}$ on the \textit{simplex} $\Delta^{M-1}$. For fixed $(\gamma,\beta)$, the negative log-likelihood $\mathcal{L}(\cdot,\gamma,\beta)$ is a difference of two convex functions of $\alpha$, and the first-order conditions on the simplex are $D_j=0$ for $\alpha_j>0$ and $D_j\leq 0$ for $\alpha_j=0$, which is the candidate-points version of Theorem \ref{thm:necessary condition}. Our default update rule is the pairwise Frank--Wolfe method \citep{Lacoste-Julien2015-fw} with reparametrization (see Appendix \ref{sec:PFW_repara}), which sets redundant $\alpha_j$ exactly to zero. As a result, the support is selected without specifying the number of experts. Alternatives include the projected gradient method onto the simplex \citep{Duchi2008-aa} and the other variants of the Frank--Wolfe method summarized in \cite{Lacoste-Julien2015-fw}.

\textbf{Update of $\gamma$ given $(\alpha,\beta)$}. 
For this update,
we minimize $F_{\lambda}$ with respect to the gate parameters $\gamma_{\mathcal{A}}=(\gamma_j)_{j\in\mathcal{A}}$ of the active support points $\mathcal{A}=\{j:\alpha_j>0\}$, holding $\alpha$, $\beta$, and the gates outside $\mathcal{A}$ fixed. This inner problem is solved by the L-BFGS-B \citep{Byrd1995-bj} method applied to $\gamma_{\mathcal{A}}$.

\textbf{Update of $\beta$ given $(\gamma,\alpha)$}. 
For this update, we use the majorization function. With the weights $w_{ij}$ fixed at the current $\beta=(\beta_1,\ldots,\beta_M)$, 
the definition \eqref{def:exmoe_weights}
and
Jensen's inequality give, for any $\beta'=(\beta'_1,\ldots,\beta'_M)$,
\begin{equation}
    \label{eq:exmoe_beta_majorization}
    \mathcal{L}(\alpha,\gamma,\beta')-\mathcal{L}(\alpha,\gamma,\beta)\leq\sum_{j=1}^{M}\bigl\{Q_j(\beta'_j)-Q_j(\beta_j)\bigr\},\qquad
    Q_j(b)\coloneq-\sum_{i=1}^{n}w_{ij}\log\psi\bigl(y_i;\theta(x_i^{\top}b),\phi\bigr),
\end{equation}
so that any $\beta'$ that decreases the separable surrogates $Q_j$ decreases $\mathcal{L}$. The minimization of $Q_j$ is a weighted GLM fit with the weights $w_{1j},\ldots,w_{nj}$, and its minimizer satisfies the weighted score equation 
\[\sum_{i=1}^{n}w_{ij}\,\nabla_{\beta}\log\psi(y_i;\theta(x_i^{\top}b),\phi)=0
\]
with the weights $w_{ij}$ in place of $v_i(\beta,\gamma)/\Xi_i(\hat{H})$. 
We solve it for each active support point $j\in\mathcal{A}$ with iteratively reweighted least squares (IRLS) started at the current $\beta_j$. We accept the update only if the minimizer lies in $B$.

Since the $\alpha$-update does not increase $\mathcal{L}$, the $\gamma$-update is accepted only if it does not increase $F_{\lambda}$, and the $\beta$-update does not increase $F_{\lambda}$, the objective $F_{\lambda}$ is non-increasing along the outer iterations. A final $\alpha$-update is taken so that the returned weights satisfy the first-order conditions.

Because $F_{\lambda}$ is not convex in $(\alpha,\gamma,\beta)$, Algorithm \ref{alg:exmoe} converges to a stationary point that depends on the initial value. We therefore run Step 1 of Algorithm \ref{alg:exmoe} from $R$ initial values, the first with $\gamma=0$ and the others with $\gamma_j$ drawn from $N(0,s_{\gamma}^{2}I)$, and select the run with the smallest negative log-likelihood $\mathcal{L}$.  The dispersion parameter $\phi$ is not estimated inside the algorithm and is chosen by $K$-fold cross-validation evaluating the log-likelihood.

Our default values of tuning parameters are summarized in Appendix \ref{app:tunings}.

\subsection{Grid-EXMoE: support points on the exemplars and a grid}
\label{subsec:fixed_dictionary}

When the dimension of $\gamma$ is relatively small,
we can also employ the following estimation procedure for the weights $\alpha$ 
with the gate parameters $\gamma$ fixed.
Let $\beta^{(1)},\ldots,\beta^{(M_{\beta})}$ be the exemplars generated by Algorithm \ref{alg:exemplar}, and let $\mathcal{G}_h=\{\gamma^{(1)},\ldots,\gamma^{(d_{\mathrm{grid}})}\}\subset\Gamma$ be a grid of spacing $h$ on the compact set $\Gamma$. We define the candidate dictionary given by the product set
\begin{equation}
    \label{def:dictionary}
    \mathcal{D}_M \coloneq\bigl\{(\beta^{(m)},\gamma^{(l)}):\ m=1,\ldots,M_{\beta},\ l=1,\ldots,d_{\mathrm{grid}}\bigr\},\qquad M=M_{\beta} d_{\mathrm{grid}},
\end{equation}
whose $j$-th point is $(\beta^{(m_j)},\gamma_j)$ with $\gamma_j=\gamma^{(l_j)}$,
and the mixing measure is restricted to $H=\sum_{j=1}^{M}\alpha_j\delta_{(\beta^{(m_j)},\gamma_j)}$ with $\alpha\in\Delta^{M-1}$.
With $(\beta^{(m_j)},\gamma_j)$ fixed, the mixing measure $H$ is determined by $\alpha$, so that the negative log-likelihood $\mathcal{L}(\alpha)\coloneq-\ell(H)=\sum_{i=1}^{n}\{\log \Lambda_i(H)-\log \Xi_i(H)\}$, with $\Xi_i$ and $\Lambda_i$ defined in \eqref{def:Lambda_Xi_tau}, is a function of $\alpha$ alone. The estimation is its minimization over the simplex. 

The weights are estimated by a two-stage procedure: the first stage is the MM update, which is used to accelerate the decrease in negative log-likelihoods for initial iterations. The second stage is the active-set refinement, which combines the Pairwise Frank-Wolfe method with reparametrization to induce sparsity and a step of adding candidate points. These are summarized in Algorithm \ref{alg:fixed_dictionary}, and those details, including the definition of some symbols, are described in Appendix \ref{app:alpha_update}.
Since $\mathcal{L}$ is not convex, we run the procedure from several initial values with same manner as in Algorithm \ref{alg:exmoe}.

Our default values of tuning parameters are summarized in Appendix \ref{app:tunings}.

\begin{algorithm}[h]
    \caption{Grid-EXMoE}
    \label{alg:fixed_dictionary}
    \KwModel{Data $\{(x_i,y_i,z_i)\}_{i=1}^{n}$; dispersion $\phi$; exemplars $\beta^{(1)},\ldots,\beta^{(M_{\beta})}$; grid $\mathcal{G}_h\subset\Gamma$}
    \KwTune{$K_1$, $\epsilon_0$, $\epsilon_{\mathrm{tol}}$, and $R$}
    \KwOut{ $\hat{G}(\cdot|z)=\sum_{j\in\hat{\mathcal{S}}}\hat{g}_j(z)\delta_{\beta^{(m_j)}}$ with $\hat{g}_j(z)=\hat{\alpha}_j e^{z^{\top}\gamma_j}/\sum_{k\in\hat{\mathcal{S}}}\hat{\alpha}_k e^{z^{\top}\gamma_k}$}
    \BlankLine
    Form the dictionary $\mathcal{D}_M$ in \eqref{def:dictionary}\;
    \BlankLine
    \For{$r=1,\ldots,R$}{
        Initialize $\alpha=(1/M,\ldots,1/M)^{\top}$ if $r=1$, and draw $\alpha$ at random otherwise\;
        \tcp{Step 1: warm start by the MM update}
        \For{$k=1,\ldots,K_1$}{
            $\alpha_j\leftarrow\alpha_j P_j(\alpha)/S_j(\alpha)$ for $j=1,\ldots,M$ as in \eqref{eq:app_mm_update}\;
        }
        Set $\mathcal{S}\leftarrow\{j:\tilde{\alpha}_j\geq\epsilon_0\}$ with the average gate weights $\tilde{\alpha}_j=s_j\alpha_j$, $s_j=S_j(\alpha)/n$, of Appendix \ref{sec:PFW_repara}, and $\alpha_j\leftarrow 0$ for $j\notin\mathcal{S}$\;
        \tcp{Step 2: active-set refinement}
        \Repeat{$\max_{j}\tilde{D}_j(\alpha)\leq\epsilon_{\mathrm{tol}}$}{
           Updata $\left\{\alpha_{j}:j\in \mathcal{S}\right\}$ by the pairwise Frank--Wolfe method with the reparametrization of Appendix \ref{sec:PFW_repara} until $\max_{j\in\mathcal{S}}|\tilde{D}_j(\alpha)|\leq\epsilon_{\mathrm{tol}}$, and set $\mathcal{S}\leftarrow\{j:\alpha_j>0\}$\;
            Compute $\tilde{D}_j(\alpha)$ for all $j=1,\ldots,M$ and set $\mathcal{T}\leftarrow\{j\notin\mathcal{S}:\tilde{D}_j(\alpha)>\epsilon_{\mathrm{tol}}\}$\;
            \If{$\mathcal{T}\neq\emptyset$}{
                Set $\tilde{\alpha}_j\leftarrow\epsilon_0$ for $j\in\mathcal{T}$, normalize $\tilde{\alpha}$, return to $\alpha_j=\tilde{\alpha}_j/s_j$, and set $\mathcal{S}\leftarrow\mathcal{S}\cup\mathcal{T}$\;
            }
        }
        Record $\alpha^{[r]}\leftarrow\alpha$, $\mathcal{S}^{[r]}\leftarrow\mathcal{S}$, and $\mathcal{L}^{[r]}\leftarrow\mathcal{L}(\alpha)$\;
    }
    \BlankLine
    \tcp{Step 3: selection of the restart}
    Set $\hat{r}\leftarrow\argmin_{r}\mathcal{L}^{[r]}$, $\hat{\mathcal{S}}\leftarrow\mathcal{S}^{[\hat{r}]}$, and $\hat{\alpha}_j\leftarrow\alpha_j^{[\hat{r}]}$ for $j\in\hat{\mathcal{S}}$\;
\end{algorithm}

\subsection{A posteriori check of the first-order condition}
\label{subsec:posterior_check}

Once $\hat{H}$ is computed from a fit,
the derivative can be evaluated at any candidate point, and thus the first-order optimality conditions are checked as follows.

Let $\{(\beta_j,\gamma_j)\}_{j=1}^{M}\subset B\times\Gamma$ be a set of candidate points, and write $D_j(\hat{H})\coloneq D(\hat{H};(\beta_j,\gamma_j))$. For a fit $\hat{H}$ with weights $\hat{\alpha}$, we report the two conditions
\begin{equation}
    \label{eq:dictionary_kkt}
    \text{(i)}\quad \max_{j=1,\ldots,M} D_j(\hat{H})\leq\epsilon_{\mathrm{tol}},\qquad
    \text{(ii)}\quad \max_{{j:\hat\alpha_j>0}}|D_j(\hat{H})|\leq\epsilon_{\mathrm{tol}}.
\end{equation}
The first requires that no candidate point increases the marginal log-likelihood, and the second that the directional derivative vanishes on the estimated support.

\subsection{Aggregation of experts}
\label{sec:experts_aggregation}

We finally present an aggregation procedure for the experts.
The output of EXMoE (Algorithm \ref{alg:exmoe}) tends to exhibit redundant support points even when the true measure $H^{\ast}$ in a simulation setting is a finite discrete measure. This overestimation phenomenon is pointed out for the classical NPMLE by \cite{Jiang2025-am}, in which the estimated prior measure is modified post-hoc by the Bayesian information criterion (BIC; \cite{schwarz1978estimating}).
Inspired by this modification, we consider the aggregation of support points based on hierarchical clustering of the estimated support points with respect to regression coefficients $\hat{\beta}_j$.

Our aggregation follows an observation in practice. The redundancy of $\hat{H}$ appears as groups of support points whose regression coefficients $\hat{\beta}_j$ are close to each other. Therefore, we aggregate the experts with similar regression coefficients $\hat{\beta}_j$ to the center of cluster $\bar{\beta}_{k}$ (Figure \ref{fig:agg_beta_after}) and combine the corresponding gate functions to form the aggregated gate function $g^{\dagger}_{k}$.

\begin{figure}[htbp]
    \centering
    \begin{subfigure}[b]{0.4\textwidth}
    \centering\includegraphics[width=\linewidth]{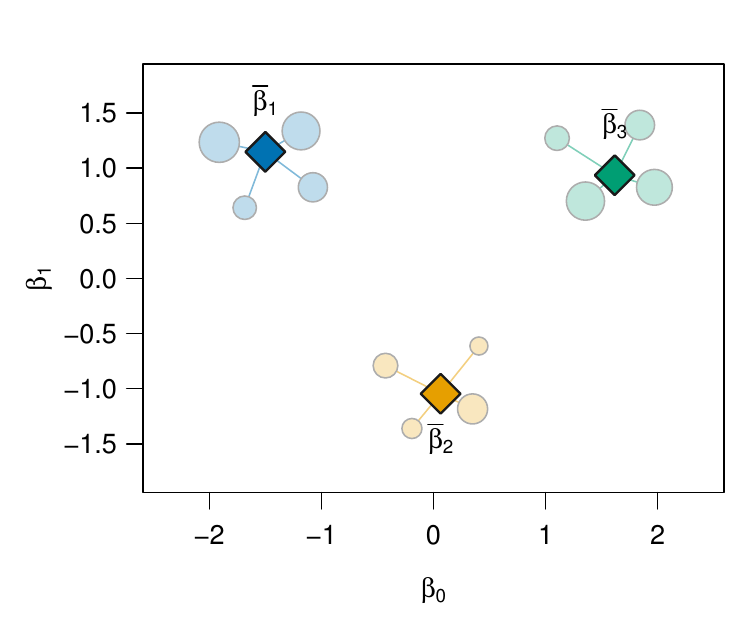}
    \caption{$(\beta_0,\beta_{1})$ plane}
        \label{fig:agg_beta_after}
     \end{subfigure}
     \begin{subfigure}[b]{0.4\textwidth}
     \centering\includegraphics[width=\linewidth]{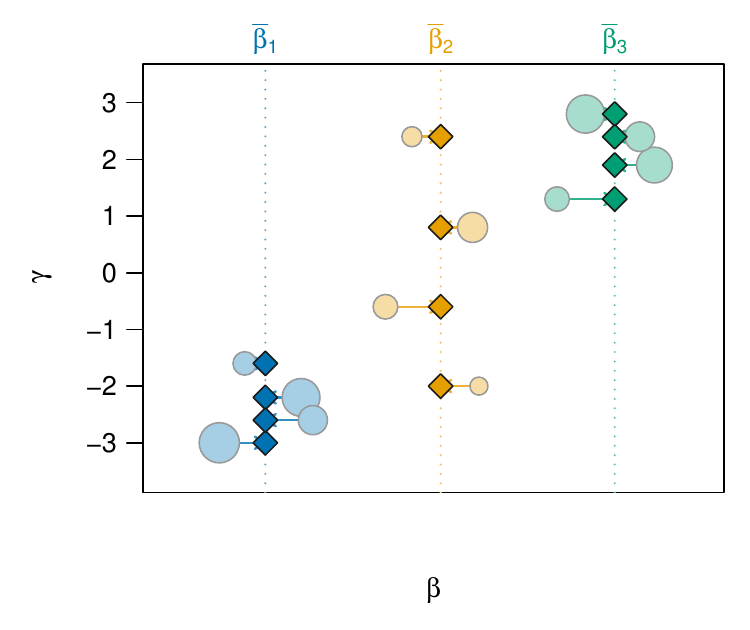}
     \caption{$(\beta,\gamma)$ hyperplane}
     \label{fig:aggregation_mixing}
     \end{subfigure}
    \caption{Illustration of the aggregation of the regression coefficients on a toy estimate $\hat{H}$ with $12$ support points (circles). The size of a circle is proportional to its average gate weight $\bar{g}_j$ in \eqref{def:average_gate_weight}. The $K=3$ clusters $\mathcal{C}_1,\mathcal{C}_2,\mathcal{C}_3$ (colors) and the aggregated coefficients $\bar{\beta}_k$ of \eqref{def:aggregated_gate} (diamonds).  (a): Aggregation to the three clusters with $p=2$ on the plane of $(\beta_0,\beta_1)$.
    (b): Geometrical illustration of the aggregation from $\hat{H}$ to $\bar{H}$ on $(\beta,\gamma)$ hyperplane.}
    \label{fig:aggregation_beta}
\end{figure}

Write the output of either algorithm as
\begin{equation}
    \label{def:fitted_gate}
    \hat{H}=\sum_{j\in\hat{\mathcal{S}}}\hat{\alpha}_j\,\delta_{(\hat{\beta}_j,\hat{\gamma}_j)},\qquad
    \hat{G}(\cdot|z)=\sum_{j\in\hat{\mathcal{S}}}\hat{g}_j(z)\,\delta_{\hat{\beta}_j},\qquad
    \hat{g}_j(z)\coloneq g_j(z;\hat{\alpha},\hat{\gamma})=\frac{\hat{\alpha}_j e^{z^{\top}\hat{\gamma}_j}}{\sum_{k\in\hat{\mathcal{S}}}\hat{\alpha}_k e^{z^{\top}\hat{\gamma}_k}},
\end{equation}
where $\hat{\mathcal{S}}$ is the index set of the support points with positive weights and $\hat{g}_j$ is the softmax gate \eqref{def:extended_softmax_gate_eq} of the $j$-th expert evaluated at the estimate. Relabeling the indices in $\hat{\mathcal{S}}$ as $1,\ldots,\hat{J}$ with $\hat{J}\coloneq|\hat{\mathcal{S}}|$ recovers the representation \eqref{def:proposed_NPMLE_discrete} in Section 2.

Firstly, we perform clustering of the support points $\{\hat{\beta}_j:j\in\hat{\mathcal{S}}\}$ to obtain $K$ distinct clusters $\{\mathcal{C}_1,\ldots,\mathcal{C}_K\}$ that partition $\hat{\mathcal{S}}$. Our default clustering method is Ward's method \citep{Ward1963-hierarchical} with the weights
\begin{equation}
    \label{def:average_gate_weight}
    \bar{g}_j\coloneq\frac{1}{n}\sum_{i=1}^{n}\hat{g}_j(z_i),\qquad j\in\hat{\mathcal{S}}.
\end{equation}
The weight is used to incorporate the importance of the gate on the observations. We define the aggregated support point, and the \emph{aggregated gate} by
\begin{equation}
    \label{def:aggregated_gate}
\bar{\beta}_k\coloneq\frac{\sum_{j\in\mathcal{C}_k}\bar{g}_j\hat{\beta}_j}{\sum_{j\in\mathcal{C}_k}\bar{g}_j},\qquad
    g^{\dagger}_k(z)\coloneq\sum_{j\in\mathcal{C}_k}\hat{g}_j(z),
    \qquad k=1,\ldots,K,
\end{equation}
respectively.
 The aggregated conditional prior and the corresponding marginal density are given by
\begin{equation}
    \label{def:aggregated_prior}
    \bar{G}(\cdot|z)\coloneq\sum_{k=1}^{K}g^{\dagger}_k(z)\,\delta_{\bar{\beta}_k},\qquad
    f_{\bar{G}}(y|x,z)=\sum_{k=1}^{K}g^{\dagger}_k(z)\,\psi\bigl(y;\theta(x^{\top}\bar{\beta}_k),\phi\bigr),
\end{equation}
respectively.
The aggregated gate $g^{\dagger}_k(z)$ is the probability that a unit with covariate $z$ belongs to the $k$-th regression relationship, and $\bar{\beta}_k$ is the regression coefficient of that relationship.

The number of clusters $K$ is determined through the elbow method with the Kneedle rule (see Appendix \ref{sec:keedle}) or BIC calculated as

\begin{equation}
    \label{def:bic_aggregation}
    \mathrm{BIC}(K)\coloneq-2\sum_{i=1}^{n}\log f_{\bar{G}_K}(y_i|x_i,z_i)+pK\log n,\qquad \hat{K}\coloneq\argmin_{K=1,\ldots,K_{\max}}\mathrm{BIC}(K),
\end{equation}
where $f_{\bar{G}_K}$ is the marginal density \eqref{def:aggregated_prior} of the aggregated model with $K$ clusters. The elbow method is our default, and BIC is reported for comparison. In our experience, BIC tends to select a larger $K$ than the elbow method when the true prior is not a finite discrete measure.

From the perspective of $\hat{H}$, the aggregation moves the regression coefficient of each support point of $\hat{H}$ to the center of its cluster and keeps the weight and the gate parameter (Figure \ref{fig:aggregation_mixing}), that is,
\begin{equation}
    \label{def:aggregated_H}
    \bar{H}\coloneq\sum_{k=1}^{K}\sum_{j\in\mathcal{C}_k}\hat{\alpha}_j\,\delta_{(\bar{\beta}_k,\hat{\gamma}_j)}.
\end{equation}
Thus $G_{\bar{H}}(\cdot|z)=\bar{G}(\cdot|z)$ remains in the infinite Mixture-of-Expertizing class $\mathcal{G}_{\mathrm{iMoE}}$.

\section{Theoretical results}
In this section, we show several theoretical results for the iMoE-NPMLE \eqref{def:proposed_NPMLE}.

\label{sec:Theory}
\subsection{Posterior representation in the hierarchical unit-level GLM}
\label{subsec:Theory_interpretation}

We first derive a representation of the conditional covariate-dependent prior in a hierarchical unit-level GLM in which the covariate $z_{i}$ is random.
Consider the hierarchical unit-level GLM
\begin{equation}
    \label{model:generative_representation}
    \begin{split}
        y_i \sim \mathrm{EF}(\mu_i,\phi), \
        \mu_{i}= g^{-1}(\eta_i),\   \eta_i = x_i^{\top}\beta_i, \quad z_i \overset{\mathrm{ind}}{\sim} p(z|\gamma_i),
        \quad (\beta_i,\gamma_i) \overset{\mathrm{i.i.d}}{\sim} \tilde{H}, \quad  \quad i=1,\ldots,n,
    \end{split}
\end{equation}
where $\tilde{H}\in \mathcal{P}(B\times\Gamma)$ is a prior probability measure on the pairs of the regression coefficient and the covariate parameter, and $p(z|\gamma)$ is the distribution of the covariate $z_i$ given its parameter $\gamma_i$. The explanatory variable $x_i$ affects only through the linear predictor and is treated as given, as in the model \eqref{model: unit_level_GLM_proposed}. 

Here we assume that the covariate distribution is an exponential family whose density takes the standard form
\begin{equation}
    \label{def:explanatory distribution}
    p(z|\gamma) = \frac{h(z)\exp(z^{\top}\gamma)}{c(\gamma)}, \quad  c(\gamma)\coloneq \int h(z) \exp(z^{\top}\gamma) \dd z,
\end{equation}
where $h(z)$ is a known non-negative function and $c(\gamma)$ is the normalizing constant, so that $\gamma$ is the canonical parameter and $z$ is the sufficient statistic. Then Bayes' theorem gives the posterior distribution of $(\beta_i,\gamma_i)$ given $z$ as
\begin{equation}
    \label{eq:posterior_interpretation_beta_gamma}
    \begin{split}
        \dd G(\beta,\gamma|z) &= \frac{p(z|\gamma)\dd \tilde{H}(\beta,\gamma)}{\int_{B\times\Gamma} p(z|\gamma)\dd \tilde{H}(\beta,\gamma)} \\
        &= \frac{e^{z^{\top}\gamma}\,c(\gamma)^{-1}\dd \tilde{H}(\beta,\gamma)}{\int_{B\times\Gamma} e^{z^{\top}\gamma}\,c(\gamma)^{-1}\dd \tilde{H}(\beta,\gamma)}.
    \end{split}
\end{equation}
Define the probability measure $H$ by 
\[
\dd H(\beta,\gamma)\coloneq 
\dd \tilde{H}(\beta,\gamma)/\{c(\gamma)\,d\} \quad\text{with}\quad
d\coloneq\int_{B\times\Gamma} c(\gamma)^{-1}d\tilde{H}(\beta,\gamma).\]
Marginalizing \eqref{eq:posterior_interpretation_beta_gamma} over $\gamma\in\Gamma$, we obtain the posterior distribution of the regression coefficient given the covariate,
\begin{equation}
    \label{eq:posterior_interpretation_gamma}
    \dd G(\beta|z) = \frac{\int_{\Gamma} e^{z^{\top}\gamma}\dd H(\beta,\gamma)}{\int_{B\times\Gamma} e^{z^{\top}\gamma}\dd H(\beta,\gamma)},
\end{equation}
which is exactly the conditional prior $G_H(\cdot|z)$ in $\mathcal{G}_{\mathrm{iMoE}}$ with mixing measure $H$.

Conversely, given any $H\in\mathcal{P}(B\times\Gamma)$ with $\int c(\gamma)\dd H<\infty$, the measure $\dd \tilde{H}\propto c(\gamma)\dd H$ generates $G_H(\cdot|z)$ as its posterior under \eqref{model:generative_representation}. Therefore, our estimation target can be viewed as the posterior distribution of the regression coefficient in the hierarchical model \eqref{model:generative_representation}, and vice versa.

The following proposition summarizes the above correspondence formally.

\begin{proposition}[Posterior representation of $G_{H}$]
\label{prop:covariate_representation}
Let $p^{\ast}$ be any probability density on $\mathbb{R}^{d}$ and let $H\in\mathcal{P}(B\times\Gamma)$ with $\Gamma$ bounded. Define
\begin{equation}
    \label{def:covariate_construction}
    h(z)\coloneq\frac{p^{\ast}(z)}{\int_{B\times\Gamma} e^{z^{\top}\gamma}\dd H(\beta,\gamma)},\qquad
    \text{and}\qquad
    \dd \tilde{H}(\beta,\gamma)\coloneq c(\gamma)\dd H(\beta,\gamma).
\end{equation}
Then 
$p^{\ast}$ becomes the marginal density of $z_{i}$
and $G_{H}(\cdot|z)$ becomes the posterior distribution of the regression coefficient, respectively:
that is,
\begin{equation}
    \label{eq:covariate_representation}
    \int_{B\times\Gamma} p(z|\gamma)\dd \tilde{H}(\beta,\gamma) = p^{\ast}(z)\quad\text{for all } z\in\mathbb{R}^{d},
    \quad
    \text{and}
    \quad
    \dd G(\beta|z) = \dd G_{H}(\beta|z)\quad\text{for } p^{\ast}\text{-a.e.} z.
\end{equation}
\end{proposition}
The proof is given in Appendix \ref{app:proofs}.
Note that $h(z)$ and $c(\gamma)$ do not appear in the likelihood of the iMoE-NPMLE, and hence are not estimated by it. Nevertheless, they can be obtained from \eqref{def:covariate_construction} by combining $\hat{H}$ with a density estimate of $p^{\ast}(z)$. 

\subsection{Representation ability}
\label{subsec:Theory_representaiton}

We proceed to a proposition and a theorem on the representation ability of the infinite Mixture-of-Expertizing class and of the discrete solution of the iMoE-NPMLE.

The following proposition shows that the bounded support $\Gamma$ of the gate parameter makes every member of $\mathcal{G}_{\mathrm{iMoE}}$ Lipschitz continuous in the covariate $z$, so that the size of $\Gamma$ controls how sharply the conditional prior can change with $z$.
For probability measures $\mu$ and $\nu$ on $B$, let $\|\mu-\nu\|_{\mathrm{TV}}\coloneq\sup_{A}|\mu(A)-\nu(A)|$ be the total variation distance, where the supremum is taken over Borel sets $A\subset B$, and let $W_1(\mu,\nu)\coloneq\inf_{\pi}\int\|u-v\|_2\,d\pi(u,v)$ be the $1$-Wasserstein distance, where the infimum is taken over couplings $\pi$ of $\mu$ and $\nu$. Since $B$ is bounded, $W_1$ metrizes the weak convergence on $\mathcal{P}(B)$. Let
\begin{equation}
    \label{def:R_gamma_diam_B}
    R_{\Gamma}\coloneq\sup_{\gamma\in\Gamma}\|\gamma\|_2,\qquad
    \mathrm{diam}(B)\coloneq\sup_{\beta,\beta'\in B}\|\beta-\beta'\|_2,
\end{equation}
both of which are finite because $B$ and $\Gamma$ are bounded.

\begin{proposition}[Lipschitz continuity of the iMoE class]
\label{prop:lipschitz_gamma}
Let $H\in\mathcal{P}(B\times\Gamma)$ and $z,z'\in\mathbb{R}^{d}$. Then the conditional priors 
$G_H(\cdot|z), G_H(\cdot|z') \in \mathcal{G}_{\mathrm{iMoE}}$
at $z$ and $z'$ satisfy
\begin{equation}
    \label{eq:lipschitz_w1}
    \begin{split}
        \|G_H(\cdot|z)-G_H(\cdot|z')\|_{\mathrm{TV}}&\leq 1-e^{-2R_{\Gamma}\|z-z'\|_2},\\
        W_1\bigl(G_H(\cdot|z),G_H(\cdot|z')\bigr)&\leq 2R_{\Gamma}\,\mathrm{diam}(B)\,\|z-z'\|_2.
    \end{split}
\end{equation}
In particular, $z\mapsto G_H(\cdot|z)$ is Lipschitz continuous from $(\mathbb{R}^{d},\|\cdot\|_2)$ to $(\mathcal{P}(B),W_1)$ with constant $2R_{\Gamma}\,\mathrm{diam}(B)$ independent from $H$.
\end{proposition}
The proof is given in Appendix \ref{app:proofs}. 
The quantity $1/R_{\Gamma}$ is the scale on which the prior can change. 
By the first bound in \eqref{eq:lipschitz_w1}, a total variation distance of at least $1-\epsilon$ requires $\|z-z'\|_2\geq\log(1/\epsilon)/(2R_{\Gamma})$. 
Hence no member of the class can switch its support over a distance in $z$ shorter than order $1/R_{\Gamma}$, whatever the number and the location of its support points. The bound $R_{\Gamma}$ therefore plays the role of a bandwidth.

The following Theorem \ref{thm:uat_gate} shows that the vectorized aggregated softmax gate function
$g^{\dagger}(z;\alpha,\gamma)\coloneq(g^{\dagger}_1,\ldots,g^{\dagger}_C)^{\top}:\mathbb{R}^{d}\to \mathring\Delta^{C-1}$ can approximate any continuous gate function on a compact set. 
The theorem is an existence result with a limiting case of the slopes in the gates. One characteristic structure in the proof is that gates share their slopes within groups, and the intercepts needed to represent a partition of the covariate space are recovered from the weights $\alpha$ by a reparametrization.

\begin{theorem}[Universal approximation theorem]
    \label{thm:uat_gate}
    Let $C\geq2$ be a fixed integer, let $\mathcal{Z}\subset\mathbb{R}^{d}$ be a compact set contained in $\mathrm{ball}(0,\rho)$ for some $\rho>0$, and let $\Pi^{\ast}_{C}\coloneq\left\{\pi^{\ast}:\mathcal{Z}\rightarrow \mathring\Delta^{C-1} \;\middle|\; \pi^{\ast}\;\text{is continuous} \right\}$ denote the class of continuous gate functions valued in the open simplex $\mathring\Delta^{C-1}\coloneq\{\pi\in\mathbb{R}^{C}: \pi_c>0,\;\sum_{c=1}^{C}\pi_c=1\}$. 
    
    Then, for every $q\in[1,\infty)$, every $\epsilon>0$, and every $\pi^{\ast}\in\Pi^{\ast}_{C}$, there exist a positive integer $L$, $J=LC$ softmax gates \eqref{def:extended_softmax_gate_eq} with parameters $\alpha\in\Delta^{J-1}$ and $\gamma_1,\ldots,\gamma_{J}\in\mathbb{R}^{d}$, and $C$ pairwise disjoint index subsets $\mathcal{S}_{c}\subset\{1,\ldots,J\}$, $c=1,\ldots,C$, such that
    \begin{equation}
        \left(
        \int_{\mathcal{Z}} \| g^{\dagger}(z;\alpha,\gamma)-\pi^{\ast}(z)\|_{\infty}^{q}\dd z
        \right)^{1/q}  <\epsilon,
    \end{equation}
    where $g^{\dagger}(z;\alpha,\gamma)\coloneq (g^{\dagger}_1(z;\alpha,\gamma),\ldots,g^{\dagger}_C(z;\alpha,\gamma))^{\top}$ and $g^{\dagger}_{c}(z;\alpha,\gamma)\coloneq\sum_{j\in \mathcal{S}_{c}}g_{j}(z;\alpha,\gamma)$, $c=1,\ldots,C$. 
\end{theorem}
The proof is given in Appendix \ref{app:proofs}. The gates constructed in the proof have the fused slope structure of Condition \ref{cond:fusion} in Appendix \ref{app:fused}, with $L$ groups of $C$ gates sharing the slope $\xi_l/\tau$ and with weights $\alpha$ given by the reparametrization \eqref{eq:alpha_reparametrization}. 

\section{Simulation studies}
\label{sec:Simulation}

This section provides several simulation studies.

\subsection{Data generating settings}
We illustrate the estimator on three synthetic data sets with a Gaussian response and a one-dimensional covariate. Throughout this section, $x_i=(1,s_i)^{\top}$ with $s_i\sim\mathrm{Unif}(-3,3)$, the gate covariate is $z_i=s_i$ (so that $z$ contains no constant coordinate, $p=2$ and $d=1$), and $y_i=x_i^{\top}\beta_i+\varepsilon_i$ with $\varepsilon_i\sim N(0,\sigma^{2})$ and $\beta_i\sim G^{\ast}(\cdot|s_i)$, where the dispersion $\sigma$ is treated as unknown. The three conditional priors $G^{\ast}(\cdot|s)$ are chosen to illustrate different types of covariate dependence. The first is a two-point measure whose weights oscillate in $s$. The second is a six-point measure whose weights are smooth functions of $s$, where two dominant support points switch at $s=-2$ and $s=2$. The third is a uniform distribution supported on a circle whose center and radius switch at $s=0$. In other words, the prior is a hybrid case of continuous measure in $\beta$ and discontinuous changepoint in $s$. The three data generating processes are summarized as follows.
\begin{enumerate}
    \item[(a)] \emph{sin}. The prior is a two-point measure whose weights oscillate in $s$:
        \begin{equation}
        \label{def:dgp_sin}
        G^{\ast}(\cdot|s)=\pi^{\ast}(s)\,\delta_{(1,\,2)}+\{1-\pi^{\ast}(s)\}\,\delta_{(-1,\,-2)},\qquad
        \pi^{\ast}(s)=\frac{1+\sin(\pi s)}{2},
    \end{equation}
    with $n=800$ and $\sigma=0.5$. The gate completes three periods on $(-3,3)$, so that the dominant support point changes six times. This setting is a simple example of a continuous conditional prior with a small number of support points and may not be in $\mathcal{G}_{\mathrm{iMoE}}$. From the perspective of Theorem \ref{thm:uat_gate}, the gate function $\pi^{\ast}(s)$ can be approximated by a sum of many  softmax gates with slope of limits. Thus, the simulation is to certify how well the finite-discrete gates with limited slope can approximate the true gate function in practice.
    \item[(b)] \emph{smooth Voronoi}. The prior is a six-point measure whose weights are the product of a smooth membership of three regions and fixed within-region weights:
    \begin{equation}
        \label{def:dgp_smooth_voronoi}
        G^{\ast}(\cdot|s)=\sum_{r=1}^{3}\sum_{a=1}^{2}w_{ra}\phi_r(s)\,\delta_{\beta_{ra}},\qquad
        \phi_r(s)=\frac{\exp\{-(s-\mu_r)^{2}/(2h_s^{2})\}}{\sum_{r^{\prime}=1}^{3}\exp\{-(s-\mu_{r^{\prime}})^{2}/(2h_s^{2})\}},
    \end{equation}
    with the region centers $(\mu_1,\mu_2,\mu_3)=(-2,0,2)$, the bandwidth $h_s=0.6$, and the support points and weights
    \begin{equation}
        \label{def:dgp_smooth_voronoi_atoms}
        \begin{array}{lll}
            \beta_{11}=(2,\,1.5)^{\top},  & \beta_{12}=(1,\,-0.5)^{\top}, & (w_{11},w_{12})=(0.7,\,0.3), \\
            \beta_{21}=(0,\,2)^{\top},    & \beta_{22}=(-1,\,-1)^{\top},  & (w_{21},w_{22})=(0.5,\,0.5), \\
            \beta_{31}=(-2,\,2)^{\top},   & \beta_{32}=(3,\,-1.5)^{\top}, & (w_{31},w_{32})=(0.6,\,0.4),
        \end{array}
    \end{equation}
    with $n=1500$ and $\sigma=0.3$. 

    By transforming the inside of exponent in $\phi_r(s)$, we have
    \begin{equation}
        \label{eq:dgp_smooth_voronoi_softmax}
        \phi_r(s)=\frac{\exp\{-\mu_r^{2}/(2h_s^{2})\}\exp\{(\mu_r/h_s^{2})\,s\}}{\sum_{r^{\prime}=1}^{3}\exp\{-\mu_{r^{\prime}}^{2}/(2h_s^{2})\}\exp\{(\mu_{r^{\prime}}/h_s^{2})\,s\}}.
    \end{equation}
    Then, with reparameterizations $\gamma_{ra}=\mu_r/h_s^{2}$, and  
    \begin{equation}
        \label{eq:dgp_smooth_voronoi_alpha}
        \alpha_{ra}=\frac{w_{ra}\,e^{-\mu_r^{2}/(2h_s^{2})}}{\sum_{r^{\prime}=1}^{3}\sum_{a^{\prime}=1}^{2}w_{r^{\prime}a^{\prime}}\,e^{-\mu_{r^{\prime}}^{2}/(2h_s^{2})}},
    \end{equation}
    gives 
    \begin{equation}
        w_{ra}\phi_r(s)=
        \frac{\alpha_{ra}\,e^{\gamma_{ra}s}}{\sum_{r^{\prime}=1}^{3}\sum_{a^{\prime}=1}^{2}\alpha_{r^{\prime}a^{\prime}}\,e^{\gamma_{r^{\prime}a^{\prime}}s}}.
    \end{equation}
    Hence, $G^{\ast}$ belongs to the infinite Mixture-of-Expertizing class with six support points.

    \item[(c)] \emph{switching circle}. The prior is the uniform distribution only supported on a circle whose center and radius switch at $s=0$:
    \begin{equation}
        \label{def:dgp_threshold_circle}
        G^{\ast}(\cdot|s)=
        \begin{cases}
            \mathrm{Unif}\bigl(\mathcal{C}(c_A,r_A)\bigr), & s\leq 0,\\
            \mathrm{Unif}\bigl(\mathcal{C}(c_B,r_B)\bigr), & s>0,
        \end{cases}
        \qquad
        \mathcal{C}(c,r)\coloneq\bigl\{c+r(\cos\vartheta,\sin\vartheta)^{\top}:\vartheta\in[0,2\pi)\bigr\},
    \end{equation}
    where $\mathrm{Unif}(\mathcal{C}(c,r))$ is the distribution of $c+r(\cos\vartheta,\sin\vartheta)^{\top}$ with $\vartheta\sim\mathrm{Unif}(0,2\pi)$, $c_A=(-1.5,\,0.5)^{\top}$, $r_A=0.6$, $c_B=(1.5,\,-0.6)^{\top}$, $r_B=1.0$, $n=2500$, and $\sigma=0.3$. The prior is continuous in $\beta$ but discontinuous in $s$ with a jump at $s=0$. Thus, the true prior lies outside the $\mathcal{G}_{\mathrm{iMoE}}$ for any fixed $\Gamma$ by Proposition \ref{prop:lipschitz_gamma}. Our interest is how well the estimator can approximate both the continuous distribution in $\beta$ and the discontinuity at $s=0$.

\end{enumerate}
Figure \ref{fig:sim_true} shows the true gate functions of the sin and smooth Voronoi priors and the support of the switching circle prior in the top row, and the three simulated data sets as scatter plots in the bottom row.

\begin{figure}[!htb]
    \centering
    \begin{subfigure}[b]{0.32\textwidth}\centering\includegraphics[width=\linewidth]{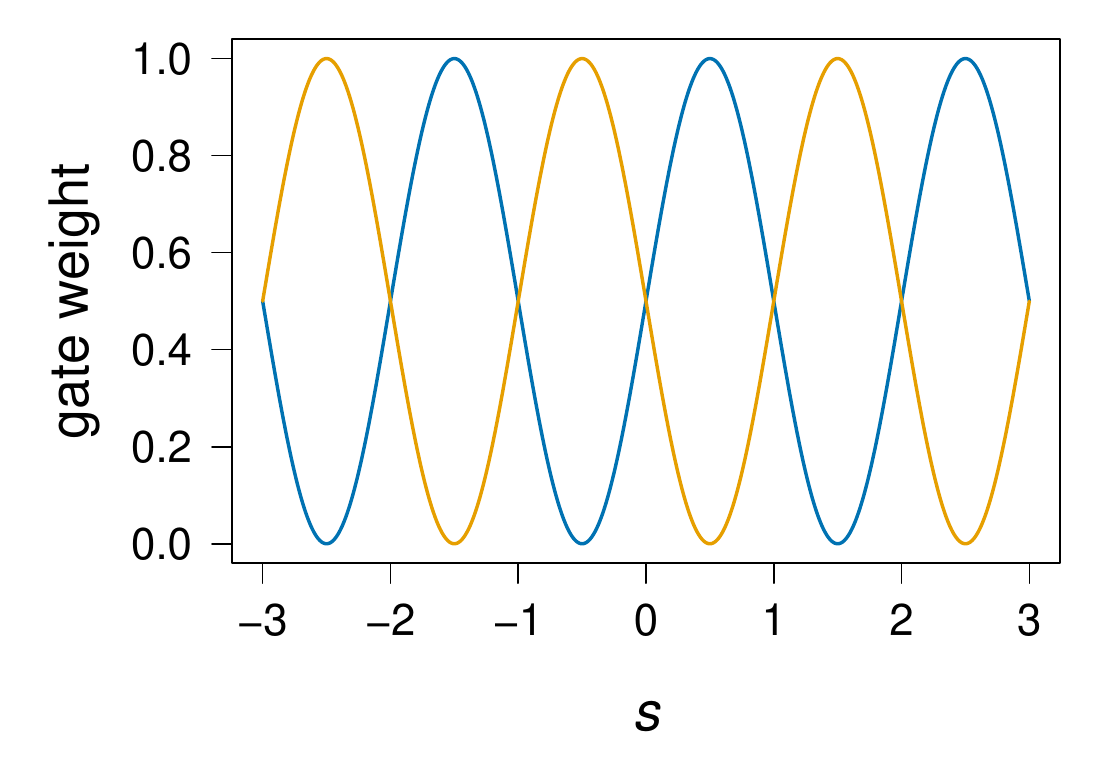}\caption{sin}\end{subfigure}\hfill
    \begin{subfigure}[b]{0.32\textwidth}\centering\includegraphics[width=\linewidth]{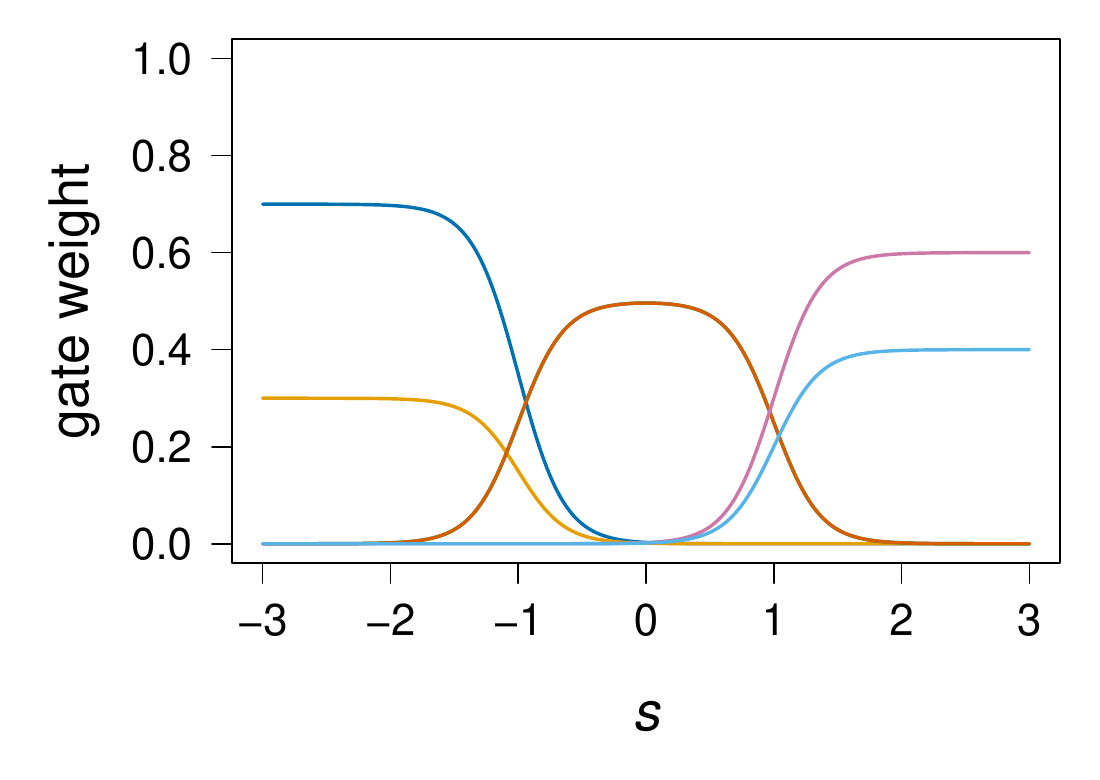}\caption{smooth Voronoi}\end{subfigure}\hfill
    \begin{subfigure}[b]{0.32\textwidth}\centering\includegraphics[width=\linewidth]{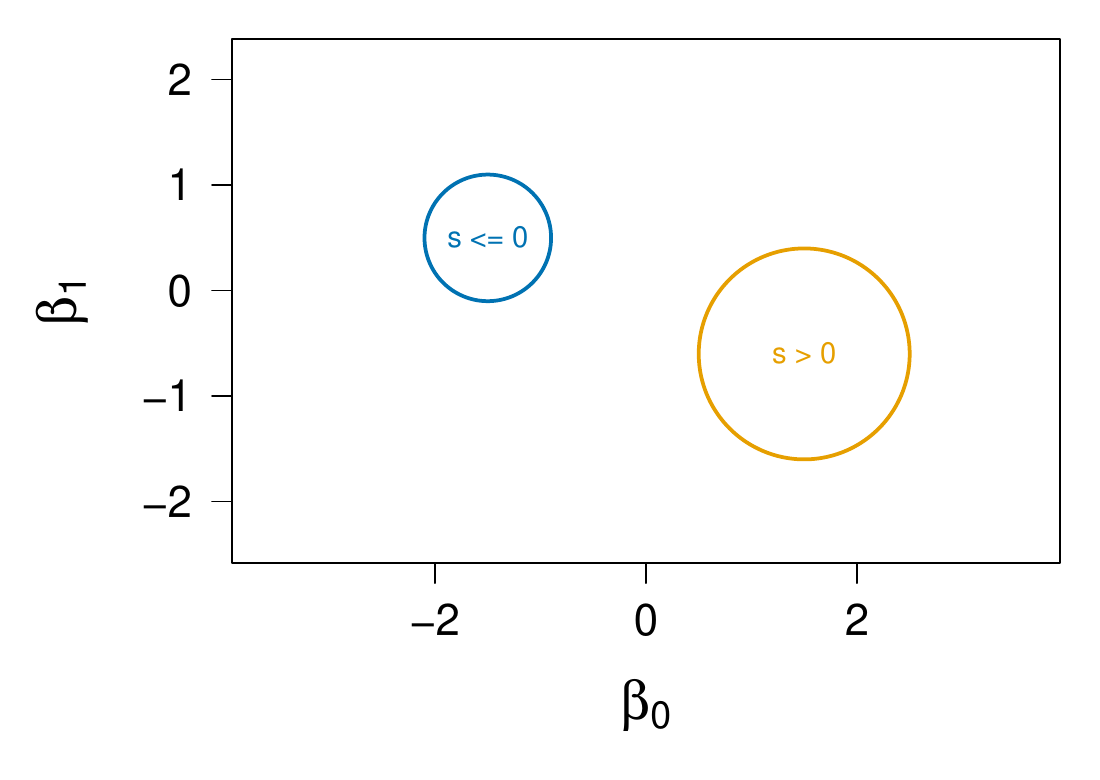}\caption{switching circle}\end{subfigure}\\[1ex]
    \begin{subfigure}[b]{0.32\textwidth}\centering\includegraphics[width=\linewidth]{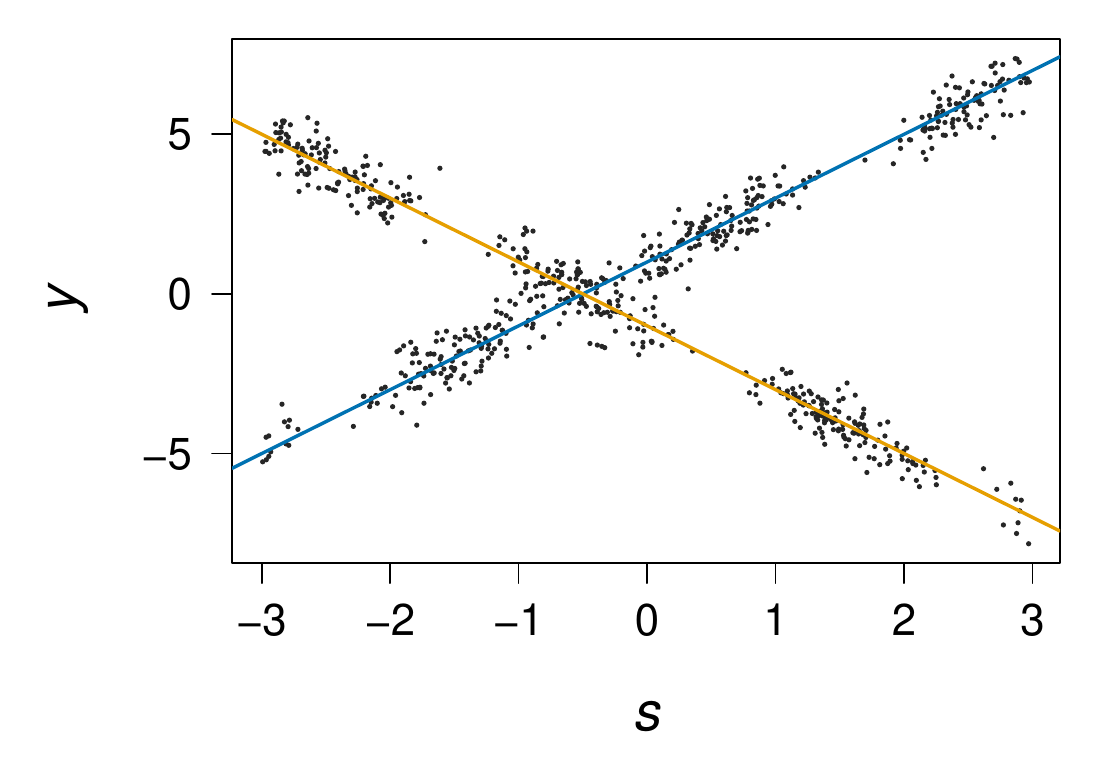}\caption{sin}\end{subfigure}\hfill
    \begin{subfigure}[b]{0.32\textwidth}\centering\includegraphics[width=\linewidth]{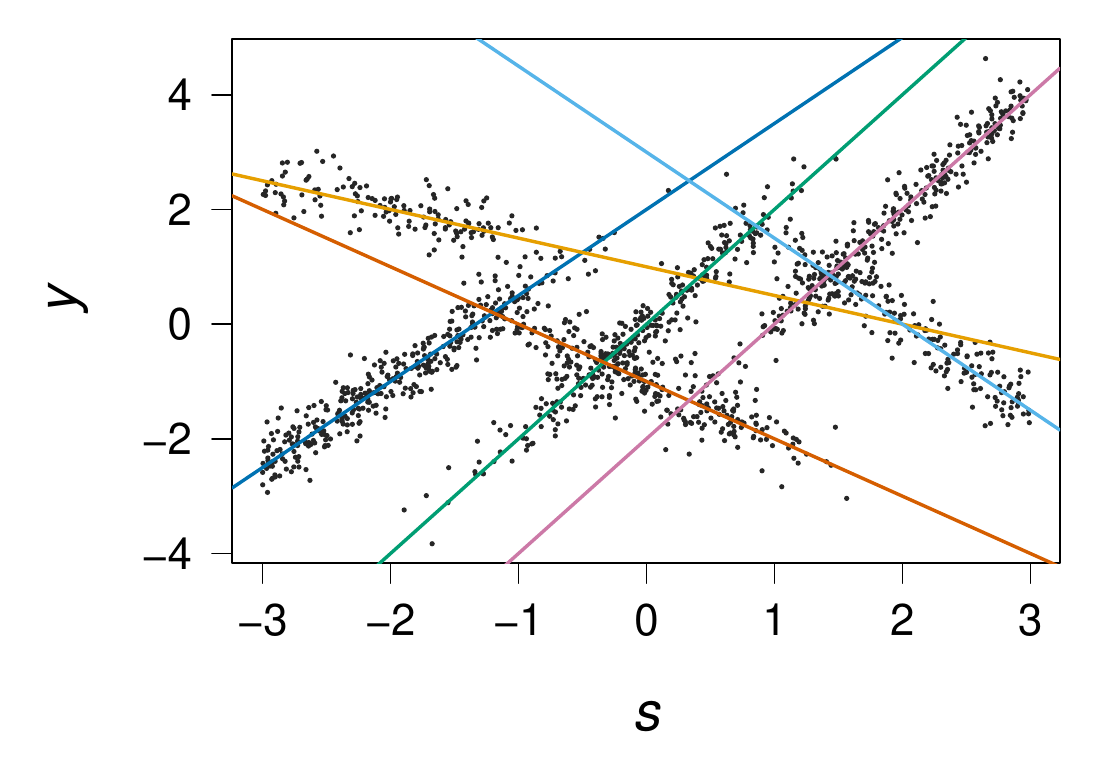}\caption{smooth Voronoi}\end{subfigure}\hfill
    \begin{subfigure}[b]{0.32\textwidth}\centering\includegraphics[width=\linewidth]{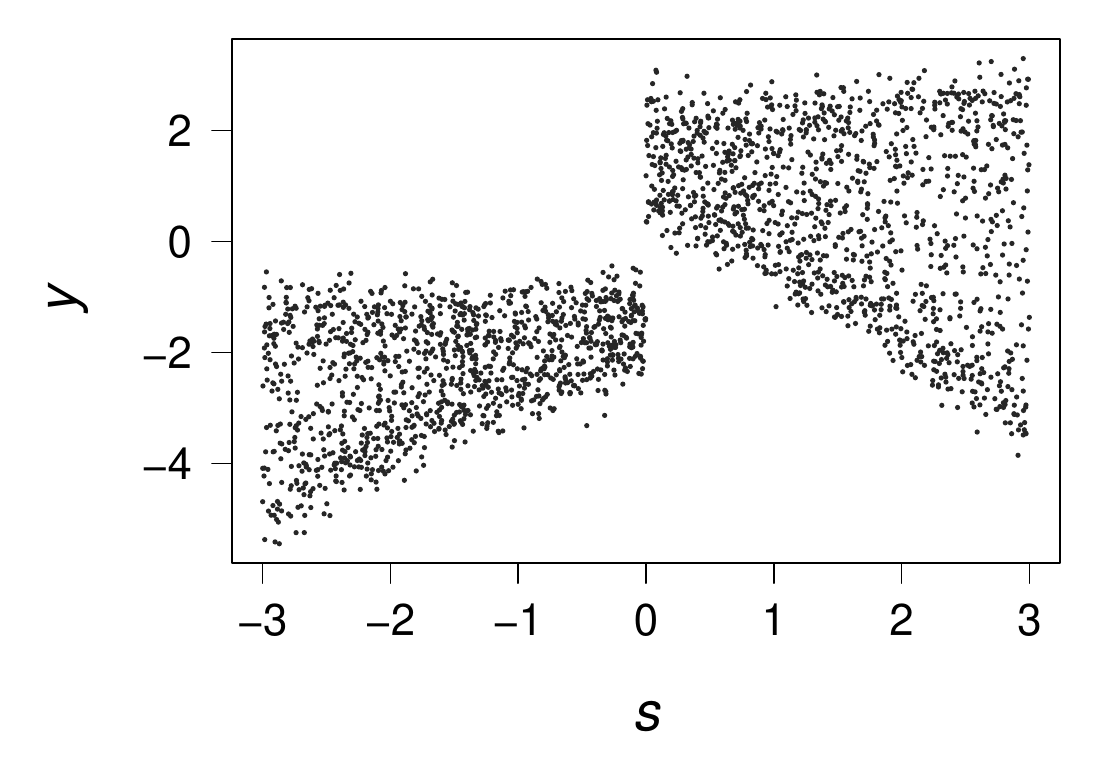}\caption{switching circle}\end{subfigure}
    \caption{The true conditional priors (top) and the three simulated data sets (bottom). (a), (b): the true gate functions of the components of the sin and smooth Voronoi priors. Note that the two gate functions in (b) take the same value at every $s$, so that the two components are not identifiable in this figure.
    (c): the support of the switching circle prior in the $(\beta_0,\beta_1)$ plane, the circle $\mathcal{C}(c_A,r_A)$ for $s\leq0$ and the circle $\mathcal{C}(c_B,r_B)$ for $s>0$. (d)--(f): the simulated data sets as scatter plots. For the sin and smooth Voronoi data, the true regression lines of the components of the prior are drawn in the colors of the components in (a) and (b).}
    \label{fig:sim_true}
    \label{fig:sim_data}
\end{figure}

\subsection{Estimation settings}

We use Grid-EXMoE (Algorithm \ref{alg:fixed_dictionary}) as the default procedure. The exemplars are generated by Algorithm \ref{alg:exemplar} with $P=n$, $m=p+1=3$, $B=[-25,25]^{2}$, and $M_{\beta}=500$, and the gate grid is $\mathcal{G}_h=\{-10,-9,\ldots,10\}$ on $\Gamma=[-10,10]$, so that the dictionary \eqref{def:dictionary} has $M=500\times21=10{,}500$ candidate points. The size of the dictionary is chosen to exceed $2n$, which is the bound on the number of support points of the iMoE-NPMLE in \eqref{def:solution_of_NPMLE_discrete_mixing_measure}. Since Algorithm \ref{alg:fixed_dictionary} only updates the weights $\alpha$ on the fixed dictionary, it is computationally efficient for a large number of candidate points. This is why we use Algorithm \ref{alg:fixed_dictionary} as the default.  We also perform Algorithm \ref{alg:exmoe}, which starts from the $500$ exemplars and updates all parameters at a considerably higher computational cost. The result achieves a slightly higher likelihood on a smaller support, as reported in Appendix \ref{app:alg2_comparison}.

The standard deviation $\sigma$ of the Gaussian experts is chosen by 5-fold cross-validation from the grid $\{0.1,0.15,0.2,0.3,0.4,0.5,0.7,1.0\}$. The exemplars are regenerated from the training data. All the results below are obtained at $\hat{\sigma}$ with the exemplars regenerated from the full sample.

\subsection{Results}
\subsubsection{Summary of estimated models}
Table \ref{tab:sim_alg3} summarizes the estimated values for each data set. 
The cross-validation selects the true value $\sigma=0.3$ for the smooth Voronoi and switching circle data, and $\hat{\sigma}=0.4$ for the sin data. The last two columns report that the first-order conditions \eqref{eq:dictionary_kkt} are satisfied on the estimated support and off the support within the tolerance.
 On the other hand, the estimated measures have $41$ to $67$ support points, which are far fewer than the bound $2n$ and the size of the dictionary. For classical NPMLE, the estimated support points exhibit such a sparse structure, and we observe the same in our simulations.

\begin{table}[!htb]
    \centering
    \caption{Algorithm \ref{alg:fixed_dictionary} on the three simulated data sets. The columns report the size $\hat{J}=|\hat{\mathcal{S}}|$ of the estimated support, the first-order conditions \eqref{eq:dictionary_kkt} on the estimated support $\hat{\mathcal{S}}$ and off the support $\hat{\mathcal{S}}^{c}$ per unit, where $D_j$ is replaced by the reparametrized gradient $\tilde{D}_j$ of Appendix \ref{sec:PFW_repara}. The estimates are at the selected $\hat{\sigma}$. }
    \label{tab:sim_alg3}
    \begin{tabular}{lrrrr}
        \hline
        Data  & $\hat{\sigma}$ & $\hat{J}$ & $\max_{j\in\hat{\mathcal{S}}}|\tilde{D}_j|/n$ & $\max_{j\in\hat{\mathcal{S}}^{c}}\tilde{D}_j/n$  \\
        \hline
        sin                & 0.4 & 43 & $9.8\times10^{-6}$ & $-4.3\times10^{-5}$  \\
        smooth Voronoi    & $\bm{0.3}$ & 41 & $9.8\times10^{-6}$ & $-1.1\times10^{-4}$  \\
        switching circle  & $\bm{0.3}$ & 67 & $9.5\times10^{-6}$ & $-1.5\times10^{-5}$  \\
        \hline
    \end{tabular}
\end{table}

\subsubsection{Estimated prior measure on the predictor scale}
\label{sec:pushforward}

Figure \ref{fig:sim_heat} displays the pushforward of the estimated conditional prior 
onto the predictor scale $(1,s)^{\top}\beta_{j}$:
\begin{equation}
    \label{def:pushforward}
    \hat{G}(\eta|s)\coloneq\sum_{j=1}^{\hat{J}} \hat{g}_j(s) \delta_{(1,s)^{\top}\hat{\beta}_{j}}(\eta).
\end{equation}
For display, the discrete measure \eqref{def:pushforward} is smoothed in $\eta$ by a Gaussian kernel. We compute
\begin{equation}
    \label{def:pushforward_heat}
    \hat{f}_{h}(\eta|s)\coloneq\sum_{j=1}^{\hat{J}}\hat{g}_j(s)\,\varphi_{h}\bigl(\eta-(1,s)^{\top}\hat{\beta}_j\bigr),
\end{equation}
where $\varphi_{h}$ is the Gaussian density with standard deviation $h$, and $h$ is set to one percent of the range of the observed responses. The heat map displays $\hat{f}_{h}(\eta|s)^{1/2}$, with the square root taken to compress the dynamic range of the color scale so that the support points with small gate weights remain visible. The kernel serves only the display and is not part of the estimate. When taking $h=\hat{\sigma}$, the smoothed pushforward $\hat{f}_{\hat{\sigma}}(\eta|s)$ \eqref{def:pushforward_heat} is equivalent to the conditional density of the linear regression at $s$. Because our purpose is to visualize only the linear relationship induced by the estimated conditional prior, we take a smaller bandwidth $h$ to avoid oversmoothing. The same panels show the estimated conditional mean $\hat{E}[y|s]=\int \eta\,d\hat{G}(\eta|s)=\sum_{j=1}^{\hat{J}}\hat{g}_j(s)\,(1,s)^{\top}\hat{\beta}_j$ together with the true conditional mean $E[y|s]$.

The heat map shows that the estimated pushforward follows the true regression lines over the whole range of $s$, and the gate weights change their dominant regression lines with $s$. Although the number of estimated support points is too large compared to the true ones in the sin data and smooth Voronoi data, the shown pushforward almost captures the true regression structures. For the switching circle data, the true prior is continuous with fixed $s$, thus the pushforward also has a continuous structure. Our method approximates such a continuous structure with many regression lines, equivalently corresponding to support points $\hat{J}=67$. Moreover, gate functions switch their dominant regression lines at $s=0$ to approximate the discontinuous change. In addition, for all three datasets, the estimated conditional mean follows the true conditional mean.

\begin{figure}[!htb]
    \centering
    \begin{subfigure}[b]{0.32\textwidth}\centering\includegraphics[width=\linewidth]{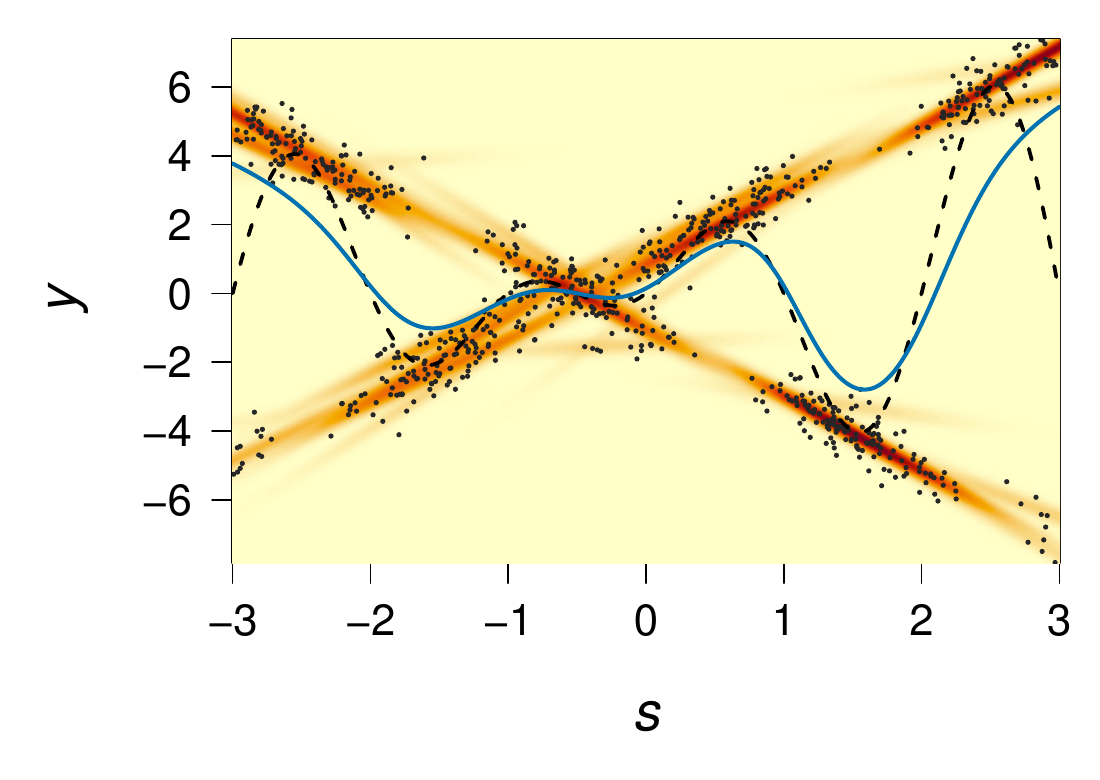}\caption{sin data}\label{fig:sim_heat_sin}\end{subfigure}\hfill
    \begin{subfigure}[b]{0.32\textwidth}\centering\includegraphics[width=\linewidth]{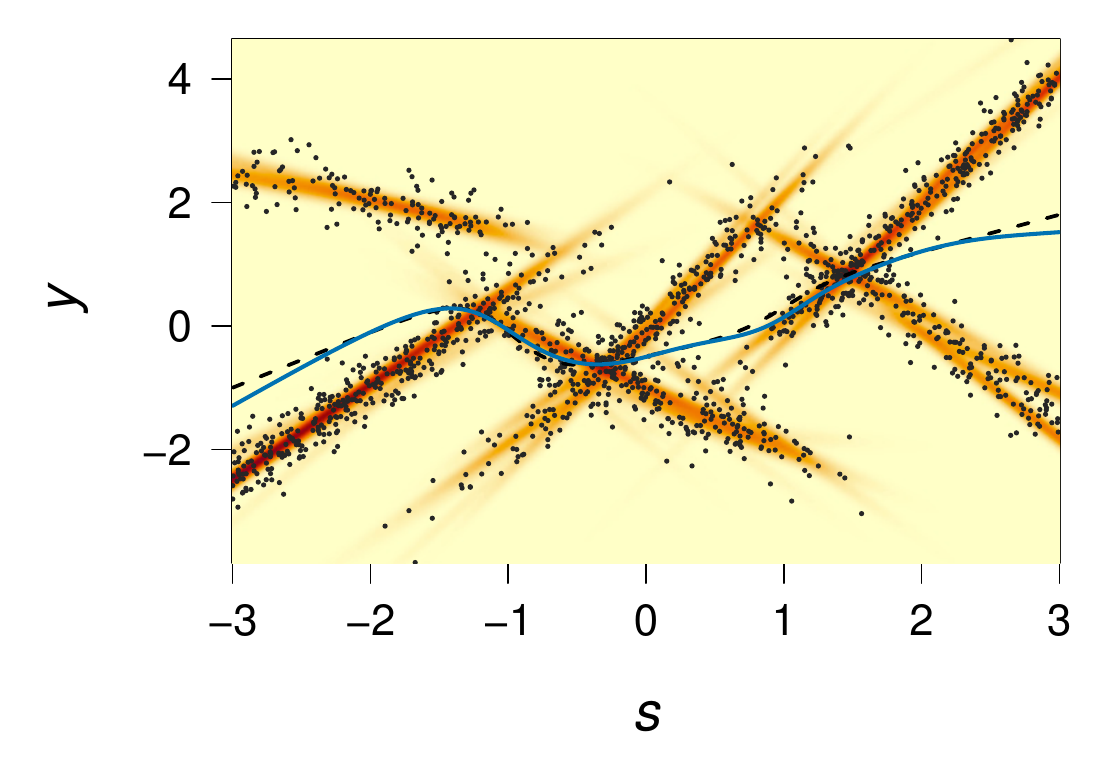}\caption{smooth Voronoi data}\label{fig:sim_heat_sv}\end{subfigure}\hfill
    \begin{subfigure}[b]{0.32\textwidth}\centering\includegraphics[width=\linewidth]{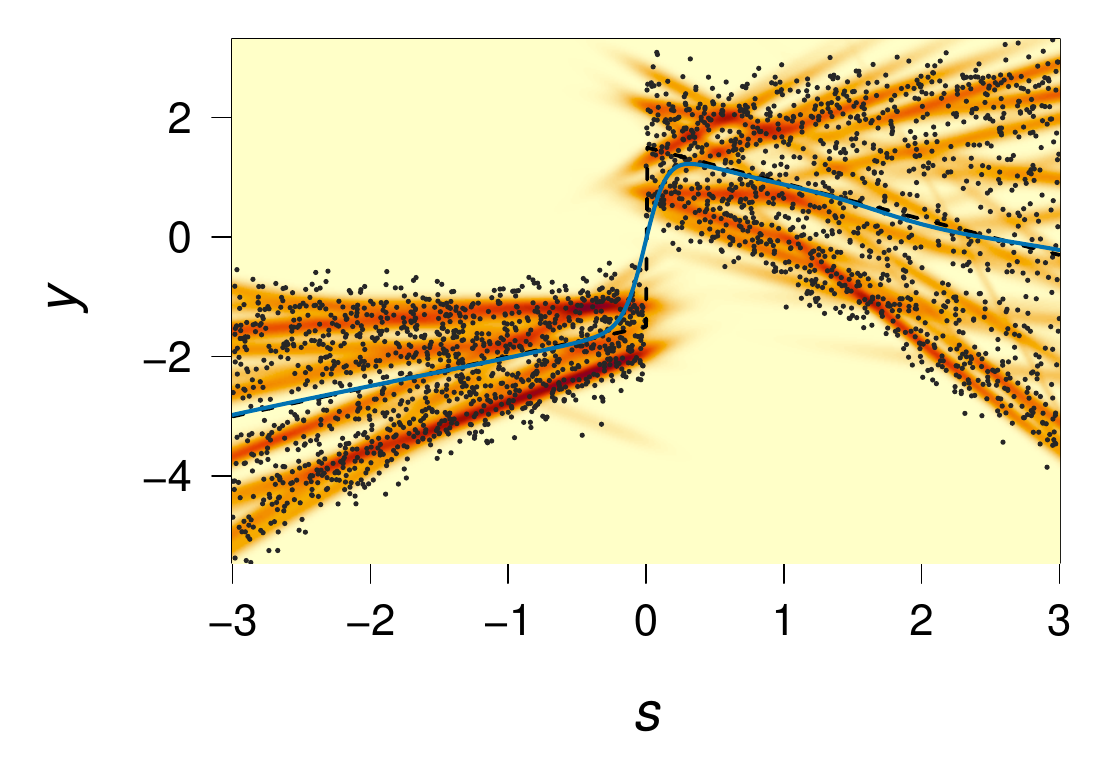}\caption{switching circle data}\label{fig:sim_heat_tc}\end{subfigure}
    \caption{Estimates of Algorithm \ref{alg:fixed_dictionary} for the three simulated data sets. Each panel shows the data (black dots), the heat map of the smoothed pushforward $\hat{f}_{h}(\eta|s)^{1/2}$ of the estimated conditional prior \eqref{def:pushforward_heat}, the estimated conditional mean $\hat{E}[y|s]$ (blue solid line), and the true conditional mean $E[y|s]$ (black dashed line).}
    \label{fig:sim_heat}
\end{figure}

\subsubsection{Prior measure}

Next, we illustrate the estimated conditional prior measure $\hat{G}(\beta|s)$ for each dataset, respectively.

\textbf{Sin data:}
Figures \ref{fig:sim_sin_prior_a} and \ref{fig:sim_sin_prior_b} show the estimated support points of $\hat{G}(\beta|s)$ with their weight probability $\hat{g}_{j}(s)$ for the sin data at $s=-1.5$, where the true prior is a point mass, and at $s=2$, where it is an equal mixture of the two true support points. The common observation over $s$ is that the support points with large weights are located on the two true support points. In particular, the gate weight is spread over a few support points close to the true point. 

Figures \ref{fig:sim_sin_gates_separate} and \ref{fig:sim_sin_gates_elbow} show how the gate is built. Because we only consider the linear term in the softmax function, each gate can only represent one modal structure, which is not enough to approximate the sin function's oscillating structure. Figure \ref{fig:sim_sin_gates_separate} would indicate that the softmax gate functions are divided into two sets colored in blue and yellow, with their roles: one tends to play the role of one mode out of the multiple modes in $\pi^{\ast}(s)$, and the other also does so in $1-\pi^{\ast}(s)$. Thus, the cloudy support points around the true point would get together to approximate the true sin gate function. Indeed, the aggregated gate function $g^{\dagger}_k(s)$ \eqref{def:aggregated_gate} obtained by the aggregation procedure in Section \ref{sec:experts_aggregation} follows the sin structure as shown in Figure \ref{fig:sim_sin_gates_elbow}. In this case, the two clusters selected by the elbow method coincide with the groups of support points assigned to the nearest \emph{oracle} support points.

A clearer approximation result of the sin gate function is shown in Figure \ref{fig:sim_alg2_gates_sin}, which is computed by Algorithm \ref{alg:exmoe}. The rather compressed structure in Figure \ref{fig:sim_sin_gates_elbow} would come from the constraint $\Gamma=[-10,10]$.

\begin{figure}[!htb]
    \centering
    \begin{subfigure}[b]{0.4\textwidth}\centering\includegraphics[width=\linewidth]{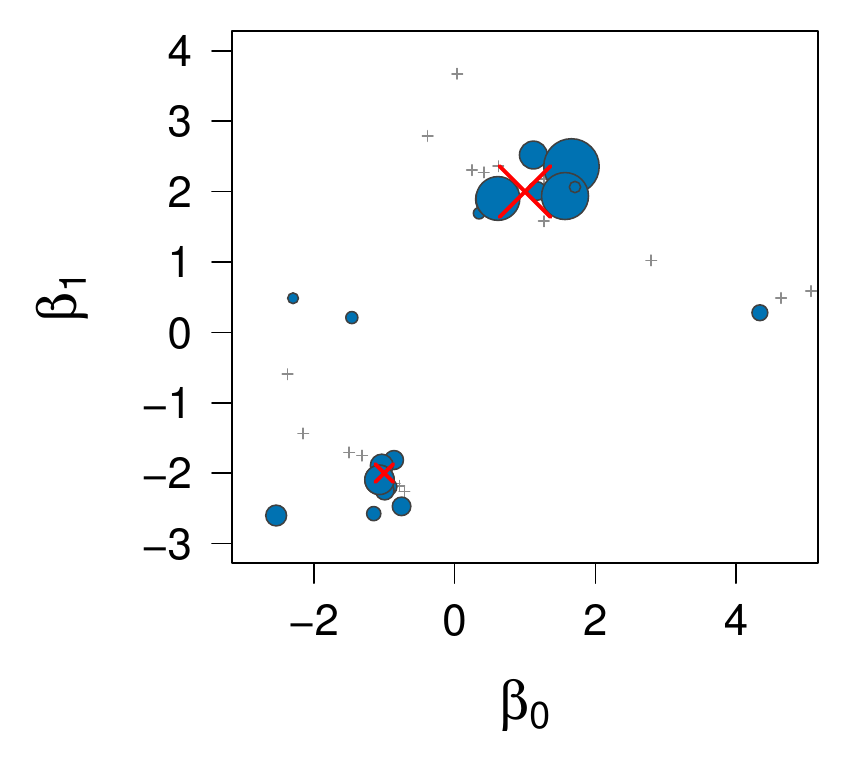}\caption{$\hat{G}(\cdot|s)$, $s=-1.5$}\label{fig:sim_sin_prior_a}\end{subfigure}\hspace{0.05\textwidth}
    \begin{subfigure}[b]{0.4\textwidth}\centering\includegraphics[width=\linewidth]{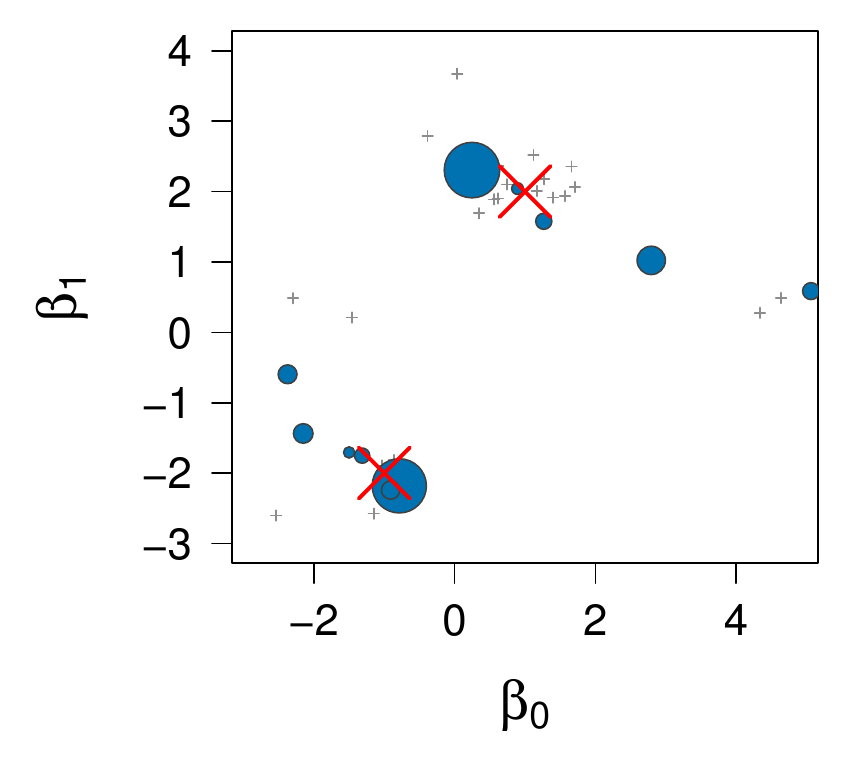}\caption{$\hat{G}(\cdot|s)$, $s=2$}\label{fig:sim_sin_prior_b}\end{subfigure}\\[1ex]
    \begin{subfigure}[b]{0.4\textwidth}\centering\includegraphics[width=\linewidth]{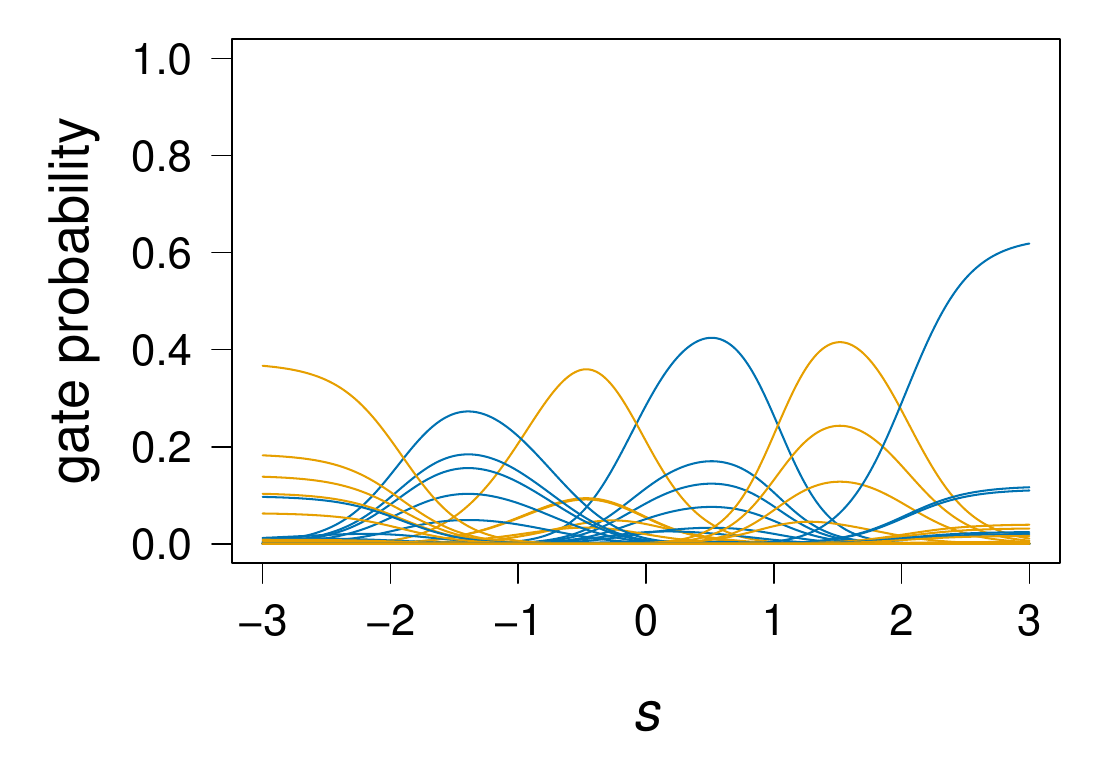}\caption{individual gates}\label{fig:sim_sin_gates_separate}\end{subfigure}\hspace{0.05\textwidth}
    \begin{subfigure}[b]{0.4\textwidth}\centering\includegraphics[width=\linewidth]{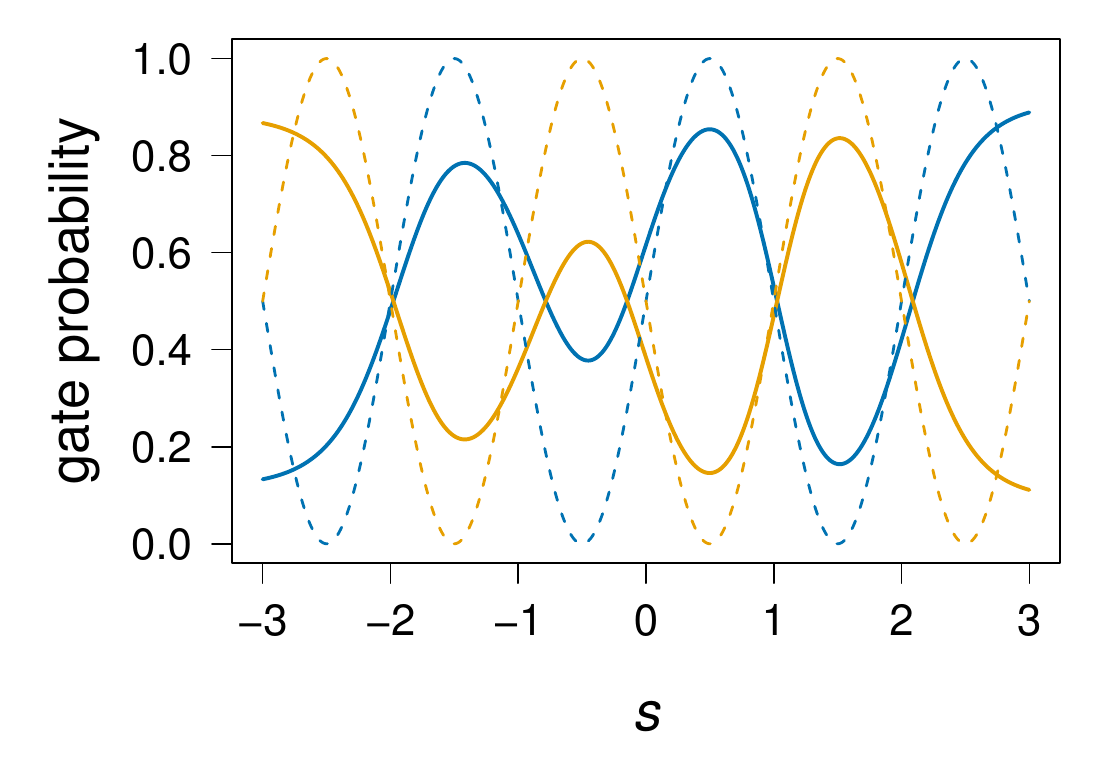}\caption{aggregated gates}\label{fig:sim_sin_gates_elbow}\end{subfigure}
    \caption{Estimate of Algorithm \ref{alg:fixed_dictionary} for the sin data. (a), (b): the conditional prior $\hat{G}(\cdot|s)=\sum_j\hat{g}_j(s)\delta_{\hat{\beta}_j}$ at $s=-1.5$, where the true prior is a point mass, and at $s=2$, where it puts equal weights on the two true support points. The size of a support point is proportional to its gate weight $\hat{g}_j(s)$ (points with gate weight below $0.001$ are drawn as small crosses), and the true prior $G^{\ast}(\cdot|s)$ is drawn as red crosses whose size is proportional to the true weight. (c): the individual gates $\hat{g}_j(s)$ of the $43$ estimated support points, colored by the nearest true support point with the colors of Figure \ref{fig:sim_data}. (d): the aggregated gates $g^{\dagger}_k(s)$ of \eqref{def:aggregated_gate} for the $\hat{K}=2$ clusters selected by the elbow method with the weights $\bar{g}_j$ (solid) and the true gates $\pi^{\ast}(s)$ and $1-\pi^{\ast}(s)$ (dashed). The two clusters coincide with the groups of support points nearest to the two true support points.}
    \label{fig:sim_sin}
    \label{fig:sim_sin_gates}
\end{figure}

\textbf{Smooth Voronoi data}:
For the smooth Voronoi data, the true gate is in $\mathcal{G}_{\mathrm{iMoE}}$. Figure \ref{fig:sim_smooth_voronoi} shows that the conditional prior $\hat{G}(\cdot|s)$ places its large and few masses near the two true support points that dominate at $s$. However, Figure \ref{fig:sim_sv_gates_separate} shows that each gate does not approximate the true six gate functions, and mass is dispersed across gates. Whereas, similar to the sin data, the nearest oracle aggregated gates recover the true structure as shown in Figure \ref{fig:sim_sv_gates_aggregated}. The elbow rule with the weights $\bar{g}_j$ selects $\hat{K}=4$ clusters, which merge two of the six true regression coefficients into neighboring clusters. While the number of components is misspecified, the gate structure still resembles the true one.

\textbf{Switching circle data}:
 Figure \ref{fig:sim_threshold_circle} displays the switching structure of the gates in the $(\beta_0,\beta_1)$ plane. At each fixed $s$, the estimated conditional prior $\hat{G}(\cdot|s)=\sum_{j}\hat{g}_j(s)\delta_{\hat{\beta}_j}$ is drawn with the size of each support point proportional to its gate weight $\hat{g}_j(s)$, together with the two true circles. At $s=-2$ and $2$, the gate weight is carried almost entirely by the support points near the active circle. At $s=-0.1$ and $0.1$, near the switching boundary $s=0$, $0.79$ and $0.78$ of the mass has already moved to the active circle on either side of the switch. Within each side, the continuous prior is approximated by $18$ to $25$ support points with gate weight above $0.01$, which are scattered on and around the active circle. Moreover, the active gate for each side of the circle has almost the same sign as $\hat{\gamma}_{j}$. Therefore, the partition of the covariate space would be realized by gates of opposite slopes, and the prior within the same side is realized by a discrete measure on the exemplars, as in the classical NPMLE.

\subsubsection{The number of aggregated gates}
Table \ref{tab:sim_K_selection} reports the number of clusters $\hat{K}$ selected by the elbow method and by $\mathrm{BIC}(K)$ in \eqref{def:bic_aggregation} when the aggregation of Section \ref{sec:experts_aggregation} is applied to the estimate of Algorithm \ref{alg:fixed_dictionary}. 

For the sin data, the elbow method and BIC both select the true number of clusters for the estimate of Algorithm \ref{alg:fixed_dictionary}, since its support consists of two tight clouds around the true coefficients and a few support points of negligible weight (Figure \ref{fig:sim_sin}), and the resulting aggregated gates are shown in Figure \ref{fig:sim_sin_gates_elbow}. For the smooth Voronoi data, the elbow method selects $4$ clusters, and BIC selects $10$, which is too large. For the switching circle data, BIC reaches $11$, whereas the elbow method returns $3$ clusters. Overall, the elbow method gives a parsimonious summary and BIC a fine one. The two agree only when the estimated support is concentrated on a few well-separated coefficients, as for the sin data.

\begin{table}[!htb]
    \centering
    \caption{Number of aggregated gates $\hat{K}$ selected by the elbow method and by BIC for the estimate of Algorithm \ref{alg:fixed_dictionary}, with the average gate weights $\bar{g}_j$ of Section \ref{sec:experts_aggregation}, over $K=1,\ldots,K_{\max}=12$. The last column is the number of distinct regression coefficients of the true prior. The switching circle prior has a continuous support and no finite number of clusters.}
    \label{tab:sim_K_selection}
    \small
    \begin{tabular}{lccc}
        \hline
        Data & elbow & BIC & true \\
        \hline
        sin              & $\bm{2}$ & $\bm{2}$ & 2 \\
        smooth Voronoi   & 4 & 10 & 6 \\
        switching circle & 3 & 11 & -- \\
        \hline
    \end{tabular}
\end{table}

\section{Application to real datasets}
\label{sec:real_data}
We report the results of application to real datasets. Throughout this section, the following process is common.  

First, we use Free-EXMoE (Algorithm \ref{alg:exmoe}) as the algorithm in this section. We choose it for the following reasons: Multi-dimensional grid points of $\Gamma$ make the dictionary easily vast, which yields high computational costs. Furthermore, the hundreds of fixed exemplars would not be enough to search for true support points of $\beta$ in multivariate settings. Free-EXMoE 
automatically updates $\beta$ to search for support points. Although the hundreds of experts may also not be enough to approximate the iMoE-NPMLE bound by $2n$ points, the practical solution exhibits sparsity in Section \ref{sec:Simulation}. Thus, Algorithm \ref{alg:exmoe} may produce a parsimonious solution, but it is a reasonable choice for real-data analysis. Through this section, we collectively refer to the Free-EXMoE and obtained solution of it as EXMoE for brevity.

Second, the dataset is divided into 80\% of units for the train and $20\%$ of units for the test. For the training dataset, we use $3$-fold CV to choose the hyperparameter $\lambda$ and the dispersion parameter $\phi$ via 2D-grid search. In the CV, each fit starts from the single initial value $\gamma=0$ and is stopped by the rule of Algorithm \ref{alg:exmoe} with $T_{\max}=100$ outer iterations. The final fit on the whole training data at the selected $(\lambda,\phi)$ uses $R=3$ restarts, $T_{\max}=2000$, and $\epsilon_{\mathrm{tol}}=10^{-9}$. For the generation of exemplars, parameters are set as $M=250, P=200$ and $m=60$ in Algorithm \ref{alg:exemplar}.

\subsection{California housing dataset}
\label{sec:california_housing}
Our first application is the California housing dataset \citep{Pace1997-housing} obtained from ``scikit-learn'' \citep{Pedregosa2011-sklearn} in Python. This dataset consists of 20,640 census block groups and nine attributes per group for California districts. The response variable is the median house value expressed in hundreds of thousands of dollars (\$100,000). The other eight attributes are median income (MedInc), median house age (HouseAge), average number of rooms (AveRooms), average number of bedrooms (AveBedrms), Population (Population), average number of household members (AveOccup), latitude, and longitude in each census block. 

We set latitude and longitude as the covariates because housing prices naturally vary by submarket based on geographical coordinates. Other six attributes are set as explanatory variables $x_i\in\mathbb{R}^{7}$. We exclude 965 blocks exceeding the upper value of housing price \$500{,}000, and we have a total of $n=19{,}675$ units. The response variables, explanatory variables, and covariates are standardized after the exclusion. For this dataset, we apply our method with Gaussian experts.

As the result of CV, $\hat{\sigma}=0.2$ and $\lambda=0.03$ are selected. The number of active support points is $\hat{J}=65$. We evaluate the predictive performance on the test data, including $3{,}935$ blocks. The predictive mean $\hat{E}[y|x,z]=\sum_{j}\hat{g}_j(z)\,x^{\top}\hat{\beta}_j$ is compared with the OLS and the non-conditional NPMLE \citep{Jiang2025-am} in the coefficient of determination (Table \ref{tab:housing_comparison}). For all methods, the predictive mean is converted and clipped to the observed range of the response, $[0.15, 5]$ in units of \$100,000, before computing $R^{2}$. This process reduces the effect of $11$ test blocks with extreme values of AveOccup. The best is
our EXMoE with $R^{2}=0.719$. This indicates the existence of heterogeneity depending on the geographical coordinates, which cannot be captured by classical NPMLE. In addition, we report EXMoEs with aggregation, and aggregation trades predictive accuracy for a parsimonious description of estimated linear relationships.

\begin{table}[htbp]
    \centering
    \caption{Coefficient of determination $R^{2}$ of the predictive mean for EXMoE, its aggregated versions and compared methods.}
    \label{tab:housing_comparison}
    \begin{tabular}{lr}
        \hline
        Method & $R^{2}$ \\
        \hline
        EXMoE & \underline{0.719} \\
        EXMoE, aggregated ($\hat{K}=3$, elbow) & 0.598 \\
        EXMoE, aggregated ($\hat{K}=12$, BIC) & 0.634 \\
        NPMLE of \cite{Jiang2025-am} & 0.595 \\
        OLS & 0.575 \\
        \hline
    \end{tabular}
\end{table}

Next, we show that assignment of each block to the gate functions.
Figure \ref{fig:housing_atom_map} assigns every block $z_i$ to the largest gate, $\argmax_{j}\hat{g}_j(z_i)$, and Figure \ref{fig:housing_cluster_map} assigns to the cluster with the largest aggregated gate, $\argmax_{k}g^{\dagger}_k(z_i)$ with $k=1,\ldots,3$ determined by elbow method. Before the aggregation, $32$ of the $65$ support points are the largest in some blocks, and the map is fragmented into their territories. Table \ref{tab:housing_clusters} reports the summary of aggregated gates and experts: the number of support points, the total average aggregated gate weight, the share of the blocks assigned to each cluster, and the aggregated regression coefficients $\bar{\beta}_k$ for each attribute. Cluster 1 dominates $68\%$ of the blocks, which are the coastal metropolitan areas around Los Angeles, San Diego, the San Francisco Bay, and Sacramento. Cluster 2 dominates $22\%$ of the blocks, which are the inland areas including the northern part of the state. It has the smallest intercept; that is, the area has the lowest price baseline. Cluster 3 dominates $10\%$ of the blocks, almost all of which are in the San Francisco Bay area, and has the largest intercept and the largest negative coefficient of the average occupancy. 

\begin{figure}[htbp]
    \centering
    \begin{subfigure}[b]{0.48\textwidth}\centering\includegraphics[width=\linewidth]{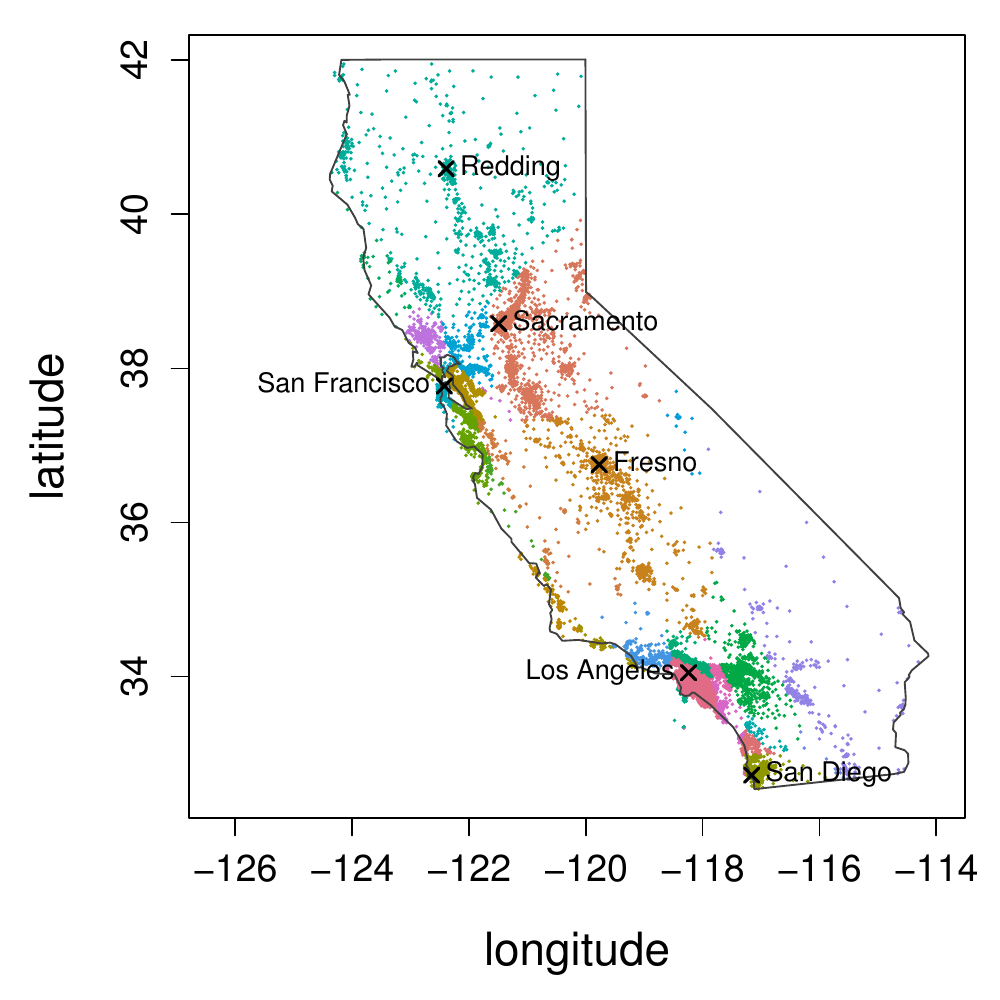}\caption{support points before aggregation}\label{fig:housing_atom_map}\end{subfigure}\hfill
    \begin{subfigure}[b]{0.48\textwidth}\centering\includegraphics[width=\linewidth]{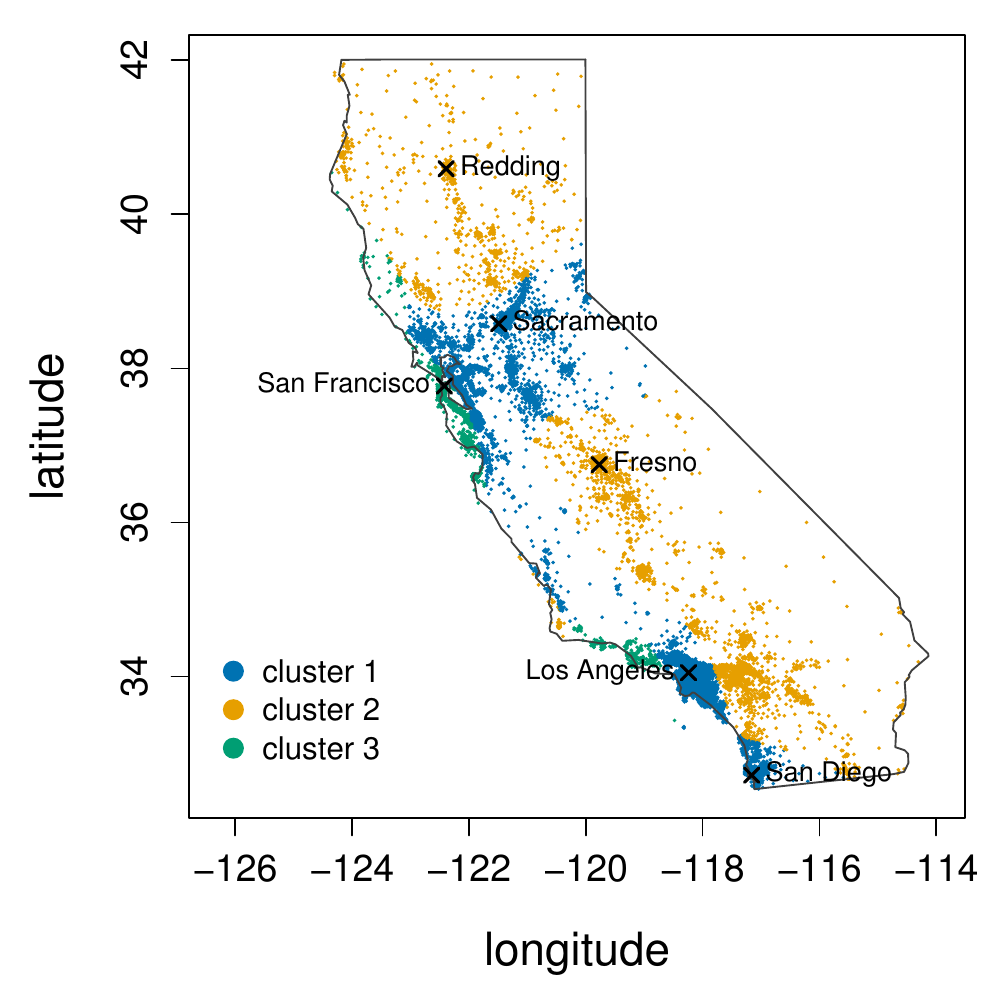}\caption{aggregated clusters}\label{fig:housing_cluster_map}\end{subfigure}
    \caption{California housing data, all $19{,}675$ blocks, with the state boundary and six major cities marked by crosses. (a): each block is colored by the support point with the largest gate $\hat{g}_j(z_i)$ among the $65$ support points with a distinct color. Only $32$ support points are the largest at some block. (b): each block is colored by the cluster with the largest aggregated gate $g^{\dagger}_k(z_i)$ among the $\hat{K}=3$ clusters.}
    \label{fig:housing_assignment}
\end{figure}

\begin{table}[htbp]
    \centering
    \caption{Aggregated clusters of the $65$ support points selected by the elbow method: the number of support points $|\mathcal{C}_k|$, the total average gate weight $\sum_{j\in\mathcal{C}_k}\bar{g}_j$, the share of the blocks assigned to the cluster in Figure \ref{fig:housing_cluster_map}, and the aggregated regression coefficients $\bar{\beta}_k$ of \eqref{def:aggregated_gate} on the standardized scale of the response and the explanatory variables.}
    \label{tab:housing_clusters}
    \small
    \begin{tabular}{crrrrrr}
        \hline
        $k$ & $|\mathcal{C}_k|$ & $\sum_{j\in\mathcal{C}_k}\bar{g}_j$ & share & intercept & MedInc & HouseAge \\
        \hline
        1 & 28 & 0.521 & 0.679 & $-0.101$ & 0.464 & 0.018 \\
        2 & 20 & 0.292 & 0.215 & $-0.603$ & 0.353 & $-0.022$ \\
        3 & 17 & 0.188 & 0.105 & 0.627 & 0.538 & 0.080 \\
        \hline
    \end{tabular}

    \vspace{1.5ex}
    \begin{tabular}{crrrr}
        \hline
        $k$ & AveRooms & AveBedrms & Population & AveOccup \\
        \hline
        1 & 0.013 & 0.244 & 0.036 & $-2.800$ \\
        2 & 0.118 & $-0.037$ & 0.031 & $-0.002$ \\
        3 & 0.042 & 0.454 & 0.011 & $-7.364$ \\
        \hline
    \end{tabular}
\end{table}
Figure \ref{fig:housing_path} visualizes the estimated conditional prior $\hat{G}(\cdot|z)$ with its $65$ support points. The visualization is along with a path $z(t)$, $t\in[0,1]$, that connects the six cities by straight lines shown in Figure \ref{fig:housing_path_map}. The path goes through in the order San Diego (SD; $t=0$), Los Angeles (LA; $t=0.15$), Fresno (FR; $t=0.43$), San Francisco (SF; $t=0.69$), Sacramento (SC; $t=0.80$), and Redding (RD; $t=1$).
For each explanatory variable $c$, the panel shows the marginal prior $\hat{G}(\beta_c|z(t))=\sum_{j\in\hat{\mathcal{S}}}\hat{g}_j(z(t))\,\delta_{\hat{\beta}_{j,c}}$ as a heat map. The horizontal axis is $t\in[0,1]$, and the vertical axis represents support points $\hat{\beta}_j$. Each probability $\hat{G}(\beta_c|z(t))$ is smoothed by a narrow Gaussian kernel as in Section \ref{sec:pushforward}. 
The figure clearly shows the varying structure of the prior measure along with the path $z(t)$. In particular, the Metropolitan Bay Area, such as SD, LA, and SF, clearly shows heterogeneous structures for all regression coefficients. Whereas FR and RD, located in inland areas, show homogeneous structures. SC is in the middle and has three dominant masses. 

Throughout this analysis, EXMoE achieves high prediction performance and inference of complex heterogeneous structure via decomposition of geographical dependence. Furthermore, EXMoE with aggregation also gives a parsimonious interpretation with a trade-off in prediction performance.
\begin{figure}[htbp]
    \centering
    \begin{subfigure}[b]{0.32\textwidth}\centering\includegraphics[width=\linewidth]{figures/housing_v2/housing_path_map.pdf}\caption{the path $z(t)$}\label{fig:housing_path_map}\end{subfigure}\hfill
    \begin{subfigure}[b]{0.32\textwidth}\centering\includegraphics[width=\linewidth]{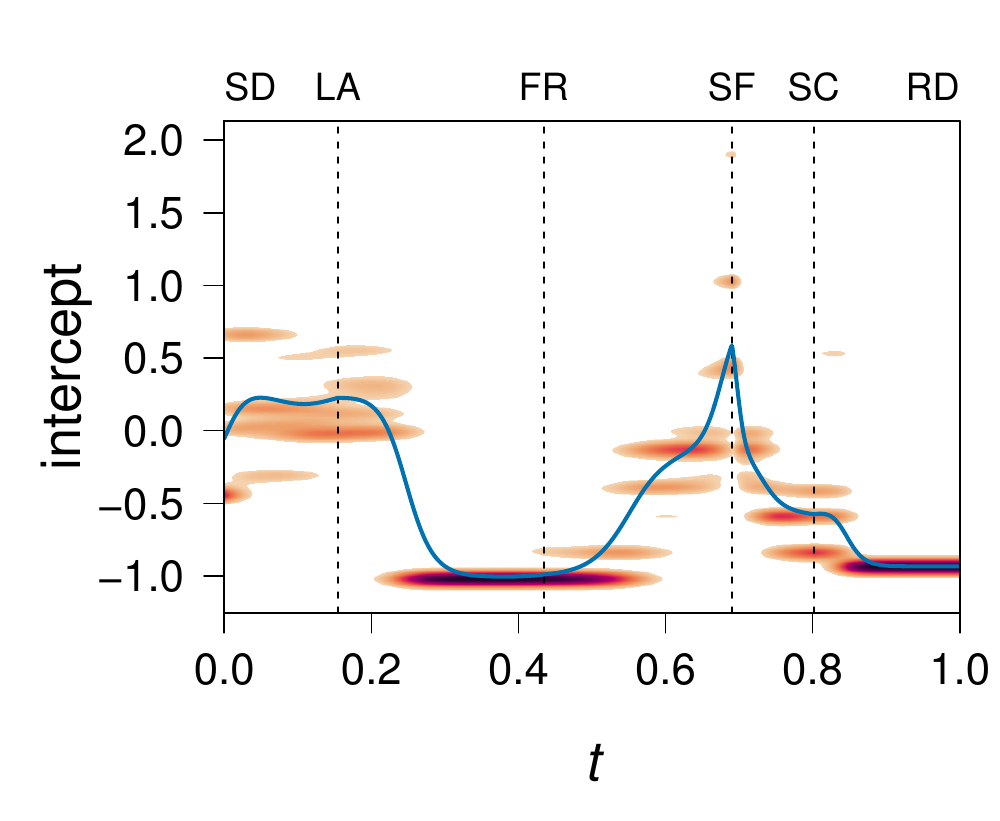}\caption{intercept}\end{subfigure}\hfill
    \begin{subfigure}[b]{0.32\textwidth}\centering\includegraphics[width=\linewidth]{figures/housing_v2/housing_path_heat_MedInc.pdf}\caption{MedInc}\end{subfigure}\\[1ex]
    \begin{subfigure}[b]{0.32\textwidth}\centering\includegraphics[width=\linewidth]{figures/housing_v2/housing_path_heat_HouseAge.pdf}\caption{HouseAge}\end{subfigure}\hfill
    \begin{subfigure}[b]{0.32\textwidth}\centering\includegraphics[width=\linewidth]{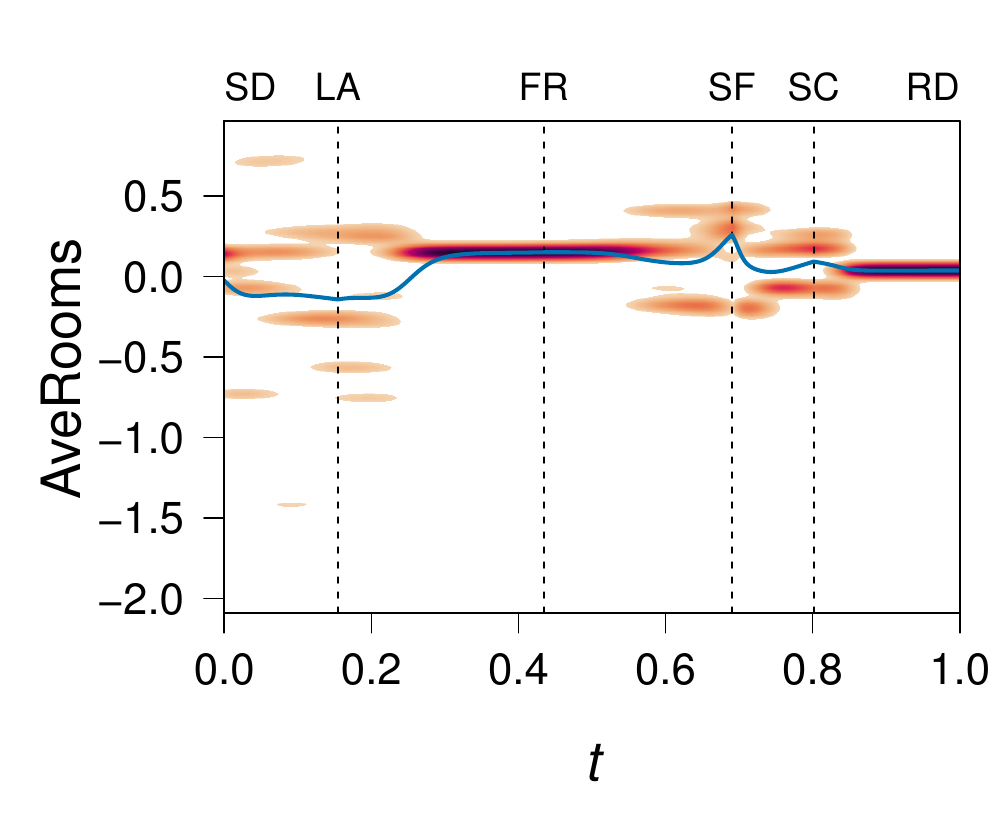}\caption{AveRooms}\end{subfigure}\hfill
    \begin{subfigure}[b]{0.32\textwidth}\centering\includegraphics[width=\linewidth]{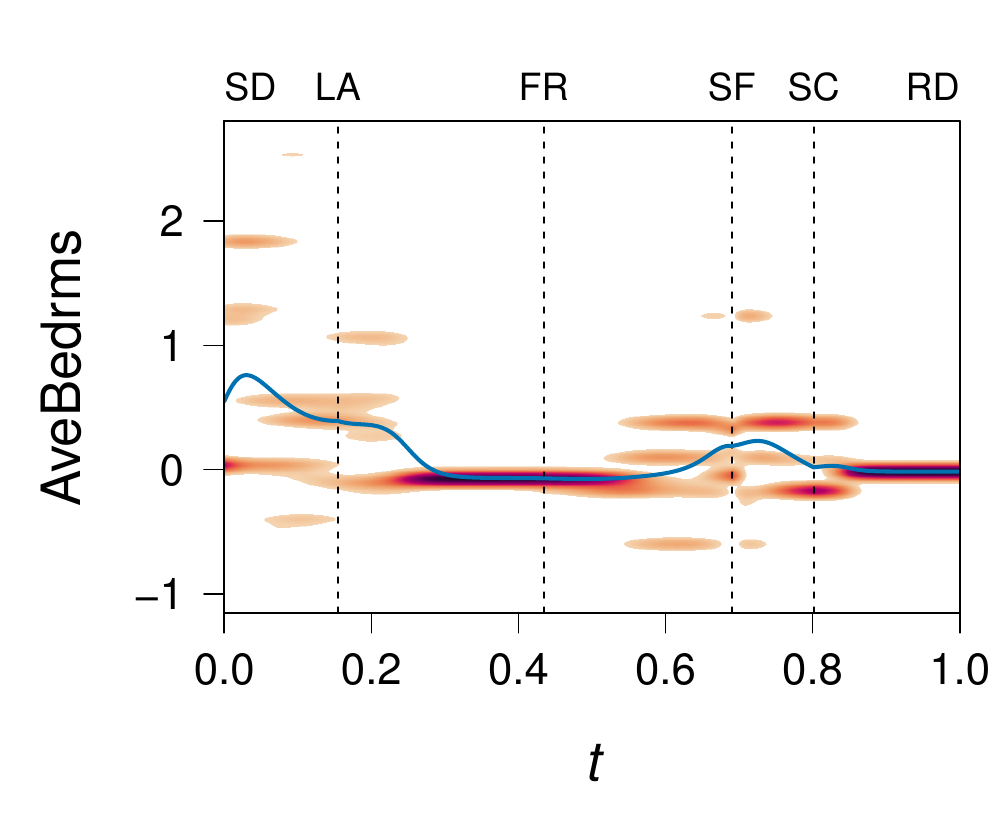}\caption{AveBedrms}\end{subfigure}\\[1ex]
    \begin{subfigure}[b]{0.32\textwidth}\centering\includegraphics[width=\linewidth]{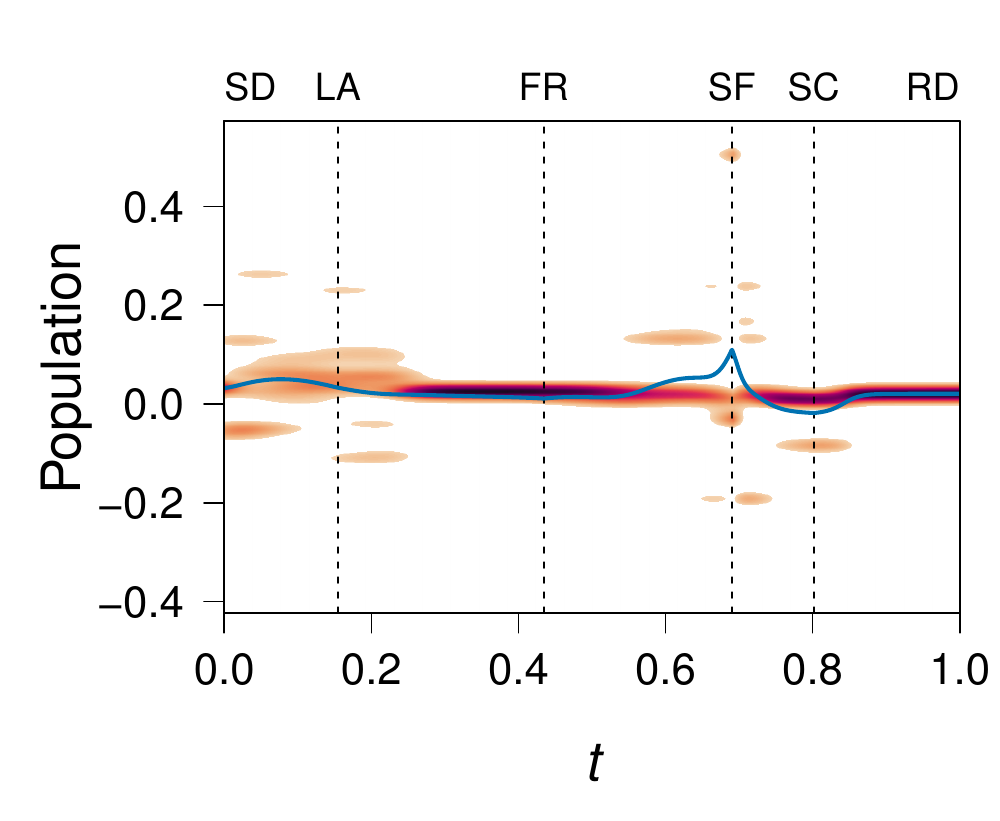}\caption{Population}\end{subfigure}\hspace{0.02\textwidth}
    \begin{subfigure}[b]{0.32\textwidth}\centering\includegraphics[width=\linewidth]{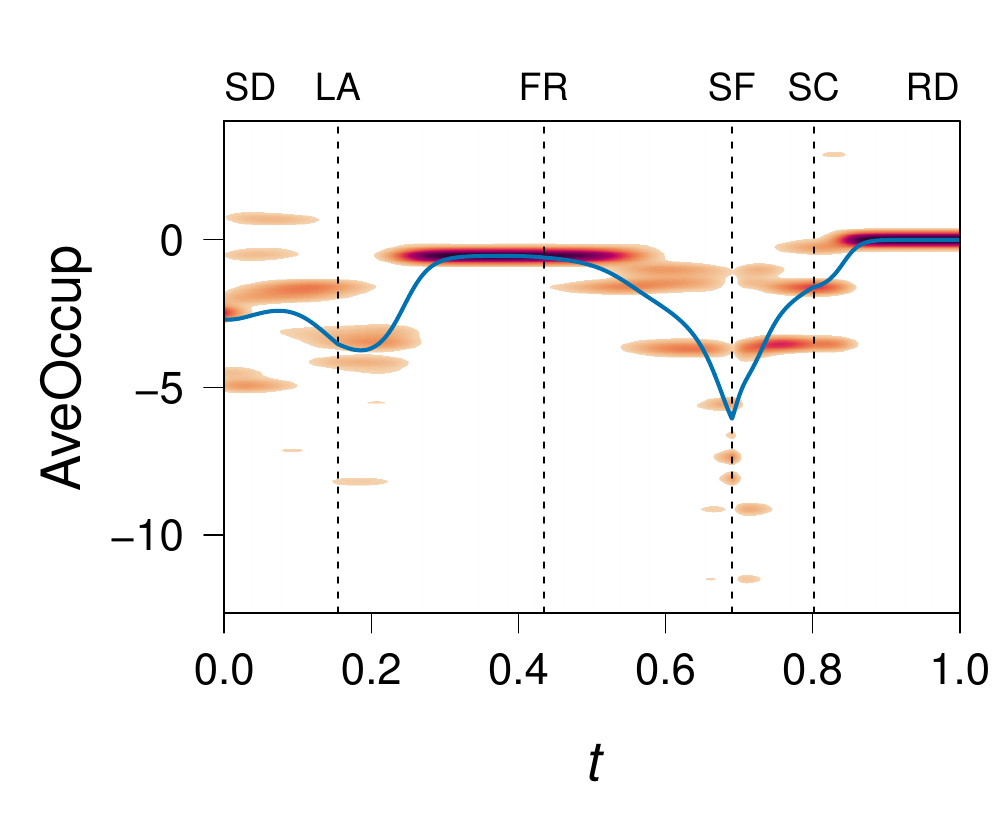}\caption{AveOccup}\end{subfigure}
    \caption{Illustration of geographically varying prior. (a): the path $z(t)$ through the six cities with arrows in the direction of increasing $t$, where SD, LA, FR, SF, SC, and RD stand for San Diego, Los Angeles, Fresno, San Francisco, Sacramento, and Redding. (b)--(h): heat maps of the marginal conditional prior $\hat{G}(\beta_c|z(t))$ of each regression coefficient along the path where the point masses at the $65$ support points are smoothed by a Gaussian kernel and the color scale is linear in the density. The blue line is the conditional mean of the coefficient, the dashed vertical lines represent the cities, and the vertical range is set as to include the support points whose gate exceeds $0.02$ somewhere on the path.}
    \label{fig:housing_path}
\end{figure}

\subsection{NHANES depression score dataset}
\label{sec:nhanes}
Our second application is the analysis of the depression score in the National Health and Nutrition Examination Survey (NHANES). We collect $n=8{,}996$ adults who answered all items of the nine-item Patient Health Questionnaire (PHQ-9) from the 2013--2014 and 2015--2016 cycles.  The response variable is the PHQ-9 total score, which is the sum of nine items scored from $0$ to $3$ and takes values in $\{0,1,\ldots,27\}$. The score of $10$ or above is defined as depression in leading analyses about prevalence of depression from NHANES \citep{Brody2018-nchs} (URL:\url{https://www.cdc.gov/nchs/products/databriefs/db303.htm}).

We use five explanatory variables: age (Age), sex (male), ratio of income to poverty (PIR), hours of sleep (SleepHrs), and body mass index (BMI). We also set these five features as the covariate $z_i\in\mathbb{R}^{5}$. The response is a count with many zeros: $31\%$ of the subjects' scores are $0$. The explanatory variables and covariates are standardized. We apply our method with negative binomial experts with the log link (NB-EXMoE), where the dispersion parameter $\omega$ of the negative binomial distribution, with $\mathrm{Var}(y)=\mu+\mu^{2}/\omega$, is selected by CV. In addition, we also apply the Poisson expert setting (Po-EXMoE), including the covariate-independent case $z_i=0$.

As a result of CV, $\hat{\omega}=2$ and $\lambda=9$ are selected for NB-EXMoE, and the covariate-independent case $z_i=0$ is selected for Po-EXMoE. The numbers of active support points are $\hat{J}=11$ for NB-EXMoE and $\hat{J}=38$ for Po-EXMoE. The $11$ support points of NB-EXMoE are listed in Table \ref{tab:nhanes_clusters}.
Table \ref{tab:nhanes_comparison} shows the comparison of log-likelihood and AUC in test data $n=1{,}799$. We calculate AUC from the classification of depression, determined by the predicted probability $\hat{P}(y_i\geq10|x_i)$. The other compared methods are Poisson regression (Po-GLM), the negative binomial regression (NB-GLM), and the zero-inflated negative binomial regression (ZINB; \cite{Lambert1992-zi}), in which the five covariates are used for the count component and the logit model of the zero inflation. The best is our NB-EXMoE with both log-likelihood $-2.248$ and $\mathrm{AUC}=0.711$. Whereas, the test log-likelihood of Po-EXMoE is largely improved from its one-component version (Po-GLM).
The reason the number of active supports in Po-EXMoE is relatively higher than that of NB-EXMoE may be that the Poisson components tend to approximate overdispersed structure.
We only report the result of NB-EXMoE below, since it shows small active components and best performance.

\begin{table}[htbp]
    \centering
    \caption{Mean log-likelihood of the test data and AUC calculated by the predicted probability of a score of $10$ or above for our methods and competing methods. }
    \label{tab:nhanes_comparison}
    \begin{tabular}{lrr}
        \hline
        Method & log-likelihood & AUC \\
        \hline
        NB-EXMoE  & $-2.248$ & \underline{0.711} \\
        NB-EXMoE, aggregated ($\hat{K}=4$, elbow) & $-2.360$ & 0.646 \\
        NB-EXMoE, aggregated ($\hat{K}=8$, BIC) & $-2.250$ & 0.704 \\
        Po-EXMoE ($z_i=0$) & $-2.254$ & 0.709 \\
        ZINB & $-2.259$ & 0.699 \\
        NB-GLM & $-2.266$ & 0.700 \\
        Po-GLM & $-3.235$ & 0.700 \\
        \hline
    \end{tabular}
\end{table}

Figures \ref{fig:nhanes_prior_intercept} and \ref{fig:nhanes_prior_bmi} show marginal of conditional prior $\hat{G}(\beta_c|z_{\mathrm{sleep}},z_{-\mathrm{sleep}}=0)$, in which SleepHrs $z_{\mathrm{sleep}}$ vary with other fixed covariates 
$z_{-\mathrm{sleep}}$ at sample mean $0$. Since the selected hyperparameter $\lambda=9$ is large, the shown estimated prior structure is rather flat for almost all covariates (summarized in Appendix \ref{app:nhanes_prior} for other covariates). We only pick SleepHrs, which shows a relatively larger heterogeneity structure than other covariates. In Figure \ref{fig:nhanes_prior}, the mass of the intercept moves from the support points with a high baseline to those with a lower baseline as the hours of sleep increase. On the other hand, the mass of BMI is dominant at nearly zero at 8 hours and below. Whereas, the support point at $0.4$ has large mass after $9$ hours and above. These results likely indicate that lower sleep time subgroups indicate a higher score of depression, and BMI does not predict that. For longer sleep time subgroups, BMI predicts a higher depression score.

Next, we discuss the predictive mean conditioned on each sleep time.
Let $\ell$ be the category of SleepHrs, and $\mathcal{Q}(\ell)$ be a subset of test subjects taking $\ell$ for it. Then, we calculate the average conditional mean $\bar{E}_{\ell}= \frac{1}{|\mathcal{Q}(\ell)|}\sum_{i\in \mathcal{Q}(\ell)}\hat{E}[y_i|x_i]$ for NB-EXMoE and the counterparts of the test observation. Figure \ref{fig:nhanes_prediction} shows that 
the predictive mean reproduces the U-shaped pattern of the observed average conditional mean, which is lowest at seven to eight hours of sleep. Although both expert and gates are linear, and no interaction term is included, NB-EXMoE predicts it.

Since the estimated number of support points is $\hat{J}=11$ for this dataset, a similar result could in principle be obtained by a conventional MoE model with linear experts and gates, fitted over a range of candidate numbers of experts $J$. However, our framework also covers more complex situations, where the true conditional distribution is continuous at some $z$, or the true gate function cannot be represented by a single linear gate. The point is that our method estimated this relatively simple structure automatically from the whole class of such complex distributional structures, without fixing the number of experts in advance. Therefore, our method can provide a comprehensive analysis of heterogeneity.

\begin{figure}[htbp]
    \centering
    \begin{subfigure}[b]{0.32\textwidth}\centering\includegraphics[width=\linewidth]{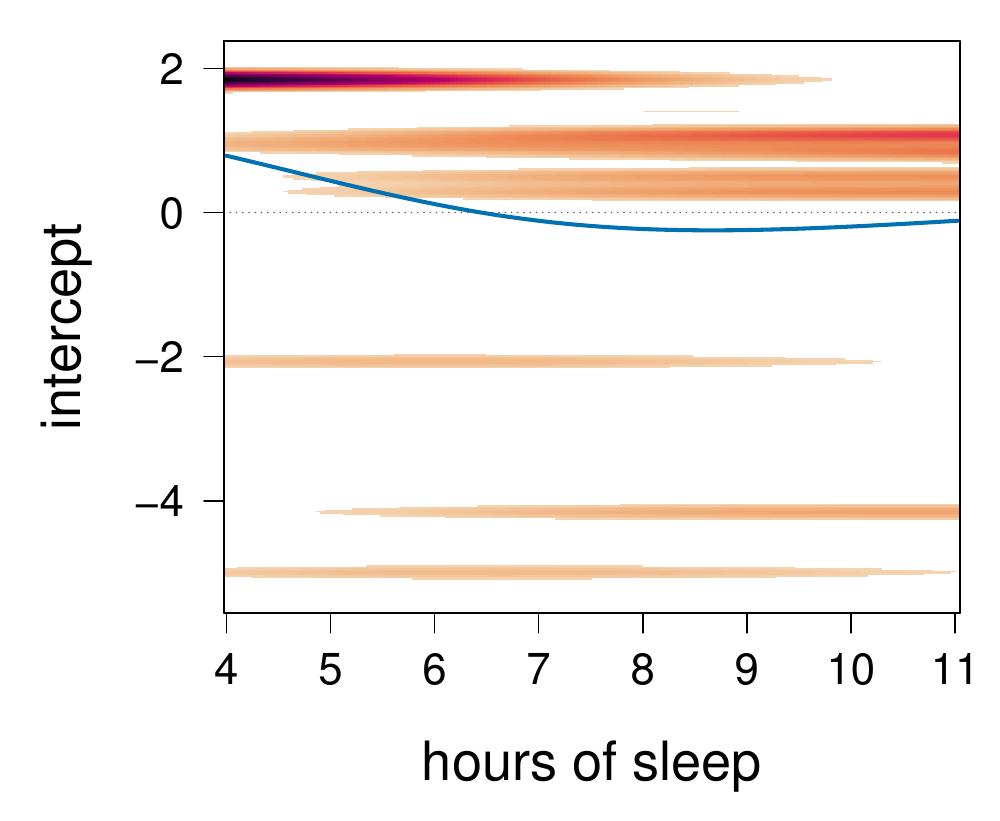}\caption{prior of the intercept}\label{fig:nhanes_prior_intercept}\end{subfigure}\hfill
    \begin{subfigure}[b]{0.32\textwidth}\centering\includegraphics[width=\linewidth]{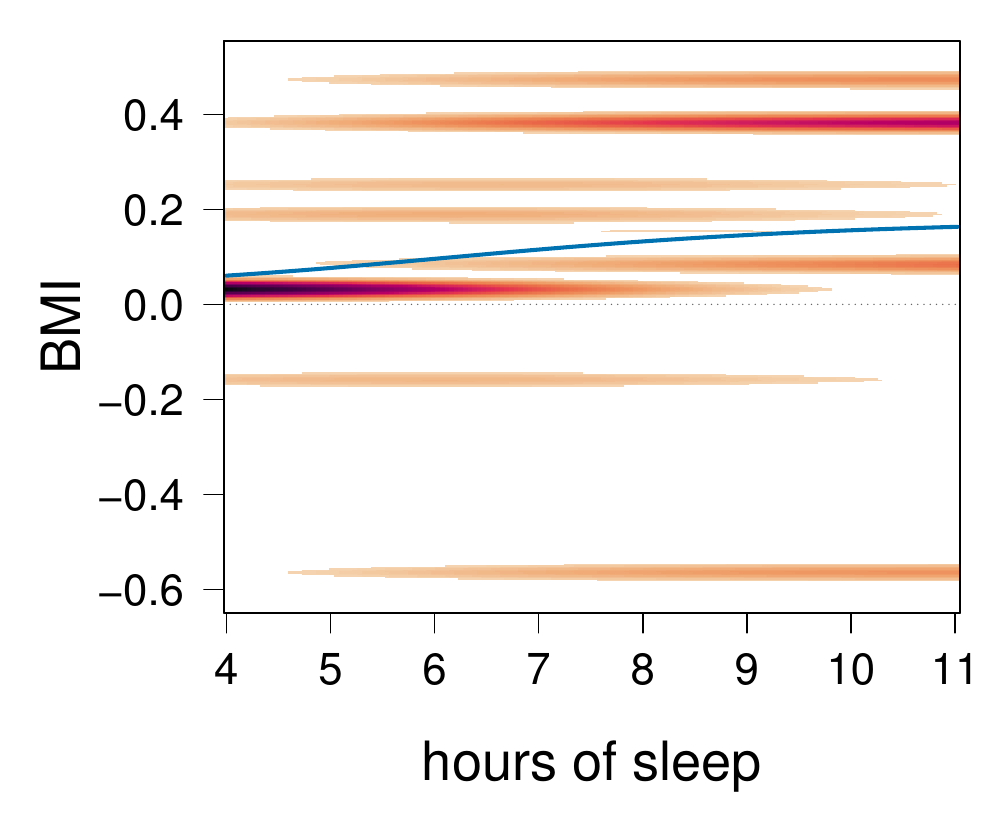}\caption{prior of BMI}\label{fig:nhanes_prior_bmi}\end{subfigure}\hfill
    \begin{subfigure}[b]{0.32\textwidth}\centering\includegraphics[width=\linewidth]{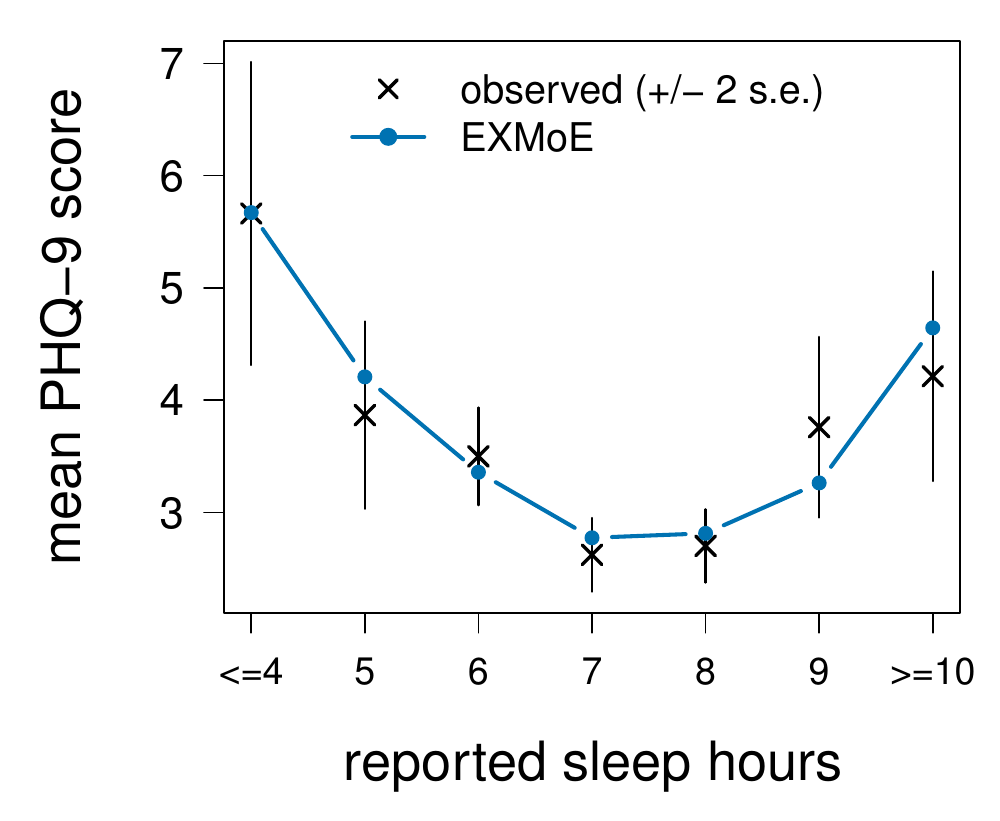}\caption{predictive mean}\label{fig:nhanes_prediction}\end{subfigure}
    \caption{NHANES depression score data along the hours of sleep. (a), (b): heat maps of the marginal conditional prior $\hat{G}(\beta_c|z)$ of the intercept and the coefficient of BMI as the hours of sleep vary with the other covariates fixed at their sample means, where the point masses at the $11$ support points are smoothed by a Gaussian kernel and the blue line is the conditional mean of the coefficient. The other coefficients are shown in Figure \ref{fig:appx_nhanes_sleep}. (c): the observed mean score with $\pm2$ standard errors and the average conditional mean of NB-EXMoE for each hour of sleep on the test data ($n=1{,}799$).}
    \label{fig:nhanes_prior}
\end{figure}

\section{Discussion}
\label{sec:discussion}
In this study, we proposed a nonparametric g-modeling of a covariate-dependent prior measure via NPMLE.
As the estimation class of priors, we introduced the infinite Mixture-of-Expertizing class to avoid ill-posedness. Under this class, the inherited solution precisely yields an MoE model whose number of components is bounded by $2n$. 
We proposed two estimation algorithms, Free-EXMoE and Grid-EXMoE. We also considered aggregating gate functions to obtain a parsimonious solution for simpler interpretation, and the aggregated gate function approximates any continuous gate function in the limit of slopes. 

The simulation studies demonstrate that the EXMoE exhibits a sparse solution, similar to the classical NPMLE problem (e.g., \cite{Polyanskiy2020-go}). Furthermore, the estimated support points of $\beta$ concentrate around a true support point, and those aggregations of linear softmax gates recover the true nonlinear gate function. For an MoE model with an overspecified number of components and linear softmax gates, \cite{nguyen2023demystifying} showed that redundant components concentrate on the true components together with their gate parameters asymptotically. However, our observed phenomenon differs, since the corresponding gates are clearly different among the concentrating components (e.g., Figure \ref{fig:sim_sin_gates_separate}). We expect that the difference comes from the estimation target class. That is, the infinite Mixture-of-Expertizing class is no longer restricted to a conventional linear softmax gate and includes the aggregated gate structure as mentioned in Section \ref{sec:experts_aggregation}.
In this sense, our aggregation of gates may be a reasonable post-hoc procedure when the true gate function has a nonlinear structure. In practice, the trade-off in prediction accuracy via aggregation is problematic, as shown in application to the real datasets. Therefore, development of a theoretically supported aggregation method like post-hoc merging for the MoE model \citep{Hai2025-jk} is an important task.

Through application to real datasets, Free-EXMoE achieves higher prediction accuracy and interpretation of covariate-dependent heterogeneity. Interestingly, the estimated number of support points is substantially different between the two datasets (65 and 11).
Together with the simulation results, this observation indicates that the iMoE-NPMLE behaves differently from an over-specified MoE model. When the data are explained by a few experts, the estimate has few support points, as in the NHANES data. When the gates or the conditional prior are more complex, as in the California housing data, the support points are used to approximate nonlinear gates and continuous prior measures. Viewed as an MoE model, our framework therefore naturally extends the conventional linear softmax-gated MoE to a nonparametric class that contains nonlinear gate functions and infinite mixtures.

From a theoretical perspective, it is still unknown whether our NPMLE converges to the true conditional prior measure asymptotically. Whether the same convergence rates are valid for the fixed component case of MoE is also an interesting topic.
On the other hand, the problem of identifiability remains. When the true prior is discrete, the identifiability is guaranteed from classical discussion in \cite{Jiang1999-tc} for normally used several distributions. However, if the prior is continuous, it is unknown what types of identifiability hold. 

Another practical issue is the interpretation in the high-dimensional case. When the dimension of covariates is greater than three, it is difficult to seek and interpret where the heterogeneity lies in the covariate space from estimated gate functions. Developing a method to solve this problem is essential for practical use.
We leave these topics as future work.

\section*{Acknowledgments}
Akira Okazaki was supported by JSPS KAKENHI Grant Number JP25K24377.
Keisuke Yano was supported by JSPS KAKENHI Grant Numbers
JP24K15120, JP24H00247, JP26K02871, JP26K23255, JP26K00325,
JP26K02869, and by JST CREST Grant Number JPMJCR26E3.
Both authors were supported by MEXT STAR-E NEXT Project (Japan Grant Number JPJ013735).

\bibliographystyle{apalike}
\bibliography{NPMoE}

\appendix
\counterwithin{figure}{section}
\counterwithin{table}{section}

\section{Details of Algorithms}
\label{app:alpha_update}
\subsection{Pairwise Frank-Wolfe with reparametrization in Algorithm \ref{alg:exmoe}}
\label{sec:PFW_repara}

This appendix describes the detail of the pairwise Frank-Wolfe algorithm with reparametrization used as our default update rule $\mathcal{U}_{\Delta}$ for the $\alpha$-block of Algorithm \ref{alg:exmoe}. 

In the following, $\gamma$ and $\beta_j$ are fixed, and we write $u_{ij}\coloneq u_i(\gamma_j)=e^{z_i^{\top}\gamma_j}$, $v_{ij}\coloneq v_i(\beta_{j},\gamma_j)=e^{z_i^{\top}\gamma_j}\psi_{ij}$, $\Lambda_i(\alpha)\coloneq\Lambda_i(H)=\sum_{j=1}^{M}\alpha_j u_{ij}$, and $\Xi_i(\alpha)\coloneq\Xi_i(H)=\sum_{j=1}^{M}\alpha_j v_{ij}$ for $H=\sum_{j=1}^{M}\alpha_j\delta_{(\beta_{j},\gamma_j)}$, so that the negative log-likelihood $\mathcal{L}(\alpha)=-\ell(H)$ and its gradient are
\begin{equation}
    \label{eq:app_alpha_objective}
    \mathcal{L}(\alpha)=\sum_{i=1}^{n}\bigl\{\log \Lambda_i(\alpha)-\log \Xi_i(\alpha)\bigr\},\qquad
    -\frac{\partial\mathcal{L}}{\partial\alpha_j}=D_j(\alpha)=\sum_{i=1}^{n}\left\{\frac{v_{ij}}{\Xi_i(\alpha)}-\frac{u_{ij}}{\Lambda_i(\alpha)}\right\}=D(H;(\beta_{j},\gamma_j)).
\end{equation}
Since $\mathcal{L}(t\alpha)=\mathcal{L}(\alpha)$ for $t>0$, minimizing $\mathcal{L}$ over the simplex is equivalent to minimizing it over the nonnegative orthant. The first-order conditions on the simplex,
\begin{equation}
    \label{eq:app_alpha_kkt}
    D_j(\alpha)=0\ \ (\alpha_j>0),\qquad D_j(\alpha)\leq 0\ \ (\alpha_j=0),
\end{equation}
are the dictionary-restricted version of Theorem \ref{thm:necessary condition}. 

Let $e_j$ be the unit-vector whose $j$-th element is one and other elements are zero. The each unit-vector is viewed as the vertex in the simplex space $\Delta^{M-1}$. The Frank--Wolfe (conditional gradient) method on the simplex space moves toward the vertex with the steepest descent of $\mathcal{L}$ in each iteration. That is, the index of steepest direction is given by $j_F\coloneq\argmax_j D_j(\alpha)$. The update of $\alpha$ to the next value $\alpha^{+}$  is given by
\begin{equation}
    \alpha^{+} = \alpha + td_{\mathrm{FW}}, \qquad d_{\mathrm{FW}} \coloneq e_{j_F}-\alpha,
\end{equation}
where $t$ is a step-size chosen by a line search over $t\in[0,1]$. In each iteration, the current value $\alpha$ is shrunk with scale $1-t$ and the vertex value $te_{j_{F}}$ is added. Thus, Frank-Wolfe method does not induce explicit sparsity. 

To induce sparsity, we use Pairwise Frank-Wolfe method that performs update
\begin{equation}
    \label{eq:PFW}
    \alpha^{+} = \alpha + td_{\mathrm{PFW}}, \qquad d_{\mathrm{PFW}} \coloneq e_{j_F}-e_{j_A},\qquad j_A\coloneq\argmin_{j:\alpha_j>0}D_j(\alpha),
\end{equation}
where $j_{A}$ is an index of the vertex towards Away-step direction, and the step size $t$ is chosen by a line search over $t\in[0,\alpha_{j_A}]$. When $t=\alpha_{j_A}$ is selected, $\alpha_{j_{A}}$ is reduced to exactly zero.
For the present difference-of-convex objective, the line search keeps $\mathcal{L}$ non-increasing. The iteration stops when $\max\{D_{j_F}(\alpha),-D_{j_A}(\alpha)\}$ falls below a tolerance, so that \eqref{eq:app_alpha_kkt} holds up to the tolerance.

In the implementation, the gradient $D_j(\alpha)$ need a modification. Let $\pi_{ij}= \frac{\alpha_ju_{ij}}{\Lambda_{i}(\alpha)}$, $r_{ij}= \frac{\alpha_j v_{ij}}{\Xi_{i}(\alpha)}$, $\Pi_j= \sum_{i=1}^{n}\pi_{ij}$, and $R_j=\sum_{i=1}^{n}r_{ij}$. Then, the gradient written as
\begin{equation}
    D_j(\alpha)= \frac{R_j-\Pi_j}{\alpha_j}.
\end{equation}
When $\alpha_j$ is close to zero, the gradient diverges, making the calculation unstable in some practice. Furthermore, $\alpha_j$ is also tied to probability weight at $z=0$ as the intercept representation $\log \alpha_j$, which would not be natural criteria for selecting active experts. Therefore, we consider reparametrization
\begin{equation}
    \tilde{\alpha}_j\coloneq s_j\alpha_j , \qquad s_j\coloneq n^{-1}\sum_{i=1}^{n}\frac{u_{ij}}{\Lambda_i(\alpha)},
\end{equation}
satisfying $\sum_{j=1}^{M}\tilde{\alpha}_{j}=1$, and $\tilde\alpha_{j}= \bar{g}_{j}$. The $\tilde{\alpha}_{j}$ evaluates the importance of the $j$-th expert by the average of  observation with fixed the current value of $\beta_j$ and $\gamma_{j}$.
In addition, the reparametrized gradient can be written by
\begin{equation}
    \tilde{D}_j (\alpha)\coloneq - \nabla_{\tilde{\alpha}_{j}} \mathcal{L}\left(\frac{\tilde\alpha_j}{s_j}\right)= \frac{D_j(\alpha)}{s_j}= n\left(\frac{R_j}{\Pi_j}-1\right).
\end{equation}

 As a result, in each $\alpha$-block, we fix $s_j$ at the current $\alpha$ and update $\tilde{\alpha}$ by the Pairwise Frank-Wolfe method \eqref{eq:PFW} with the gradient $\tilde{D}_j(\alpha)$ in place of $D_j(\alpha)$ and the step size $t\in[0,\tilde{\alpha}_{j_A}]$. The update stops when $\max\{\tilde{D}_{j_F}(\alpha),-\tilde{D}_{j_A}(\alpha)\}<\tau_{\alpha}$ or after $S_{\alpha}$ steps, and then $\alpha_j=\tilde{\alpha}_j/s_j$ is normalized to the simplex. The tolerance is measured on the scale of $\tilde{D}_j$. The procedure is summarized in Algorithm \ref{alg:pfw_share}. In Algorithm \ref{alg:fixed_dictionary} of the active set refinement, the same procedure is applied to the points in the active set $\mathcal{S}$ and repeated, with $s_j$ recomputed at the start of each repetition, until $\max\{\tilde{D}_{j_F},-\tilde{D}_{j_A}\}<\epsilon_{\mathrm{tol}}$.

\begin{algorithm}[!htp]
    \caption{$\alpha$-update by the pairwise Frank--Wolfe method with the gate-averaged reparametrization}
    \label{alg:pfw_share}
    \KwModel{Current weights $\alpha\in\Delta^{M-1}$; $u_{ij}$ and $v_{ij}$ for $i=1,\ldots,n$ and $j=1,\ldots,M$ at the current $(\gamma,\beta)$}
    \KwTune{$\tau_{\alpha}$, $S_{\alpha}$}
    \KwOut{Updated weights $\alpha\in\Delta^{M-1}$}
    \BlankLine
    Compute $\Lambda_i(\alpha)$ and $s_j=n^{-1}\sum_{i=1}^{n}u_{ij}/\Lambda_i(\alpha)$ for $j=1,\ldots,M$, and set $\tilde{\alpha}_j\leftarrow s_j\alpha_j$, $\tilde{u}_{ij}\leftarrow u_{ij}/s_j$, and $\tilde{v}_{ij}\leftarrow v_{ij}/s_j$\;
    \For{$s=1,\ldots,S_{\alpha}$}{
        Compute $\tilde{D}_j=\sum_{i=1}^{n}\bigl\{\tilde{v}_{ij}/\Xi_i-\tilde{u}_{ij}/\Lambda_i\bigr\}$ for $j=1,\ldots,M$, with $\Lambda_i=\sum_{k}\tilde{\alpha}_k\tilde{u}_{ik}$ and $\Xi_i=\sum_{k}\tilde{\alpha}_k\tilde{v}_{ik}$\;
        Set $j_F\leftarrow\argmax_{j}\tilde{D}_j$ and $j_A\leftarrow\argmin_{j:\tilde{\alpha}_j>0}\tilde{D}_j$\;
        \If{$\max\{\tilde{D}_{j_F},-\tilde{D}_{j_A}\}<\tau_{\alpha}$ or $j_F=j_A$}{\textbf{break}\;}
        Set $t^{\ast}\leftarrow\argmax_{t\in[0,\tilde{\alpha}_{j_A}]}\sum_{i=1}^{n}\bigl\{\log(\Xi_i+t(\tilde{v}_{ij_F}-\tilde{v}_{ij_A}))-\log(\Lambda_i+t(\tilde{u}_{ij_F}-\tilde{u}_{ij_A}))\bigr\}$ by a line search\;
        \If{the line search does not increase the log-likelihood}{\textbf{break}\;}
        Set $\tilde{\alpha}_{j_F}\leftarrow\tilde{\alpha}_{j_F}+t^{\ast}$ and $\tilde{\alpha}_{j_A}\leftarrow\tilde{\alpha}_{j_A}-t^{\ast}$, where $\tilde{\alpha}_{j_A}=0$ exactly if $t^{\ast}=\tilde{\alpha}_{j_A}$\;
    }
    Return $\alpha_j\leftarrow(\tilde{\alpha}_j/s_j)/\sum_{k}(\tilde{\alpha}_k/s_k)$ for $j=1,\ldots,M$\;
\end{algorithm}

\subsection{Details of Algorithm \ref{alg:fixed_dictionary}}

In this section, we provide the detail of the two-stage procedures in Algorighm \ref{alg:fixed_dictionary}.

\subsubsection{MM update for the warm-start}

In Algorithm \ref{alg:fixed_dictionary}, the first stage runs the MM algorithm, which does not increase the negative log-likelihood at each iteration. 

Let  $P_j=\sum_i v_{ij}/\Xi_i$ and $S_j=\sum_i u_{ij}/\Lambda_i$.
For the current fixed $\alpha$, the upper bound of the negative log-likelihood with next update $\alpha^{+}$ is written as
\begin{equation}
    \mathcal{L}(\alpha^{+})=\sum_{i=1}^{n}\{\log\Lambda_i(\alpha^{+})-\log\Xi_i(\alpha^{+}) \}\leq \sum_{i=1}^{n}\left\{ \log\Lambda_i(\alpha)+\frac{\Lambda_i(\alpha^{+})}{\Lambda_i(\alpha)} -1 -\sum_{j=1}^{M} r_{ij}\log \frac{\alpha_j^{+}v_{ij}}{r_{ij}} \right\}.
\end{equation}
Here, we use Jensen's inequality for $-\log \Xi_i(\alpha^{+})=-\log\sum_{j=1}^{M} r_{ij}\frac{\alpha_{j}^{+}v_{ij}}{r_{ij}}$ and inequality $\log x \leq x-1$ for $\log\frac{ \Lambda_{i}(\alpha^{+})}{ \Lambda_{i}(\alpha)}$.
Since the right-hand side is equal to $\mathcal{L}(\alpha)$ with $\alpha^{+}=\alpha$, minimizing it with respect to $\alpha^{+}$ does not increase the negative log-likelihood. The right-hand side terms only depending on $\alpha^{+}$ is written as
\begin{equation}
\sum_{i=1}^{n}\left\{\frac{\Lambda_i(\alpha^{+})}{\Lambda_i(\alpha)}- \sum_{j=1}^{M} r_{ij}\log \alpha_j^{+}\right\} = \sum_{j=1}^{M}(\alpha_j^{+}S_j-R_j\log\alpha_j^{+}).
\end{equation}
Thus, the update of each coordinate is given by

\begin{equation}
    \label{eq:app_mm_update}
    \alpha_j^{+} = \frac{R_j}{S_j} = \alpha_j \frac{P_j}{S_j},\qquad j=1,\ldots,M,
\end{equation} 
 Since it never sets $\alpha_j$ exactly to zero, we use it only for the warm-start.

\subsubsection{Active-set refinement}
\label{app:activeset_refinement}

 After the warm start, the active set is initialized as $\mathcal{S}=\{j:\tilde{\alpha}_j\geq\epsilon_0\}$, where $\tilde{\alpha}_j=s_j\alpha_j$ and $s_j=S_j(\alpha)/n$. On $\mathcal{S}$, the pairwise Frank--Wolfe method is applied to $\tilde{\alpha}$ until $\max_{j\in\mathcal{S}}|\tilde{D}_j(\alpha)|\leq\epsilon_{\mathrm{tol}}$, where $s_j$ is recomputed every $S_{\alpha}$ steps. A support point leaves $\mathcal{S}$ only when the away step sets its weight exactly to zero. Then $\tilde{D}_j$ is computed on the whole dictionary from $\tilde{D}_j/n=P_j(\alpha)/S_j(\alpha)-1$, and the points $\mathcal{T}=\{j\notin\mathcal{S}:\tilde{D}_j>\epsilon_{\mathrm{tol}}\}$ are added with $\tilde{\alpha}_j=\epsilon_0$. Then, $\tilde{\alpha}$ is normalized to lie in the simplex space. These two steps are repeated until $\mathcal{T}$ is empty, that is, until \eqref{eq:app_alpha_kkt} holds up to the tolerance $\epsilon_{\mathrm{tol}}$ on the scale of $\tilde{D}_j$.

\subsubsection{Normalization to preserve numerical stability}
We describe the normalization used for the implemention in practice.
In Algorithm \ref{alg:fixed_dictionary}, $u_{ij}$ and $v_{ij}$ are replaced by
\begin{equation}
    \tilde{u}_{ij}\coloneq\frac{e^{z_i^{\top}\gamma_j-\gamma_{\max}\|z_i\|_1}}{\nu_j},\qquad
    \tilde{v}_{ij}\coloneq\tilde{u}_{ij}\,\psi(y_i;\theta(x_i^{\top}\beta^{(m_j)}),\phi),\qquad
    \nu_j\coloneq\frac{1}{n}\sum_{i=1}^{n}e^{z_i^{\top}\gamma_j-\gamma_{\max}\|z_i\|_1},
\end{equation}
where $\gamma_{\max}$ is the bound of $\Gamma=[-\gamma_{\max},\gamma_{\max}]^{d}$. Since $z_i^{\top}\gamma_j\leq\gamma_{\max}\|z_i\|_1$, the exponent is non-positive, which prevents overflow for large $\|z_i\|$ and for $\gamma_j$ near the boundary of $\Gamma$. This replacement does not change the estimate. The factor $e^{-\gamma_{\max}\|z_i\|_1}$ is common to all the points $j$ for each unit $i$ and cancels in the ratio $\Xi_i/\Lambda_i$, and the division by $\nu_j$ is the change of variables $\alpha'_j\propto \nu_j\alpha_j$ of the weights. Its only effect is that the initial values of the weights in Algorithm \ref{alg:fixed_dictionary} are given for $\alpha'$.

\subsection{Implementation settings in estimation algorithms}
\label{app:tunings}
We summarize our default values of implementation for tuning parameters in our estimation algorithms.

\paragraph{Algorithm \ref{alg:exmoe}}
The default values are $\lambda=5\times10^{-2}$, $\tau_{\alpha}=10^{-5}\,n$, $S_{\alpha}=1000$, $T_{\max}=100$, $\epsilon_{\mathrm{tol}}=10^{-7}$, $R=3$, and $s_{\gamma}=0.3$.

\paragraph{Algorithm \ref{alg:fixed_dictionary}}
The default values are $K_1=1000$, $\epsilon_0=10^{-6}$, $\epsilon_{\mathrm{tol}}=10^{-5}\,n$ (that is, $10^{-5}$ per unit on the scale of $\tilde{D}_j$), $S_{\alpha}=1000$, at most $50$ repetitions of Step 2, and $R=3$, and $\Gamma$ is the box $[-\gamma_{\max},\gamma_{\max}]^{d}$ with $\gamma_{\max}=10$ for standardized $z$.

\section{Supplementary figures for the simulation studies}
\label{app:supp_main}

\begin{figure}[!htb]
    \centering
    \begin{subfigure}[b]{0.31\textwidth}\centering\includegraphics[width=\linewidth]{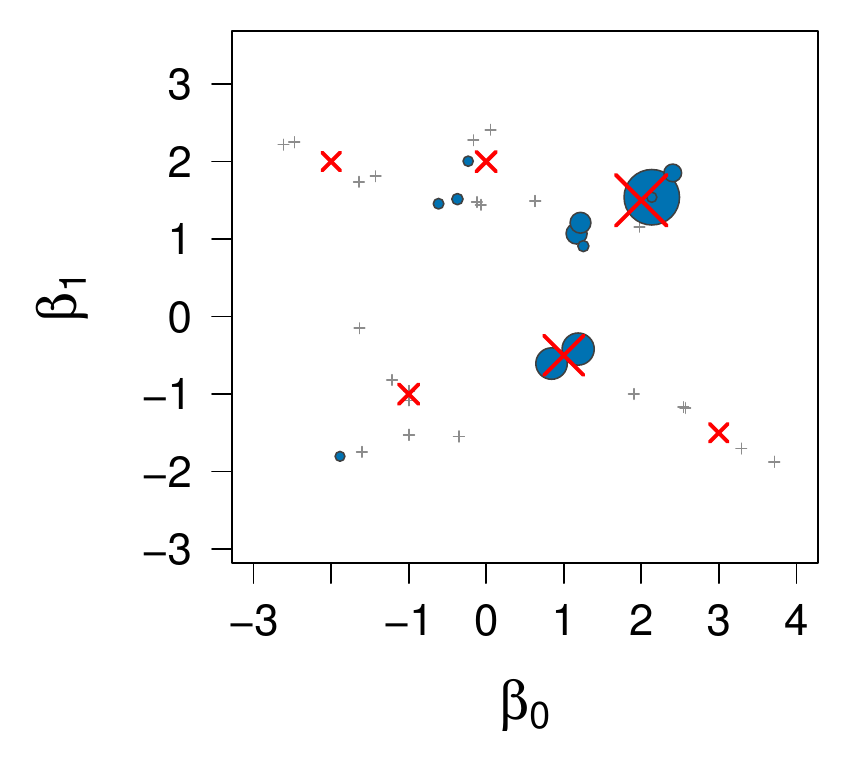}\caption{$s=-2$}\end{subfigure}\hspace{0.04\textwidth}
    \begin{subfigure}[b]{0.31\textwidth}\centering\includegraphics[width=\linewidth]{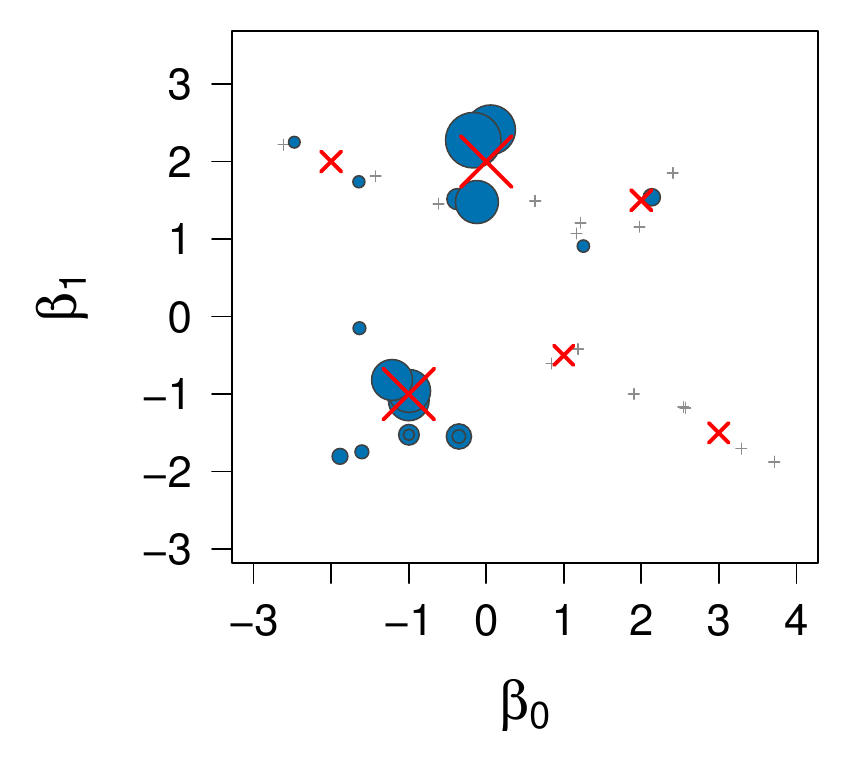}\caption{$s=0$}\end{subfigure}\\[1ex]
    \begin{subfigure}[b]{0.31\textwidth}\centering\includegraphics[width=\linewidth]{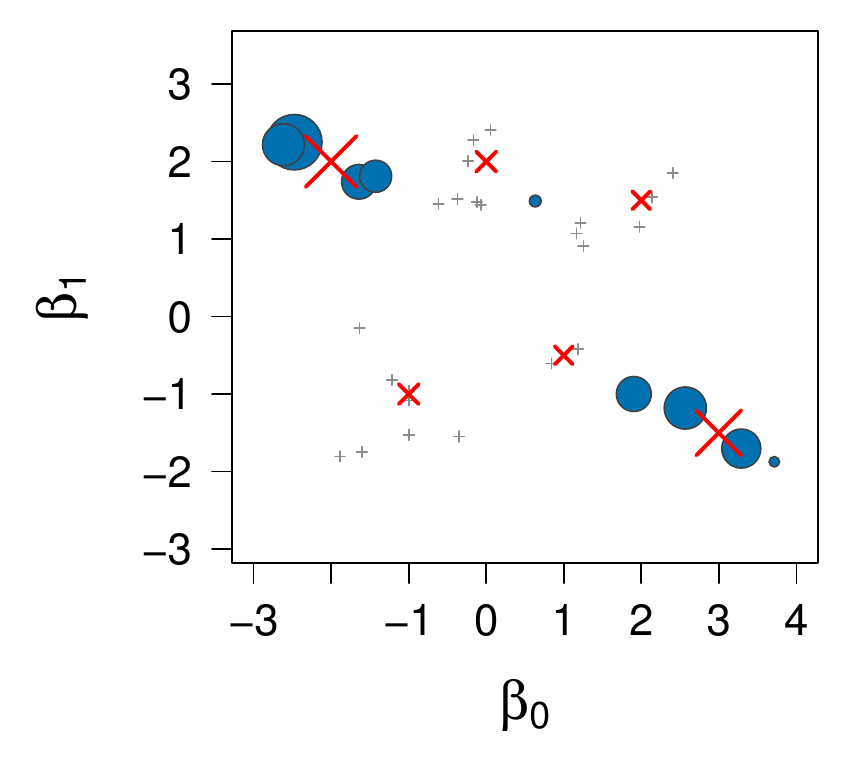}\caption{$s=2$}\end{subfigure}\hspace{0.04\textwidth}
    \begin{subfigure}[b]{0.31\textwidth}\centering\includegraphics[width=\linewidth]{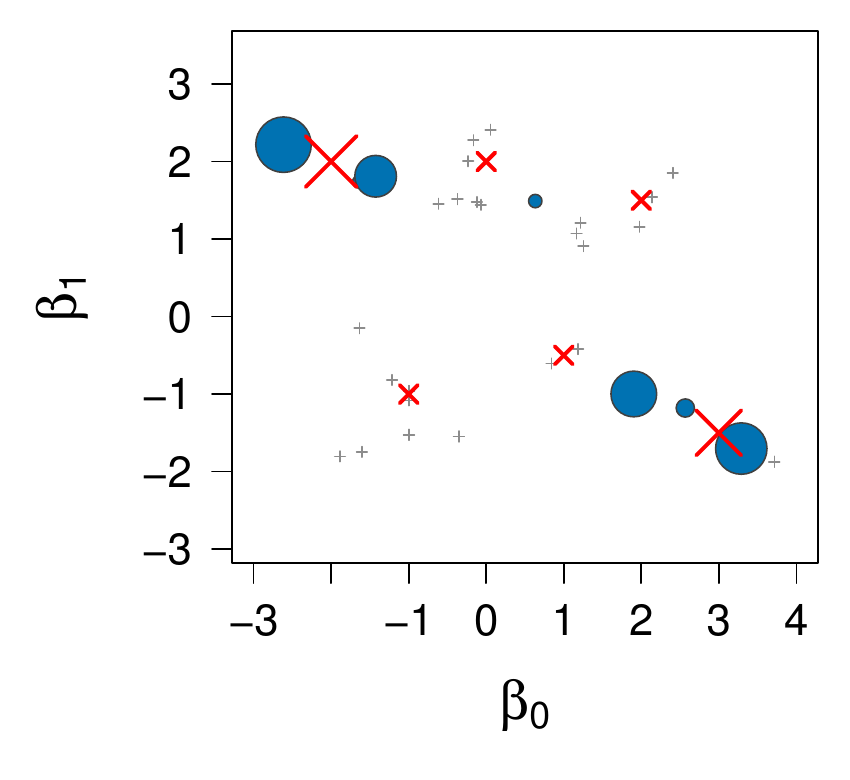}\caption{$s=3$}\end{subfigure}
    \caption{Conditional prior $\hat{G}(\cdot|s)$ estimated by Algorithm \ref{alg:fixed_dictionary} for the smooth Voronoi data at four values of $s$ in the same manner of Figure \ref{fig:sim_sin}.}
    \label{fig:sim_smooth_voronoi}
\end{figure}

\begin{figure}[!htb]
    \centering
    \begin{subfigure}[b]{0.32\textwidth}\centering\includegraphics[width=\linewidth]{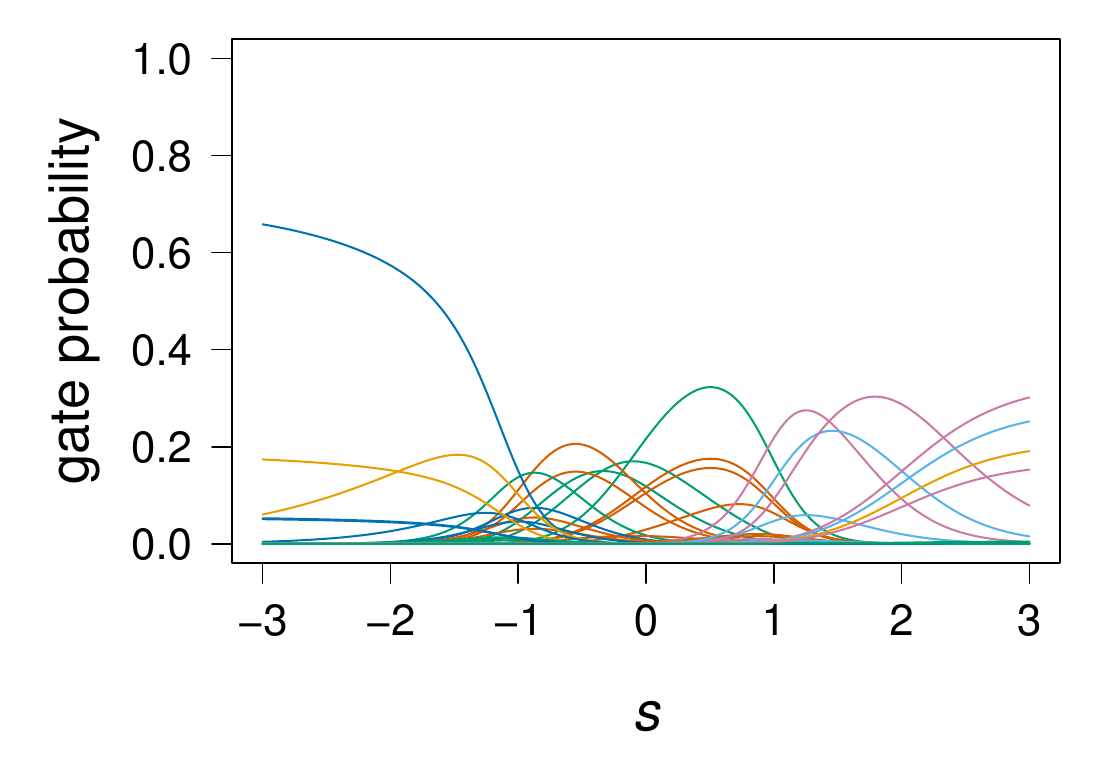}\caption{individual gates}\label{fig:sim_sv_gates_separate}\end{subfigure}\hfill
    \begin{subfigure}[b]{0.32\textwidth}\centering\includegraphics[width=\linewidth]{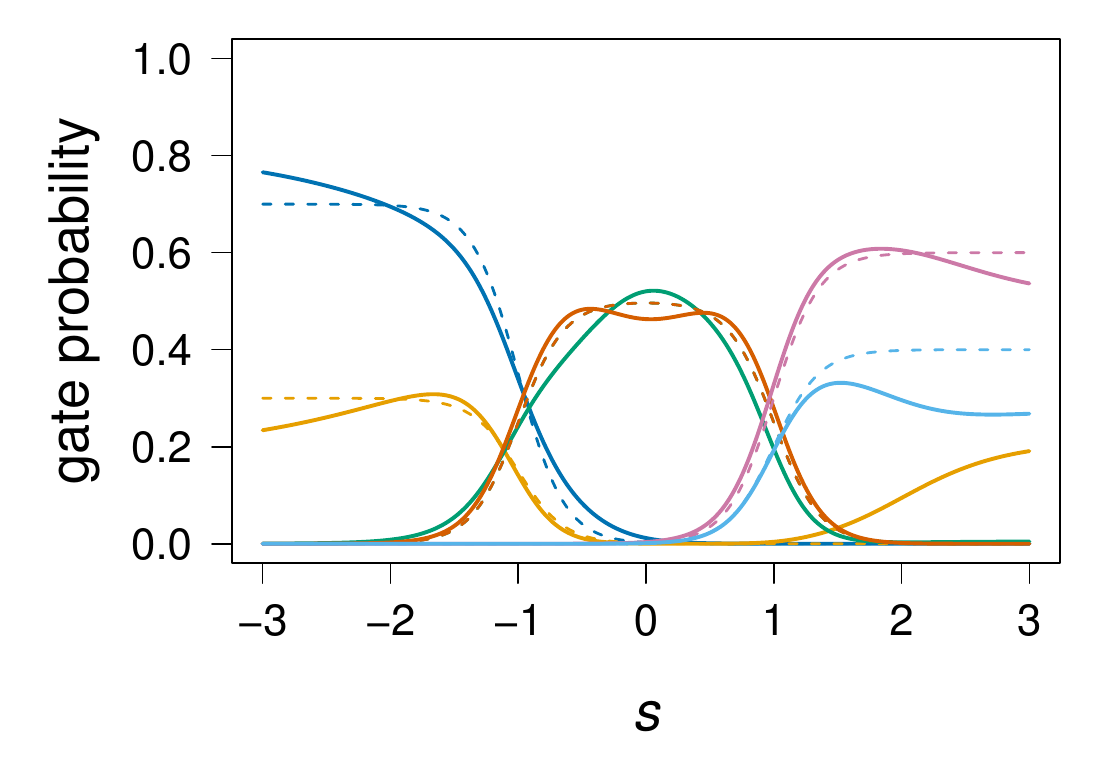}\caption{aggregation by nearest oracle}\label{fig:sim_sv_gates_aggregated}\end{subfigure}\hfill
    \begin{subfigure}[b]{0.32\textwidth}\centering\includegraphics[width=\linewidth]{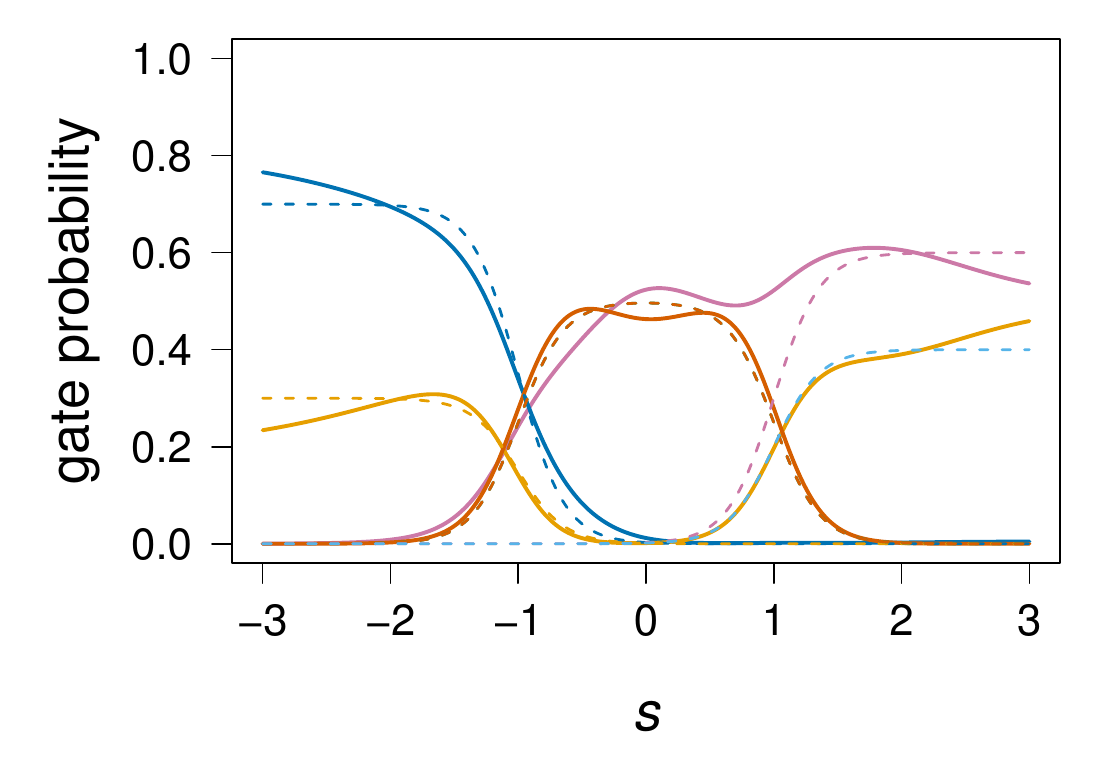}\caption{aggregation by the elbow rule}\label{fig:sim_sv_gates_elbow}\end{subfigure}
    \caption{Gates of the estimate for the smooth Voronoi data in a format similar to Figure \ref{fig:sim_sin_gates}. (a): the individual gates $\hat{g}_j(s)$ of the $41$ estimated support points, colored by the nearest of the six true support points with the colors of Figure \ref{fig:sim_data}. (b): the gates aggregated by the nearest true support point (solid). (c): the aggregated gates $g^{\dagger}_k(s)$ for the $\hat{K}=4$ clusters selected by the elbow method with the weights $\bar{g}_j$ (solid).}
    \label{fig:sim_sv_gates}
\end{figure}

\begin{figure}[!htb]
    \centering
    \begin{subfigure}[b]{0.31\textwidth}\centering\includegraphics[width=\linewidth]{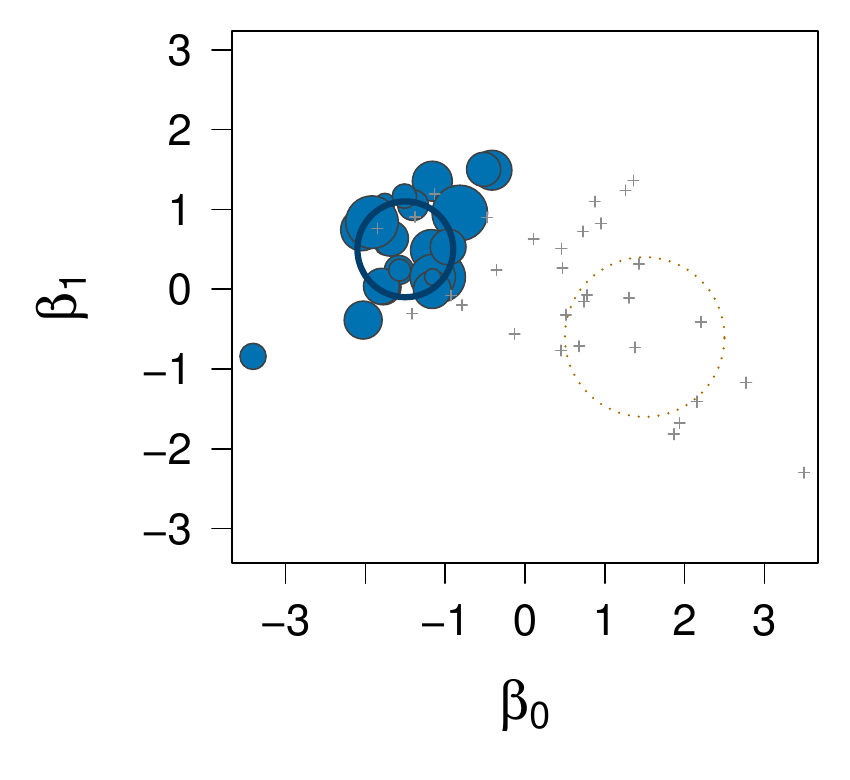}\caption{$s=-2$}\end{subfigure}\hspace{0.04\textwidth}
    \begin{subfigure}[b]{0.31\textwidth}\centering\includegraphics[width=\linewidth]{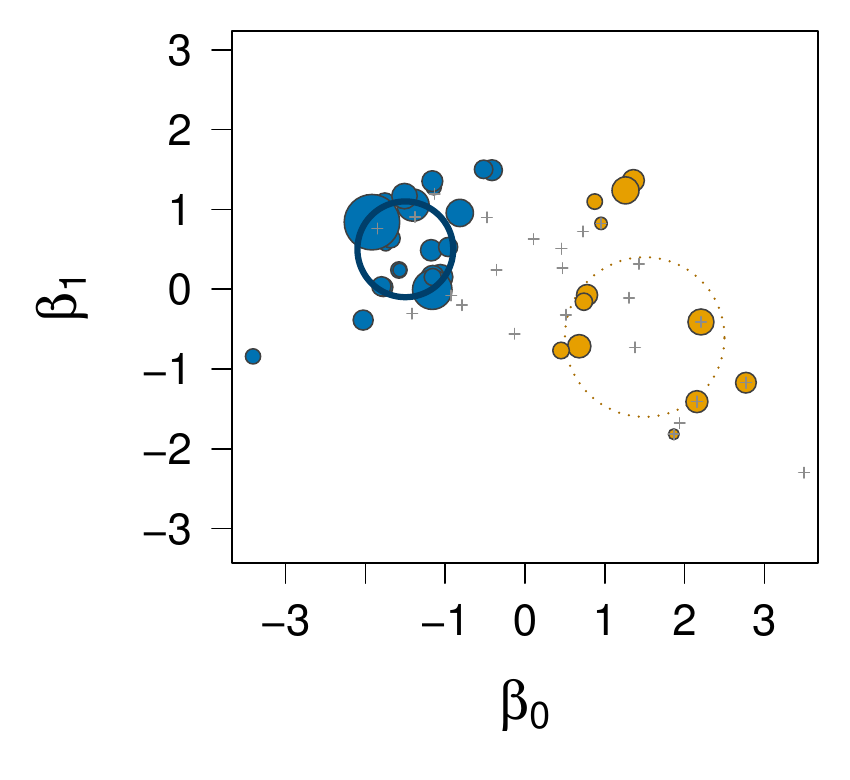}\caption{$s=-0.1$}\end{subfigure}\\[1ex]
    \begin{subfigure}[b]{0.31\textwidth}\centering\includegraphics[width=\linewidth]{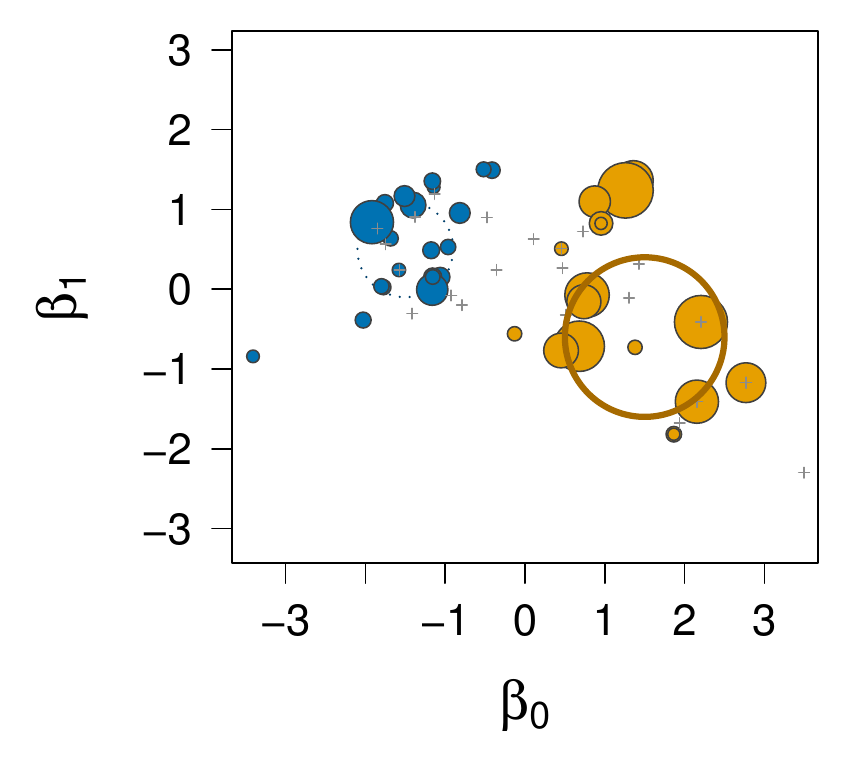}\caption{$s=0.1$}\end{subfigure}\hspace{0.04\textwidth}
    \begin{subfigure}[b]{0.31\textwidth}\centering\includegraphics[width=\linewidth]{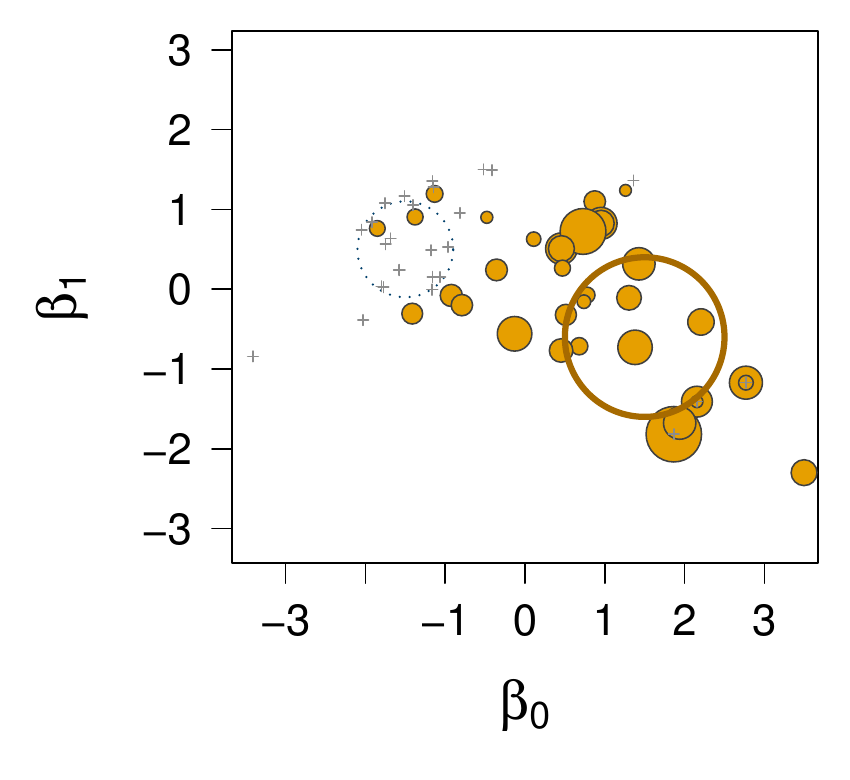}\caption{$s=2$}\end{subfigure}
    \caption{Conditional prior $\hat{G}(\cdot|s)$ estimated by Algorithm \ref{alg:fixed_dictionary} for the switching circle data at four values of $s$ with the same manner of Figure \ref{fig:sim_sin}. As an additional rule, the color of a support point indicates the sign of its gate slope $\hat{\gamma}_j$, blue for negative and orange for positive. The circles are the supports of the true prior \eqref{def:dgp_threshold_circle} in the colors of Figure \ref{fig:sim_true}, blue for $\mathcal{C}(c_A,r_A)$ and orange for $\mathcal{C}(c_B,r_B)$ in darker shades, with the circle active at $s$ drawn as a thick solid line and the other circle as a thin dotted line.}
    \label{fig:sim_threshold_circle}
\end{figure}

\FloatBarrier
\section{Additional results of the real data analyses}
\label{app:additional}
This appendix collects further results of the fits in Section \ref{sec:real_data}.

\subsection{California housing dataset}
Figure \ref{fig:housing_prediction} plots the predictive mean against the observed value on the test blocks, which almost lie along the diagonal under about \$300{,}000 and rather underpredict the upper prices. 
Figure \ref{fig:housing_maps} shows the maps of the observed response, the predictive mean, and the predictive standard deviation 
\begin{equation}
    \widehat{\mathrm{Var}}[y|x,z]^{1/2} = \{\hat{\sigma}^{2}+\mathrm{Var}_{\hat{G}(\cdot|z)}(x^{\top}\beta)\}^{1/2},
\end{equation}
 where the variance of the right-hand side is taken to the pushforward of the estimated prior.
 The predictive mean reproduces the spatial pattern of the observed response. Although the upper prices are underpredicted in Figure \ref{fig:housing_prediction}, this is consistent with the predictive standard deviation, since coastal metropolitan areas with high observed prices also show large predictive standard deviations.

\begin{figure}[htbp]
    \centering
    \includegraphics[width=0.3\linewidth]{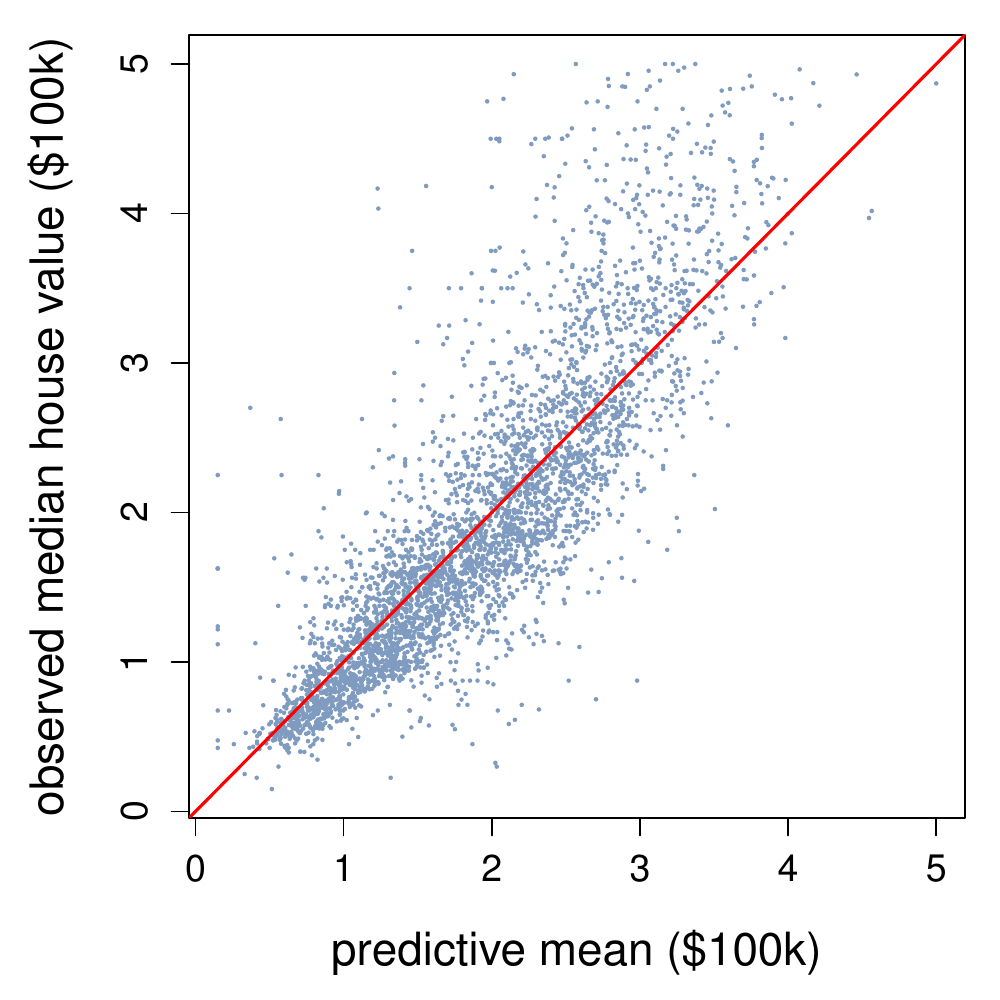}
    \caption{Predictive mean of EXMoE against test data in units of \$100,000.}
    \label{fig:housing_prediction}
\end{figure}

\begin{figure}[htbp]
    \centering
    \begin{subfigure}[b]{0.32\textwidth}\centering\includegraphics[width=\linewidth]{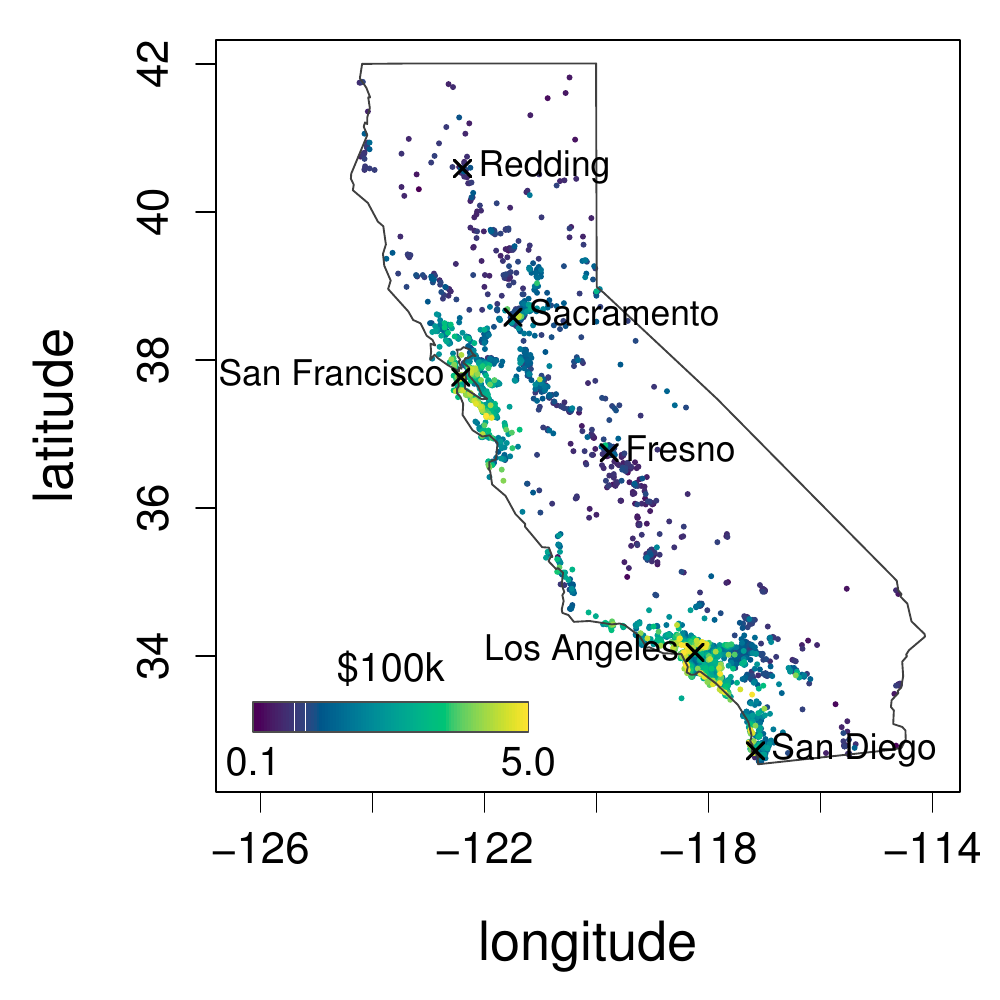}\caption{observed response}\label{fig:housing_map_observed}\end{subfigure}\hfill
    \begin{subfigure}[b]{0.32\textwidth}\centering\includegraphics[width=\linewidth]{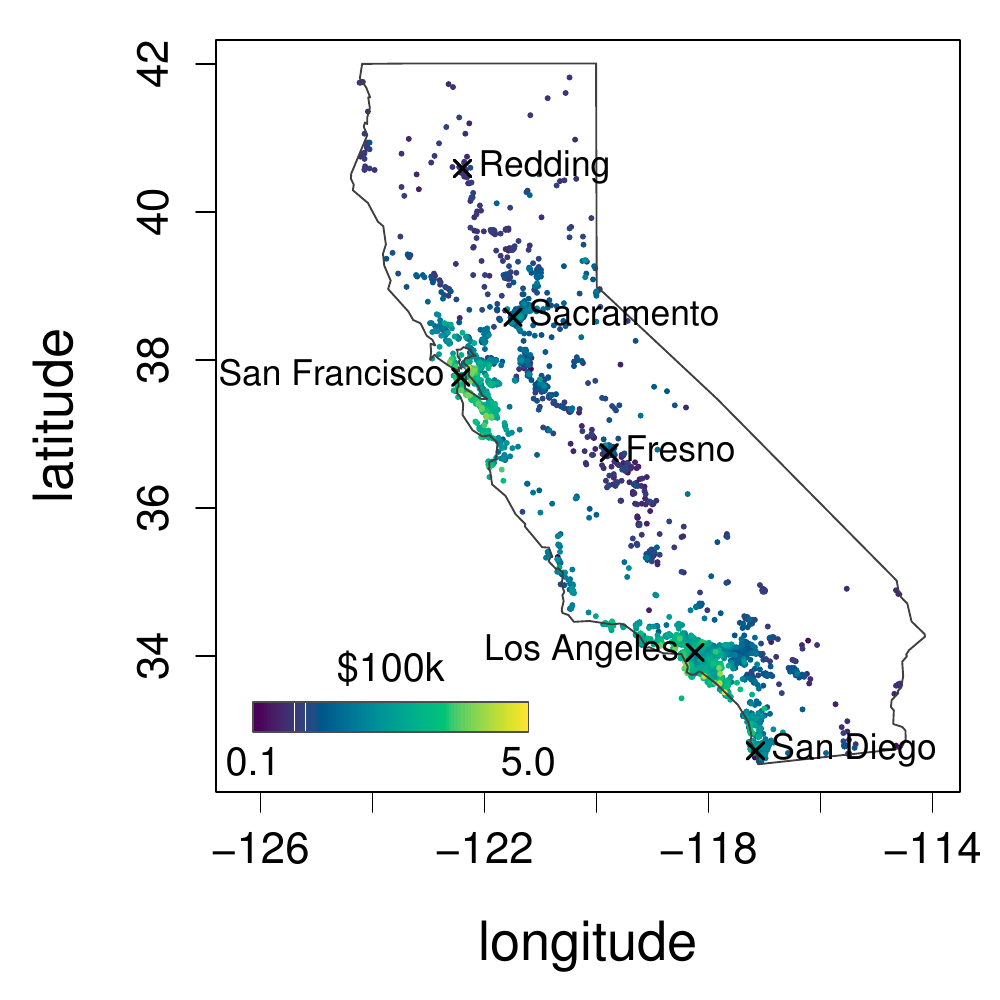}\caption{predictive mean}\label{fig:housing_map_predmean}\end{subfigure}\hfill
    \begin{subfigure}[b]{0.32\textwidth}\centering\includegraphics[width=\linewidth]{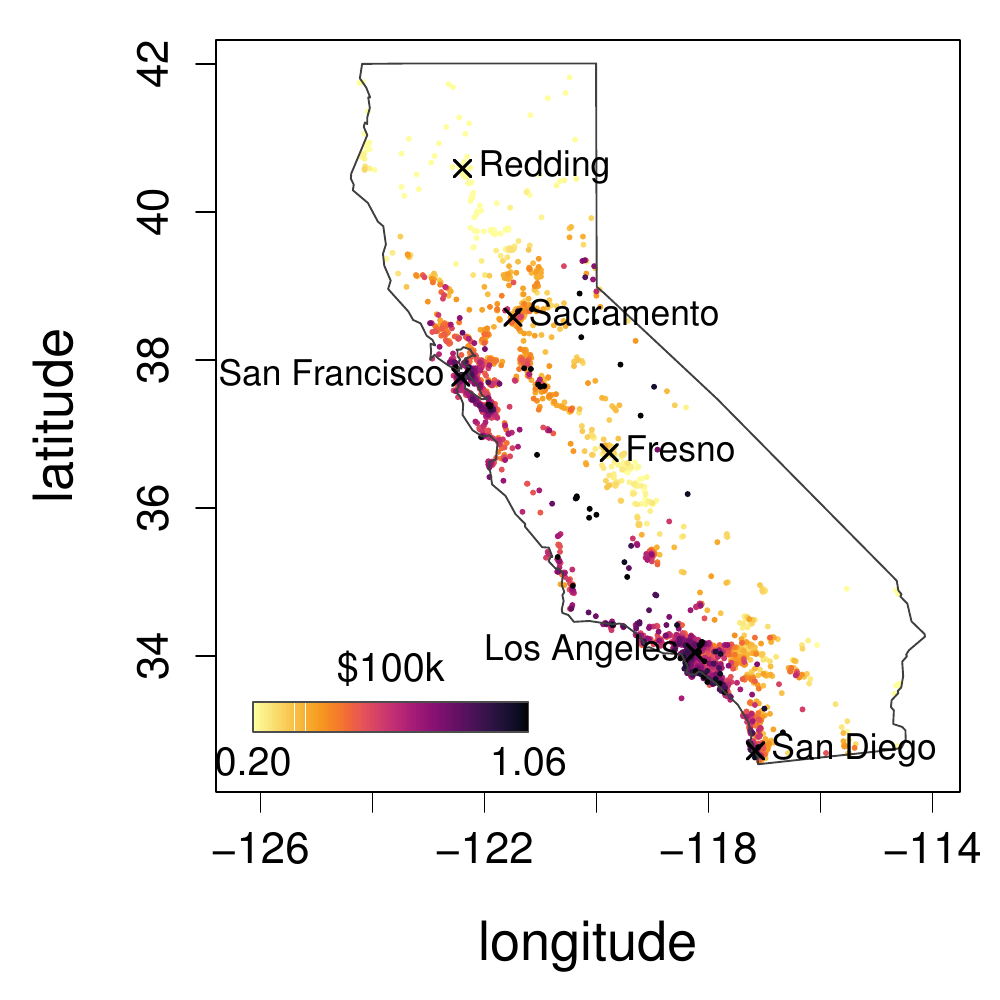}\caption{predictive standard deviation}\label{fig:housing_map_predsd}\end{subfigure}
    \caption{Illustration of prediction of test data ($3{,}935$ blocks) on the California state maps. Each map shows the observed median house value, the predictive mean, and the predictive standard deviation by its color, respectively. Six major cities are marked by crosses. }
    \label{fig:housing_maps}
\end{figure}

\FloatBarrier
\subsection{NHANES depression score dataset}
\label{app:nhanes_prior}

Table \ref{tab:nhanes_clusters} shows estimated $11$ support points with respect to $(\beta,\gamma)$. 
Since the conditional mean of $j$-th expert is given by
\begin{equation}
    \log \mu_i =  \beta_{j0} + \sum_{c=1}^{p-1} x_{ic}\beta_{jc}, 
\end{equation}
we show the regression coefficients on the exponential scale for interpretation.
In particular, $\exp(\hat{\beta}_{j0})$ represents the baseline of the depression score at sample mean $0$. The 6th, 7th, and 9th experts take values close to zero as the baseline, which would be zero-inflated components learned from the data. The $1$ and $10$-th experts show rather high baselines, and the large change of gates is driven by SleepHrs of $1$-st expert. Beyond these results, there is still room for discussion in this analysis.

\begin{table}[htbp]
    \centering
    \caption{NHANES depression score data. The $11$ estimated support points, ordered by the average gate weight $\bar{g}_j$ of \eqref{def:average_gate_weight}. For each support point, the first row gives the regression coefficients on the exponential scale $\exp(\hat{\beta}_j)$, and the second row gives the gate parameters with the intercept representation $\log\hat{\alpha}_j$.}
    \label{tab:nhanes_clusters}
    \scriptsize
    \setlength{\tabcolsep}{3pt}
    \begin{tabular}{rr|lrrrrrr}
    \hline
    $j$ & $\bar{g}_j$ & & intercept & Age & male & PIR & SleepHrs & BMI \\
    \hline
    \multirow{2}{*}{1} & \multirow{2}{*}{0.205} & $\exp(\hat{\beta}_j)$ & 6.213 & 0.976 & 0.956 & 0.895 & 0.954 & 1.030 \\
     & & $(\log\hat{\alpha}_j,\hat{\gamma}_j)$ & 0.000 & 0.199 & $-0.134$ & $-0.483$ & $-0.893$ & 0.048 \\[2pt]
    \multirow{2}{*}{2} & \multirow{2}{*}{0.130} & $\exp(\hat{\beta}_j)$ & 2.902 & 1.173 & 0.885 & 0.822 & 1.377 & 1.459 \\
     & & $(\log\hat{\alpha}_j,\hat{\gamma}_j)$ & $-0.321$ & $-0.058$ & $-0.453$ & 0.026 & 0.343 & $-0.268$ \\[2pt]
    \multirow{2}{*}{3} & \multirow{2}{*}{0.102} & $\exp(\hat{\beta}_j)$ & 1.630 & 0.795 & 0.910 & 0.958 & 1.564 & 0.568 \\
     & & $(\log\hat{\alpha}_j,\hat{\gamma}_j)$ & $-0.656$ & 0.304 & 0.106 & 0.326 & 0.232 & $-0.583$ \\[2pt]
    \multirow{2}{*}{4} & \multirow{2}{*}{0.101} & $\exp(\hat{\beta}_j)$ & 1.307 & 0.617 & 0.998 & 0.793 & 0.674 & 1.601 \\
     & & $(\log\hat{\alpha}_j,\hat{\gamma}_j)$ & $-0.622$ & $-0.421$ & $-0.105$ & 0.499 & 0.253 & 0.035 \\[2pt]
    \multirow{2}{*}{5} & \multirow{2}{*}{0.093} & $\exp(\hat{\beta}_j)$ & 2.528 & 1.396 & 0.776 & 0.791 & 0.905 & 1.207 \\
     & & $(\log\hat{\alpha}_j,\hat{\gamma}_j)$ & $-0.753$ & $-0.161$ & 0.455 & $-0.018$ & $-0.255$ & 0.626 \\[2pt]
    \multirow{2}{*}{6} & \multirow{2}{*}{0.093} & $\exp(\hat{\beta}_j)$ & 0.125 & 0.116 & 0.769 & 1.177 & 0.792 & 0.852 \\
     & & $(\log\hat{\alpha}_j,\hat{\gamma}_j)$ & $-0.893$ & $-0.838$ & 0.207 & $-0.350$ & $-0.268$ & $-0.142$ \\[2pt]
    \multirow{2}{*}{7} & \multirow{2}{*}{0.075} & $\exp(\hat{\beta}_j)$ & 0.016 & 14.930 & 0.687 & 2.217 & 0.670 & 1.468 \\
     & & $(\log\hat{\alpha}_j,\hat{\gamma}_j)$ & $-0.823$ & 0.391 & 0.163 & 0.035 & 0.213 & $-0.202$ \\[2pt]
    \multirow{2}{*}{8} & \multirow{2}{*}{0.071} & $\exp(\hat{\beta}_j)$ & 2.262 & 1.001 & 1.010 & 0.394 & 1.102 & 1.085 \\
     & & $(\log\hat{\alpha}_j,\hat{\gamma}_j)$ & $-1.124$ & 0.455 & $-0.139$ & $-0.278$ & 0.505 & 0.434 \\[2pt]
    \multirow{2}{*}{9} & \multirow{2}{*}{0.071} & $\exp(\hat{\beta}_j)$ & 0.007 & 0.973 & 0.007 & 2.165 & 0.470 & 1.284 \\
     & & $(\log\hat{\alpha}_j,\hat{\gamma}_j)$ & $-0.902$ & 0.490 & $-0.105$ & 0.264 & $-0.175$ & $-0.260$ \\[2pt]
    \multirow{2}{*}{10} & \multirow{2}{*}{0.030} & $\exp(\hat{\beta}_j)$ & 3.984 & 1.135 & 1.315 & 0.948 & 1.504 & 1.163 \\
     & & $(\log\hat{\alpha}_j,\hat{\gamma}_j)$ & $-1.621$ & $-0.007$ & 0.072 & 0.091 & 0.056 & $-0.032$ \\[2pt]
    \multirow{2}{*}{11} & \multirow{2}{*}{0.029} & $\exp(\hat{\beta}_j)$ & 0.600 & 1.946 & 0.828 & 3.720 & 3.641 & 1.091 \\
     & & $(\log\hat{\alpha}_j,\hat{\gamma}_j)$ & $-1.740$ & $-0.354$ & $-0.067$ & $-0.112$ & $-0.013$ & 0.344 \\
    \hline
    \end{tabular}
\end{table}

Figure \ref{fig:appx_nhanes_sleep} shows the marginal conditional prior $\hat{G}(\beta_c|z)$ of the coefficients of the fit of Section \ref{sec:nhanes} along the hours of sleep that are not shown in Figure \ref{fig:nhanes_prior}. Figures \ref{fig:appx_nhanes_age}--\ref{fig:appx_nhanes_bmi} show the marginal conditional prior of every regression coefficient along the gate covariates other than the hours of sleep in the same manner. In each figure one covariate varies over its $1\%$--$99\%$ range, or over both levels for the sex, while the other standardized covariates are fixed at sample mean 0.

\begin{figure}[H]
    \centering
    \begin{subfigure}[b]{0.24\textwidth}\centering\includegraphics[width=\linewidth]{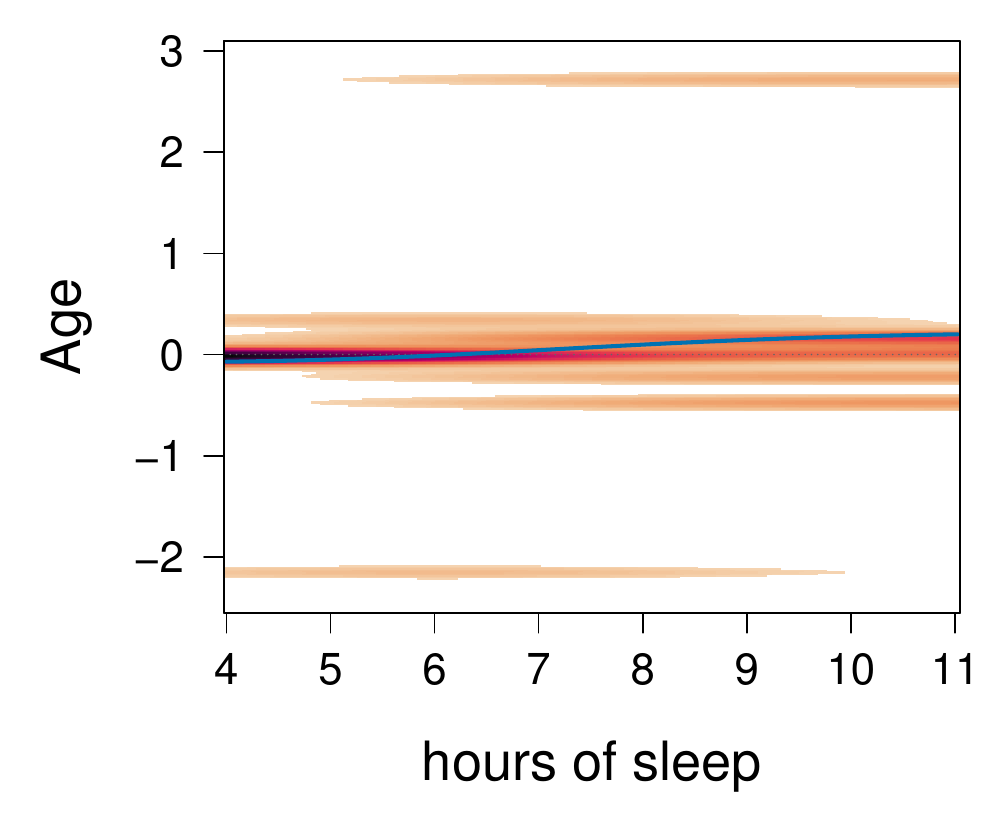}\caption{Age}\end{subfigure}\hfill
    \begin{subfigure}[b]{0.24\textwidth}\centering\includegraphics[width=\linewidth]{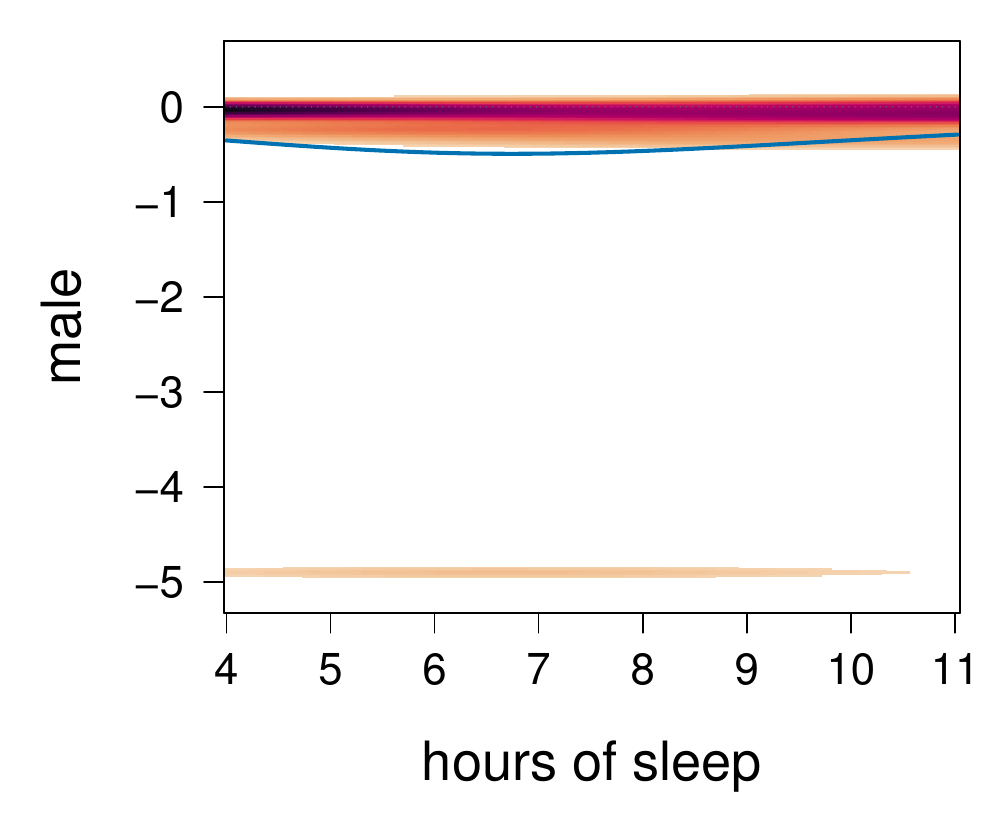}\caption{male}\end{subfigure}\hfill
    \begin{subfigure}[b]{0.24\textwidth}\centering\includegraphics[width=\linewidth]{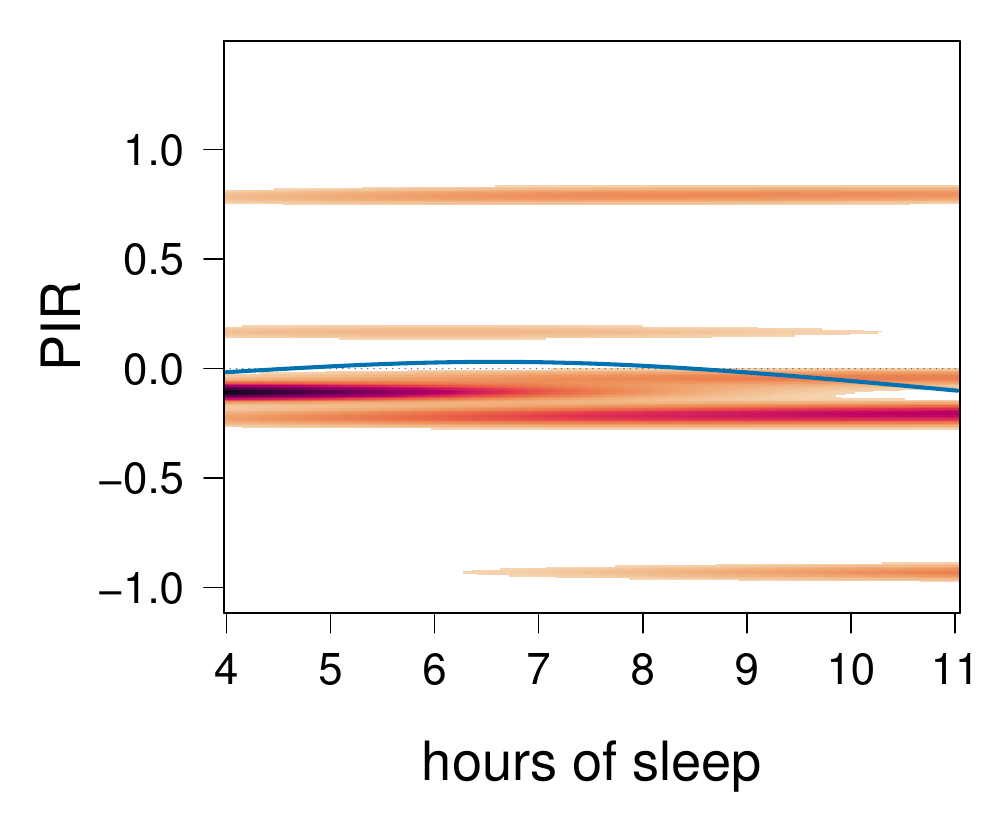}\caption{PIR}\end{subfigure}\hfill
    \begin{subfigure}[b]{0.24\textwidth}\centering\includegraphics[width=\linewidth]{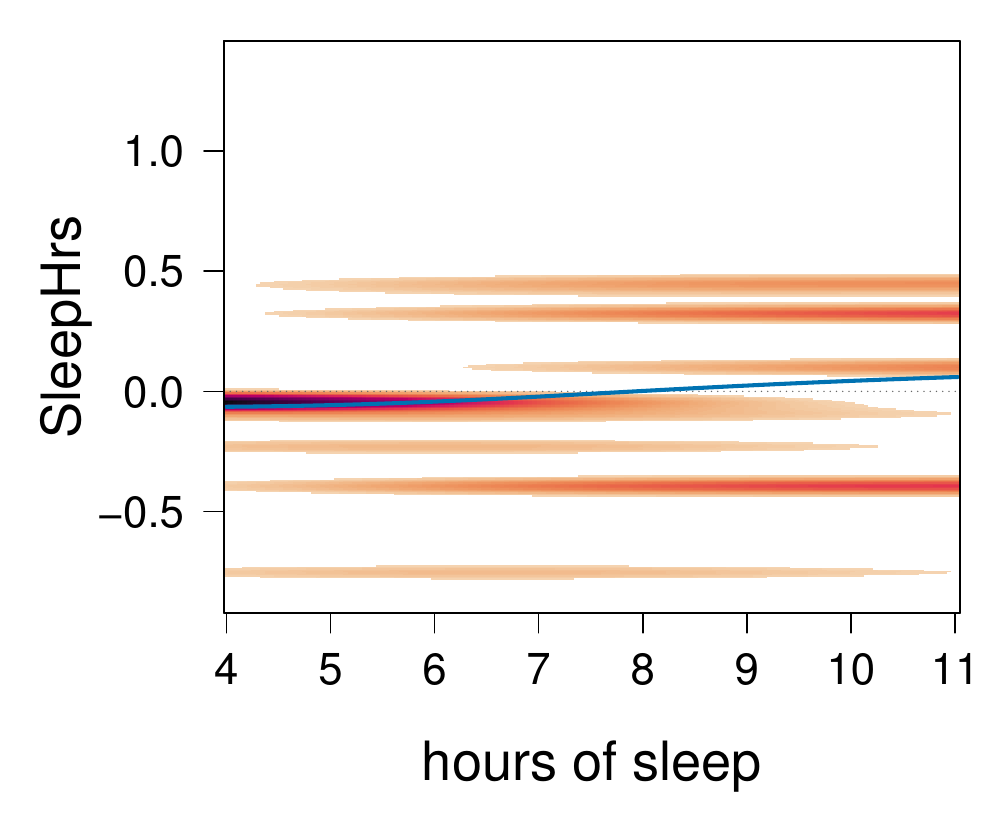}\caption{SleepHrs}\end{subfigure}
    \caption{NHANES depression score data. Marginal conditional prior $\hat{G}(\beta_c|z)$ of the coefficients of Age, male, PIR, and SleepHrs as the hours of sleep vary, in the same manner as Figure \ref{fig:nhanes_prior}.}
    \label{fig:appx_nhanes_sleep}
\end{figure}

\begin{figure}[H]
    \centering
    \begin{subfigure}[b]{0.32\textwidth}\centering\includegraphics[width=\linewidth]{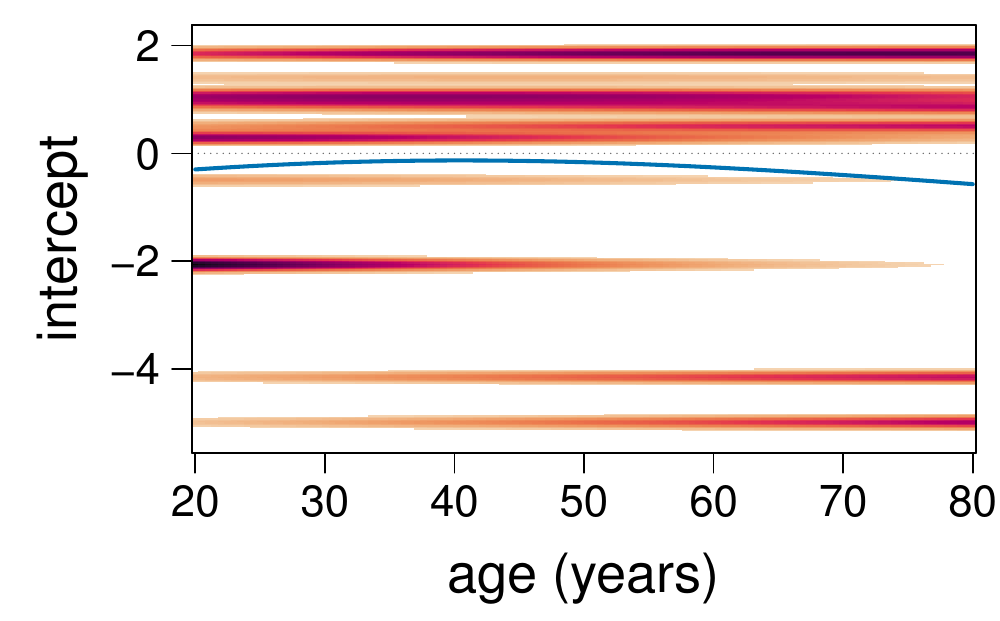}\caption{intercept}\end{subfigure}\hfill
    \begin{subfigure}[b]{0.32\textwidth}\centering\includegraphics[width=\linewidth]{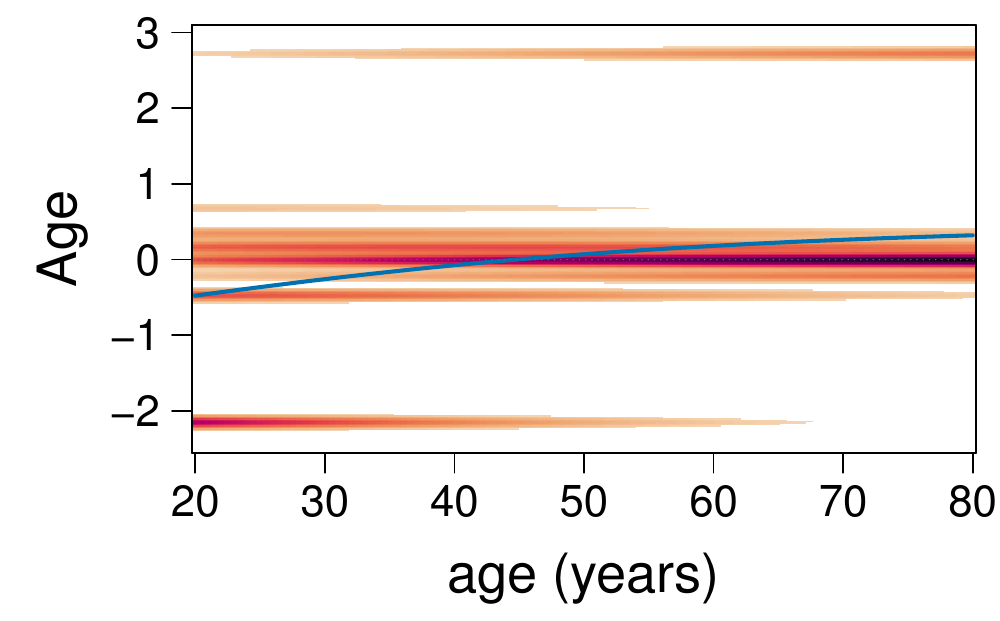}\caption{Age}\end{subfigure}\hfill
    \begin{subfigure}[b]{0.32\textwidth}\centering\includegraphics[width=\linewidth]{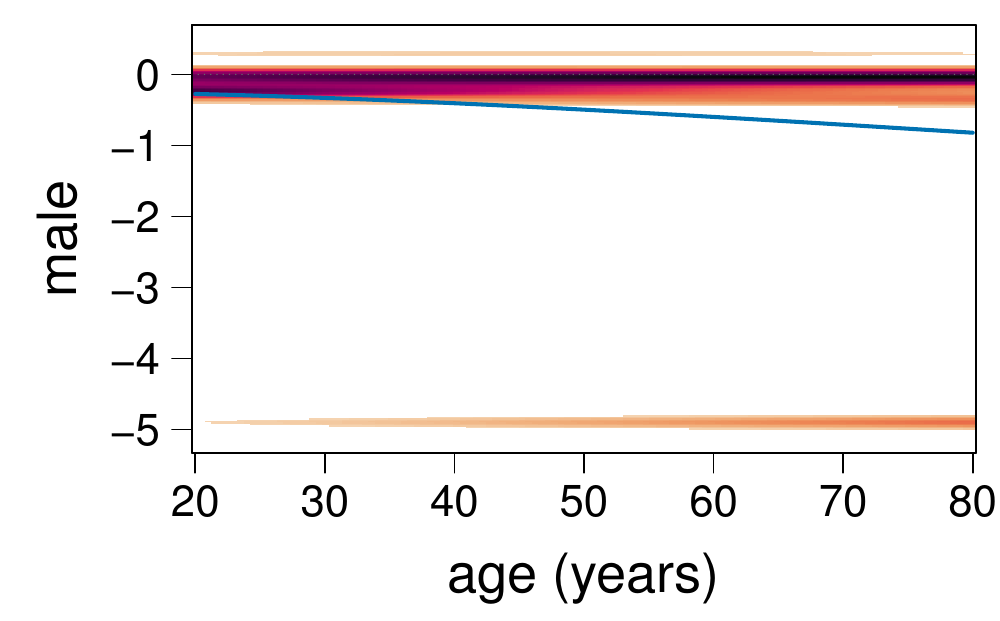}\caption{male}\end{subfigure}\\[1ex]
    \begin{subfigure}[b]{0.32\textwidth}\centering\includegraphics[width=\linewidth]{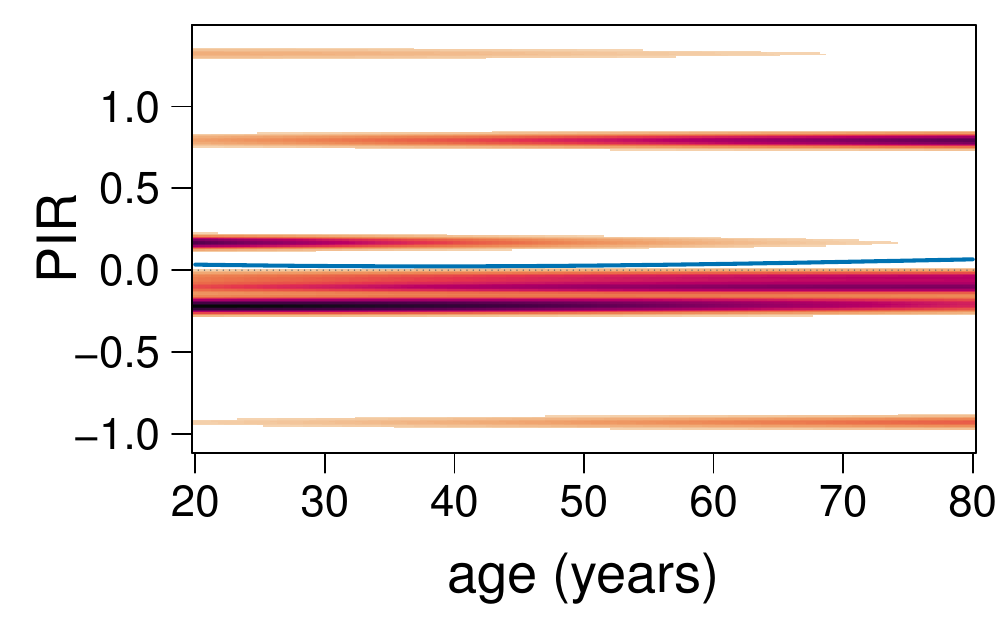}\caption{PIR}\end{subfigure}\hfill
    \begin{subfigure}[b]{0.32\textwidth}\centering\includegraphics[width=\linewidth]{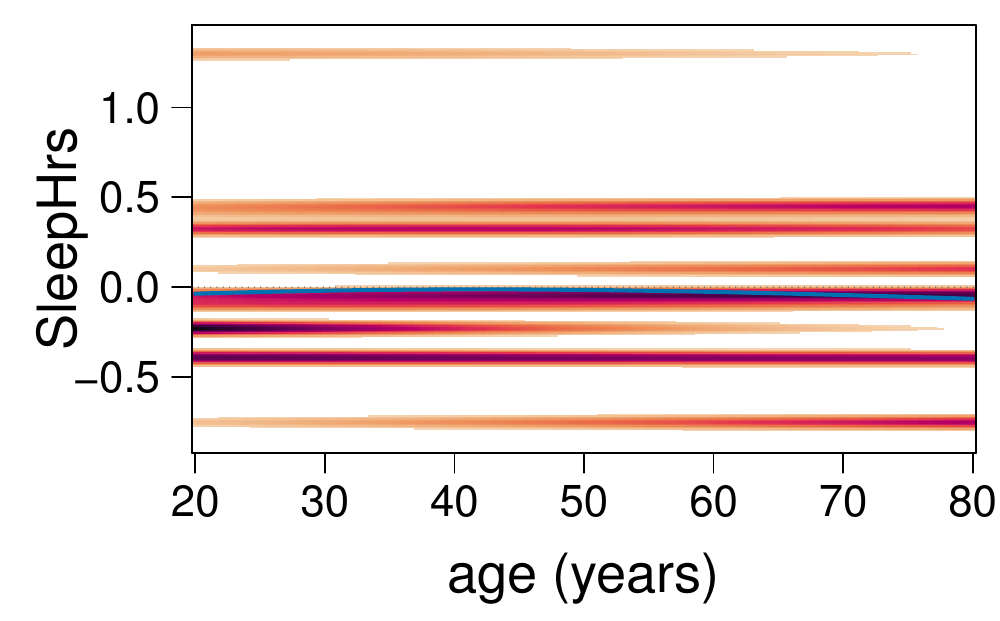}\caption{SleepHrs}\end{subfigure}\hfill
    \begin{subfigure}[b]{0.32\textwidth}\centering\includegraphics[width=\linewidth]{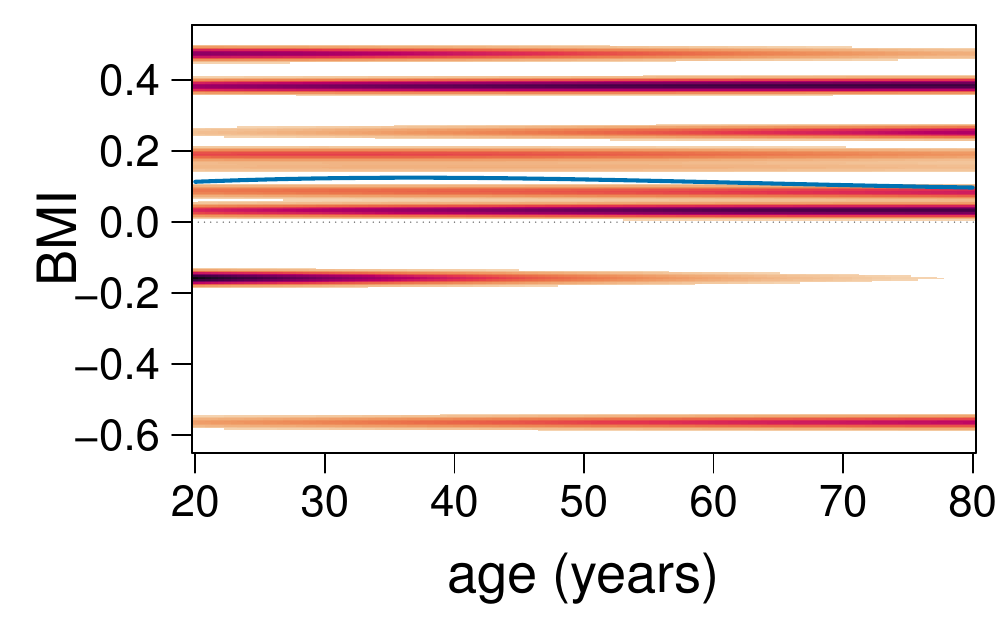}\caption{BMI}\end{subfigure}
    \caption{Conditional prior along the age.}
    \label{fig:appx_nhanes_age}
\end{figure}

\begin{figure}[H]
    \centering
    \begin{subfigure}[b]{0.32\textwidth}\centering\includegraphics[width=\linewidth]{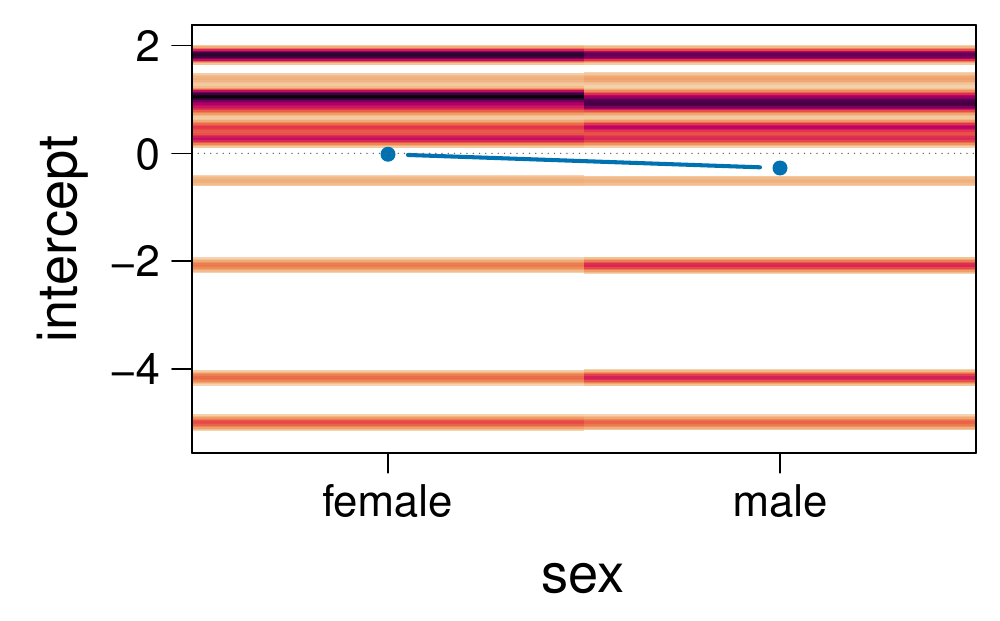}\caption{intercept}\end{subfigure}\hfill
    \begin{subfigure}[b]{0.32\textwidth}\centering\includegraphics[width=\linewidth]{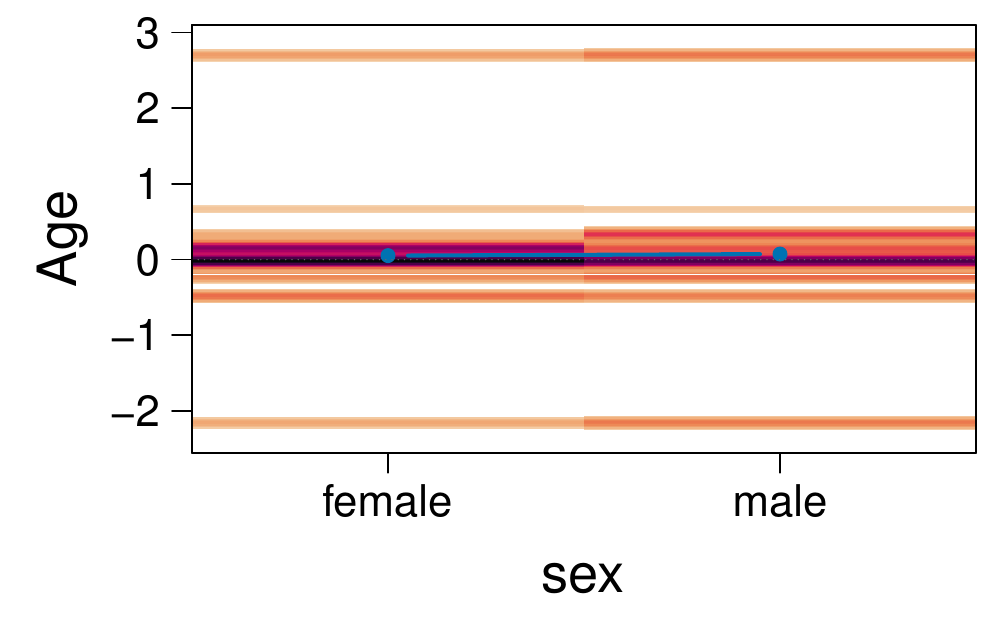}\caption{Age}\end{subfigure}\hfill
    \begin{subfigure}[b]{0.32\textwidth}\centering\includegraphics[width=\linewidth]{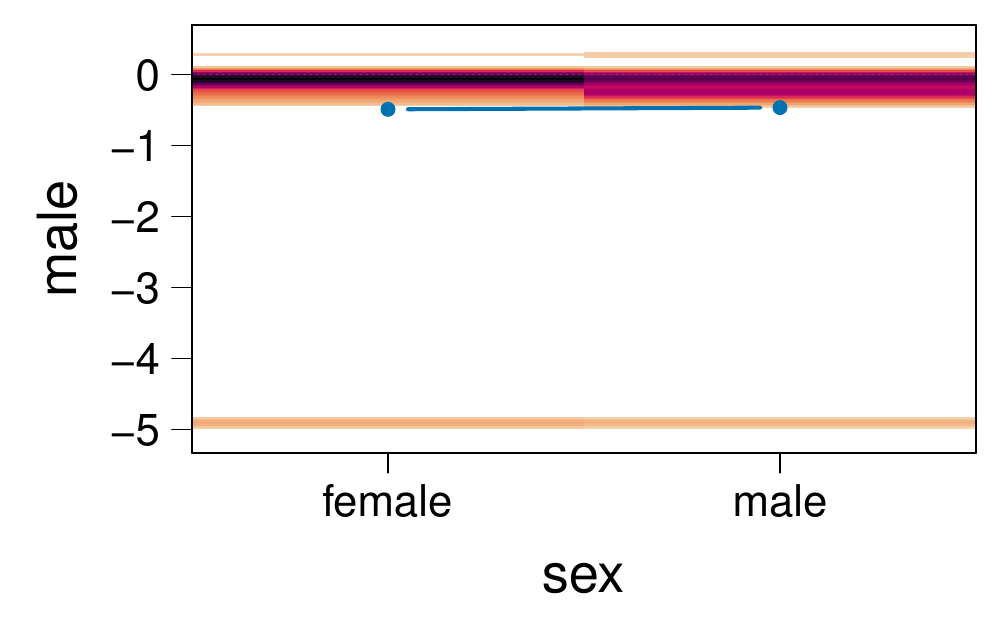}\caption{male}\end{subfigure}\\[1ex]
    \begin{subfigure}[b]{0.32\textwidth}\centering\includegraphics[width=\linewidth]{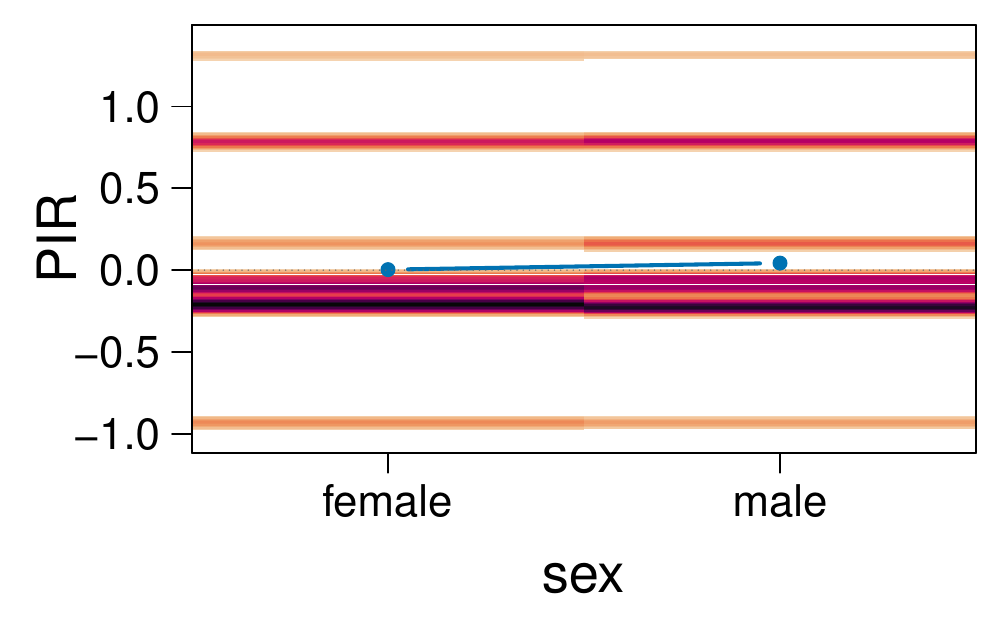}\caption{PIR}\end{subfigure}\hfill
    \begin{subfigure}[b]{0.32\textwidth}\centering\includegraphics[width=\linewidth]{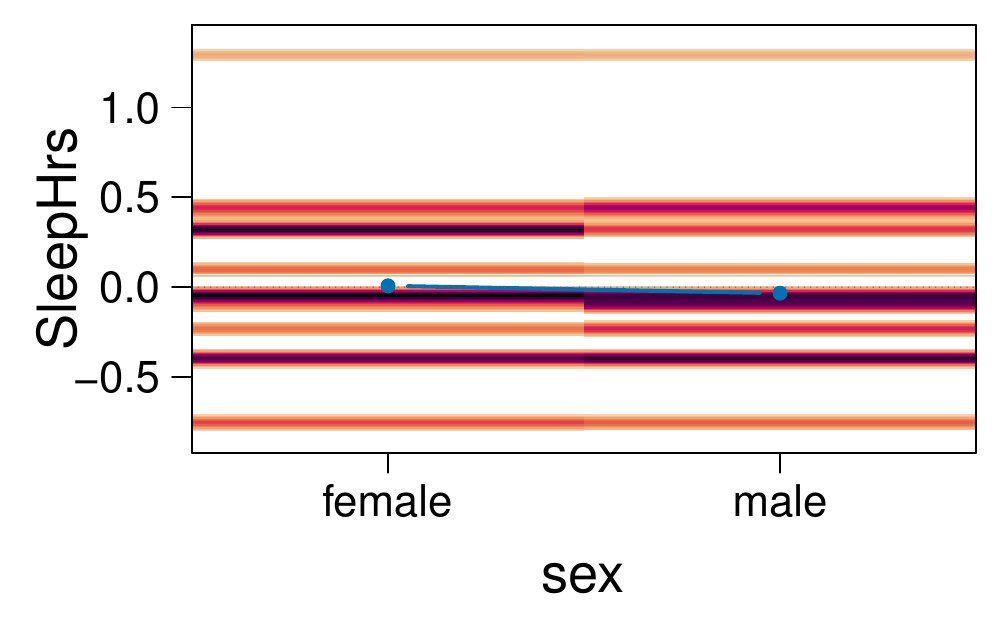}\caption{SleepHrs}\end{subfigure}\hfill
    \begin{subfigure}[b]{0.32\textwidth}\centering\includegraphics[width=\linewidth]{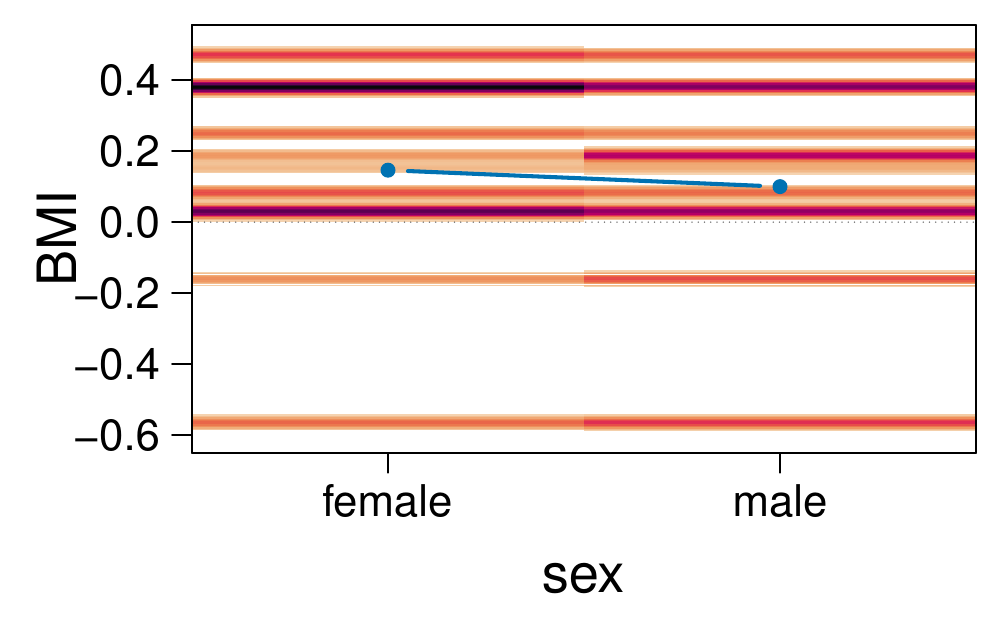}\caption{BMI}\end{subfigure}
    \caption{Conditional prior along the sex.}
    \label{fig:appx_nhanes_male}
\end{figure}

\begin{figure}[H]
    \centering
    \begin{subfigure}[b]{0.32\textwidth}\centering\includegraphics[width=\linewidth]{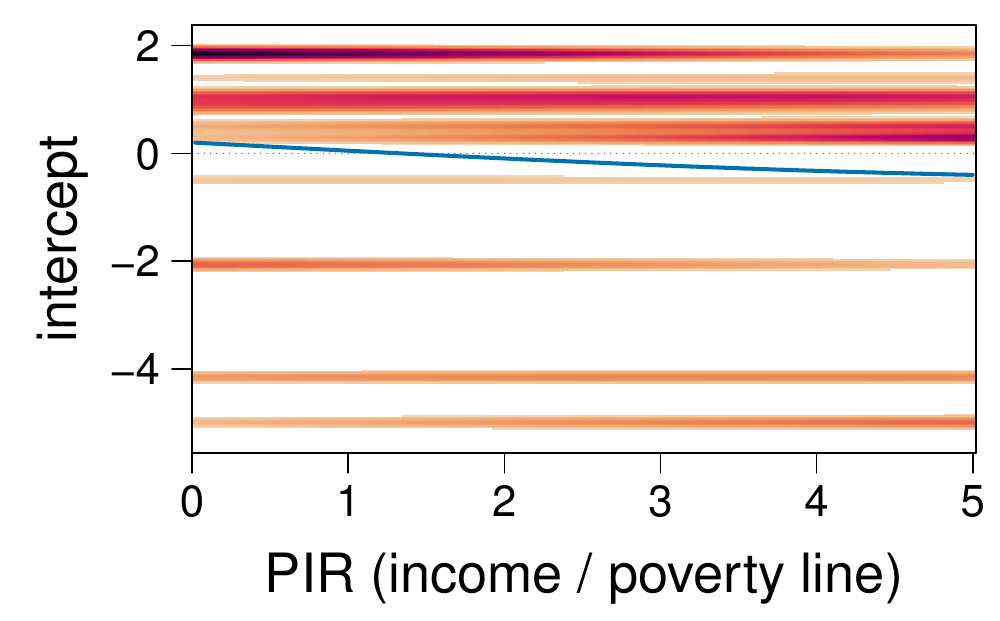}\caption{intercept}\end{subfigure}\hfill
    \begin{subfigure}[b]{0.32\textwidth}\centering\includegraphics[width=\linewidth]{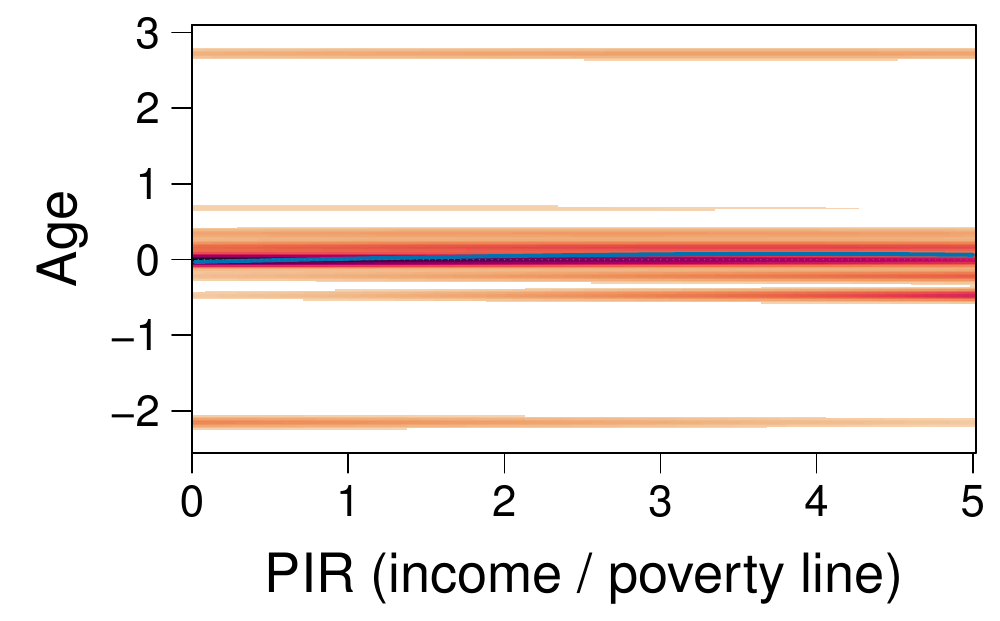}\caption{Age}\end{subfigure}\hfill
    \begin{subfigure}[b]{0.32\textwidth}\centering\includegraphics[width=\linewidth]{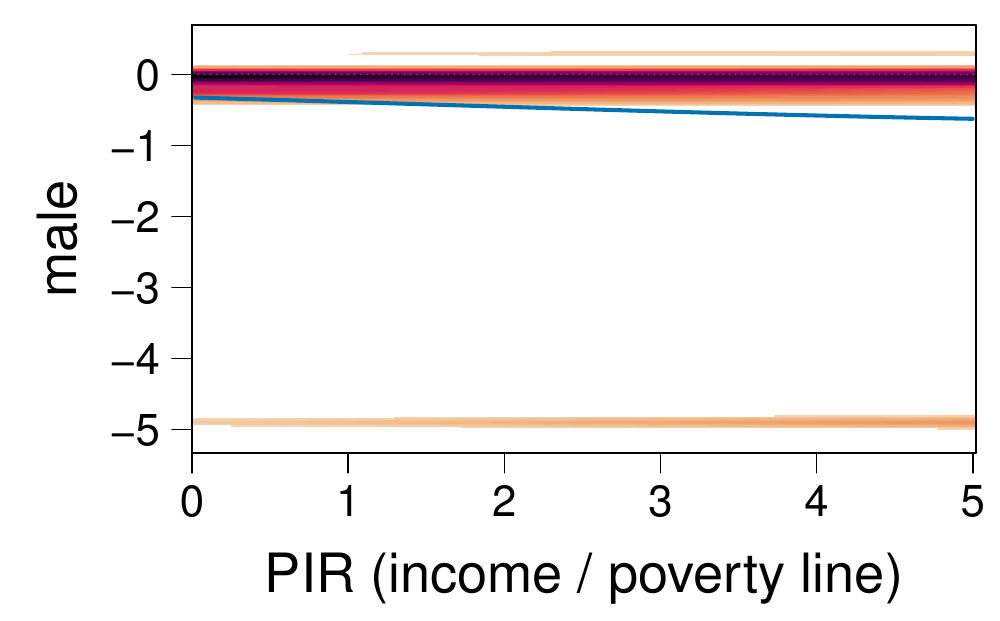}\caption{male}\end{subfigure}\\[1ex]
    \begin{subfigure}[b]{0.32\textwidth}\centering\includegraphics[width=\linewidth]{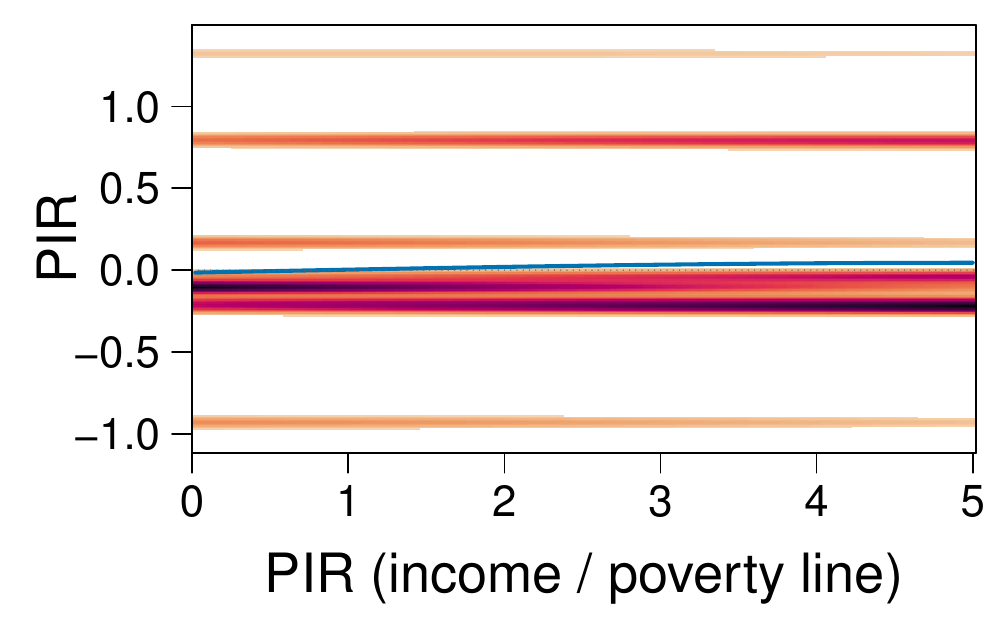}\caption{PIR}\end{subfigure}\hfill
    \begin{subfigure}[b]{0.32\textwidth}\centering\includegraphics[width=\linewidth]{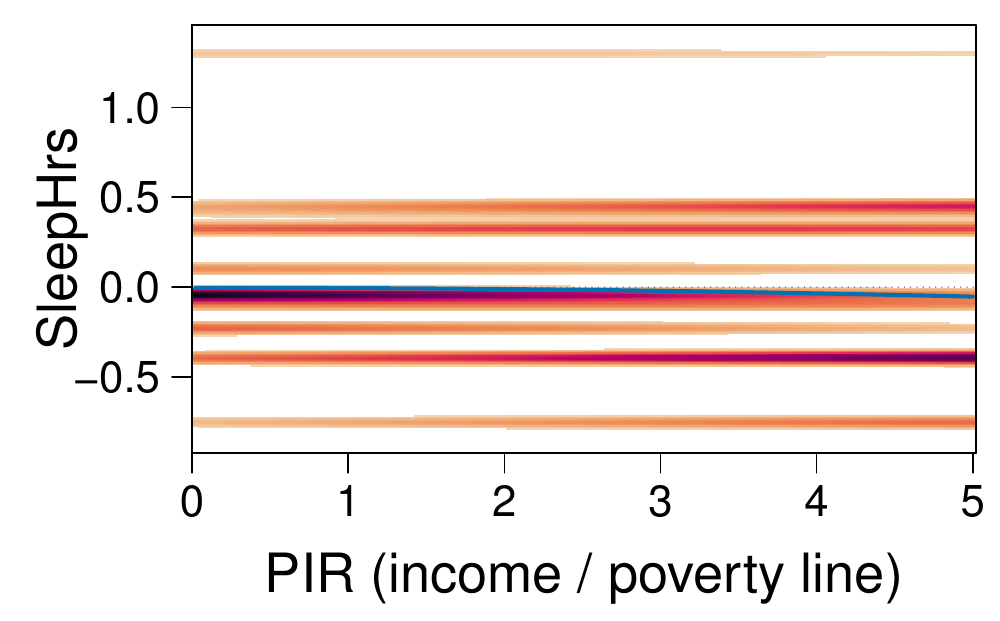}\caption{SleepHrs}\end{subfigure}\hfill
    \begin{subfigure}[b]{0.32\textwidth}\centering\includegraphics[width=\linewidth]{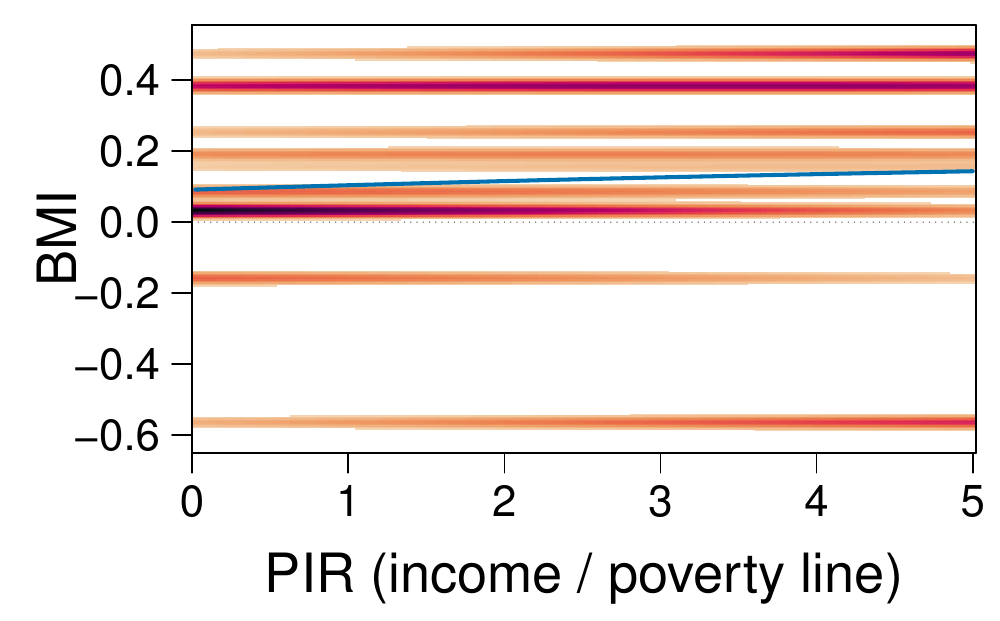}\caption{BMI}\end{subfigure}
    \caption{Conditional prior along the income.}
    \label{fig:appx_nhanes_pir}
\end{figure}

\begin{figure}[H]
    \centering
    \begin{subfigure}[b]{0.32\textwidth}\centering\includegraphics[width=\linewidth]{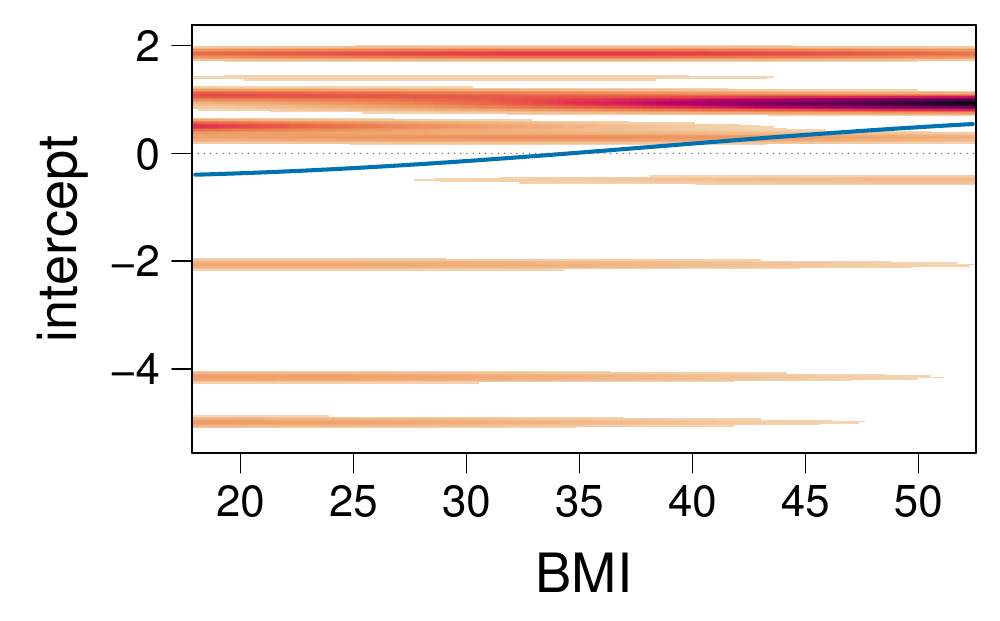}\caption{intercept}\end{subfigure}\hfill
    \begin{subfigure}[b]{0.32\textwidth}\centering\includegraphics[width=\linewidth]{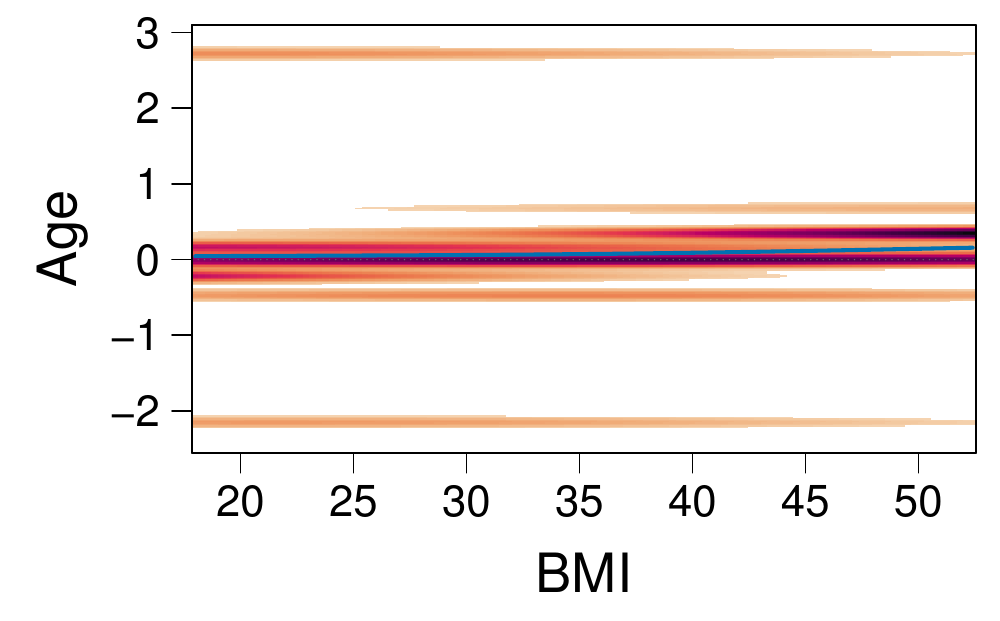}\caption{Age}\end{subfigure}\hfill
    \begin{subfigure}[b]{0.32\textwidth}\centering\includegraphics[width=\linewidth]{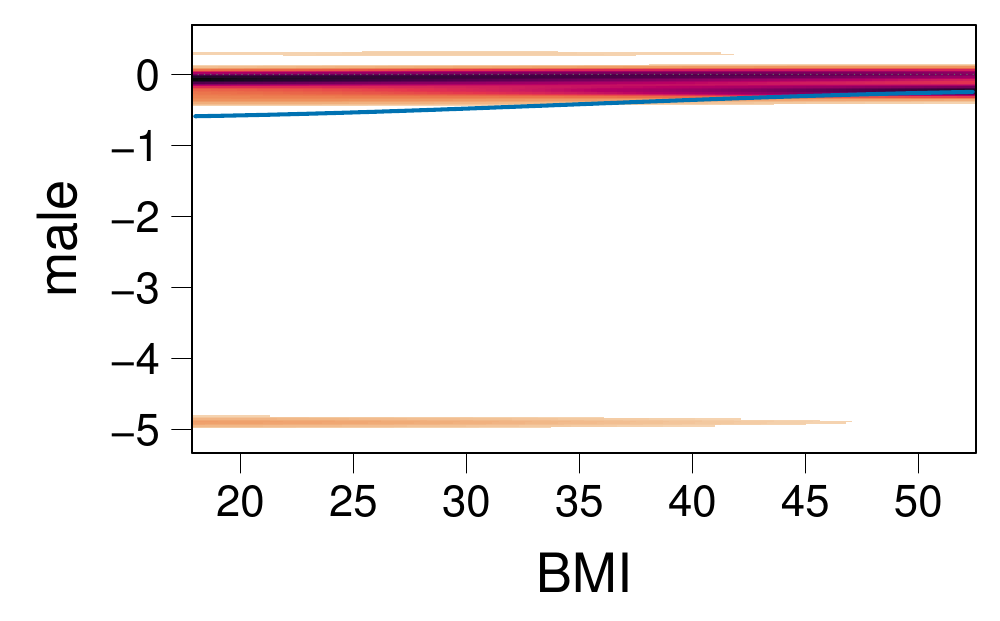}\caption{male}\end{subfigure}\\[1ex]
    \begin{subfigure}[b]{0.32\textwidth}\centering\includegraphics[width=\linewidth]{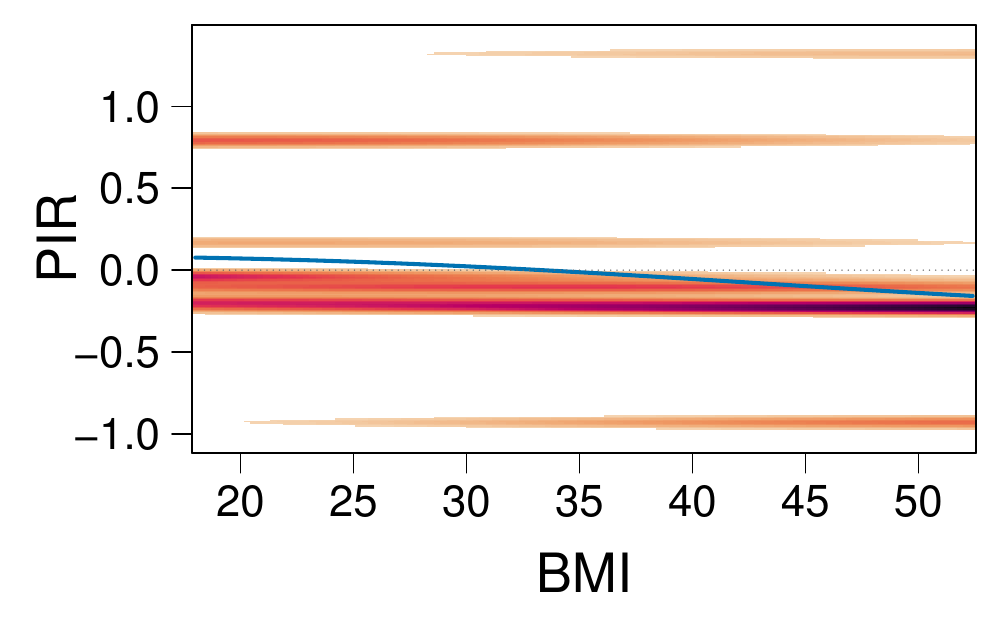}\caption{PIR}\end{subfigure}\hfill
    \begin{subfigure}[b]{0.32\textwidth}\centering\includegraphics[width=\linewidth]{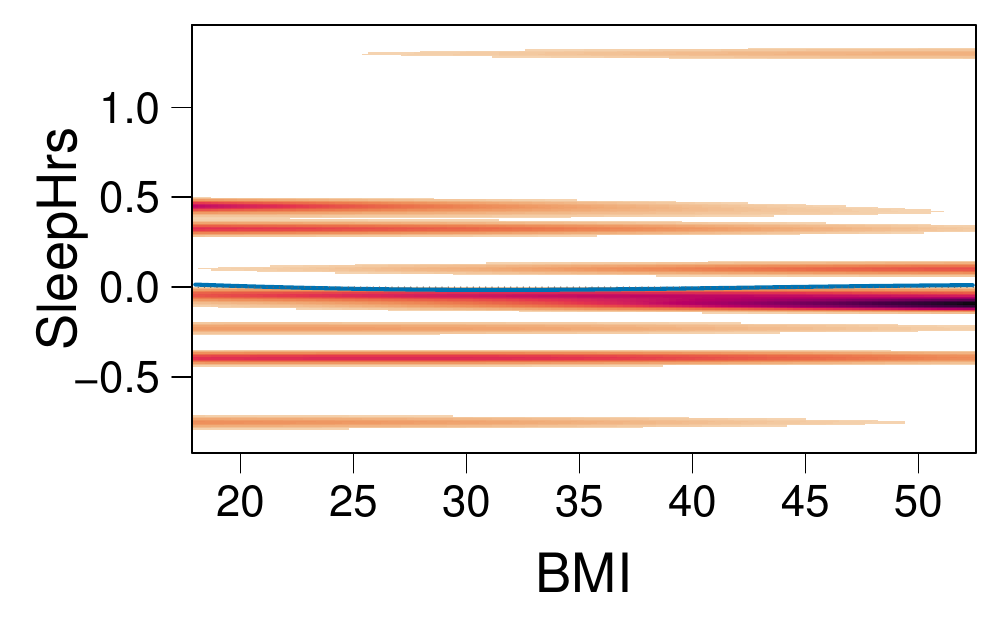}\caption{SleepHrs}\end{subfigure}\hfill
    \begin{subfigure}[b]{0.32\textwidth}\centering\includegraphics[width=\linewidth]{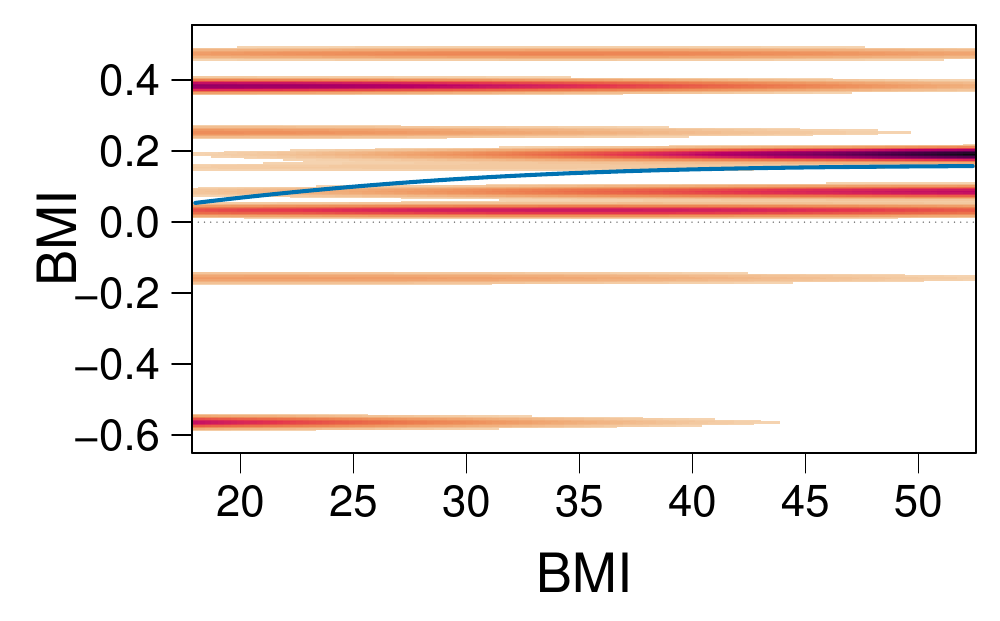}\caption{BMI}\end{subfigure}
    \caption{Conditional prior along the body mass index.}
    \label{fig:appx_nhanes_bmi}
\end{figure}

\FloatBarrier

\section{Comparison with Algorithm \ref{alg:exmoe} on a small set of exemplars}
\label{app:alg2_comparison}

This appendix compares the fits of Algorithm \ref{alg:fixed_dictionary} reported in Section \ref{sec:Simulation} with those of Algorithm \ref{alg:exmoe} on the same data. Algorithm \ref{alg:exmoe} starts from the $M=500$ exemplars $\beta^{(1)},\ldots,\beta^{(500)}$ of Algorithm \ref{alg:exemplar}, which is smaller than $2n$ in all three settings, and updates the gate parameters $\gamma_j$ with the ridge parameter $\lambda=5\times10^{-2}$. We run Algorithm \ref{alg:exmoe} with the gate covariate $z=s$ and $\sigma$ fixed at the $\hat{\sigma}$ selected for Algorithm \ref{alg:fixed_dictionary}, with the stopping rule of the final fits in Section \ref{sec:real_data}, that is, $R=3$, $T_{\max}=2000$, and $\epsilon_{\mathrm{tol}}=10^{-9}$, and with the other tuning parameters at the default values of Appendix \ref{app:tunings}. All three fits converged.

We provide a comparison between the results of Algorithm \ref{alg:exmoe} and \ref{alg:fixed_dictionary}. The compared values are the following: the training log-likelihood, log-likelihood on an independent test sample of the same size, and the root mean squared error (RMSE) of the estimated conditional mean $\hat{E}[y|s]$ against the true conditional mean on a grid of $s$. 

Figures \ref{fig:sim_alg2_gates_sin} and \ref{fig:sim_alg2_gates_sv} show the gates of these estimates for the sin and smooth Voronoi data in the format of Figure \ref{fig:sim_sin_gates}. For these estimates, the elbow method of Section \ref{sec:experts_aggregation} selects $3$, $4$, and $2$ clusters and BIC selects $11$, $6$, and $12$ clusters for the sin, smooth Voronoi, and switching circle data, against the true numbers $2$ and $6$ for the first two data sets. Figure \ref{fig:sim_alg2_circle_gate} shows estimated conditional prior for swithcing circle data as same manner of Figure \ref{fig:sim_threshold_circle}.

\begin{table}[htbp]
    \centering
    \caption{Algorithm \ref{alg:fixed_dictionary} (dictionary of $10{,}500$ points) and Algorithm \ref{alg:exmoe} ($500$ exemplars, updated $\gamma$ and $\beta$) on the three simulated data sets. The columns report the training and test log-likelihoods, the RMSE of the conditional mean, and the size of the estimated support. All estimates are at the cross-validated $\hat{\sigma}$ of Section \ref{sec:Simulation}.}
    \label{tab:sim_alg2_vs_alg3}
    \small
    \begin{tabular}{llrrrr}
        \hline
        Data & Algorithm & $\ell_{\mathrm{train}}$ & $\ell_{\mathrm{test}}$ & RMSE & $|\hat{\mathcal{S}}|$ \\
        \hline
        sin              & Alg.\ \ref{alg:fixed_dictionary} & $-891.1$  & $-972.9$   & $1.251$ & 43 \\
                         & Alg.\ \ref{alg:exmoe}            & $-850.1$  & $-942.2$  & $0.598$ & 15 \\
        smooth Voronoi   & Alg.\ \ref{alg:fixed_dictionary} & $-1360.0$ & $-1477.3$  & $0.101$ & 41 \\
                         & Alg.\ \ref{alg:exmoe}            & $-1339.8$ & $-1456.0$ & $0.112$ & 15 \\
        switching circle & Alg.\ \ref{alg:fixed_dictionary} & $-3345.2$ & $-3389.1$  & $0.217$ & 67 \\
                         & Alg.\ \ref{alg:exmoe}            & $-3324.4$ & $-3378.3$ & $0.207$ & 20 \\
        \hline
    \end{tabular}
\end{table}

\begin{figure}[htbp]
    \centering
    \begin{subfigure}[b]{0.32\textwidth}\centering\includegraphics[width=\linewidth]{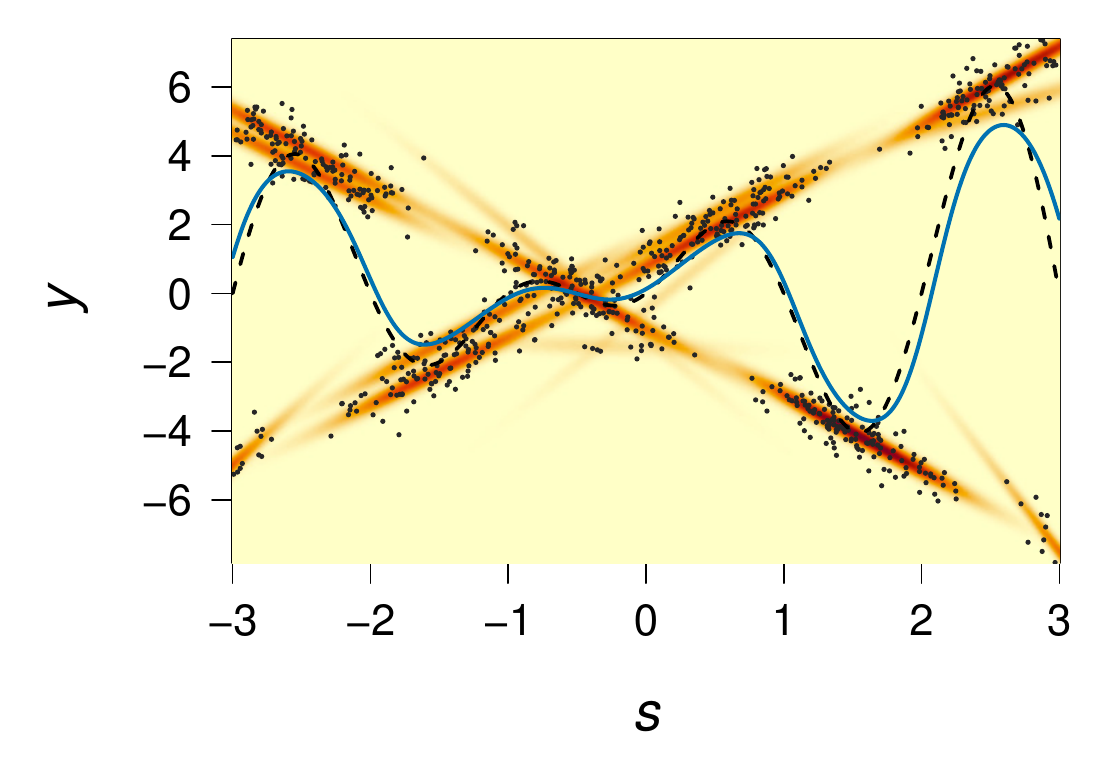}\caption{sin data}\end{subfigure}\hfill
    \begin{subfigure}[b]{0.32\textwidth}\centering\includegraphics[width=\linewidth]{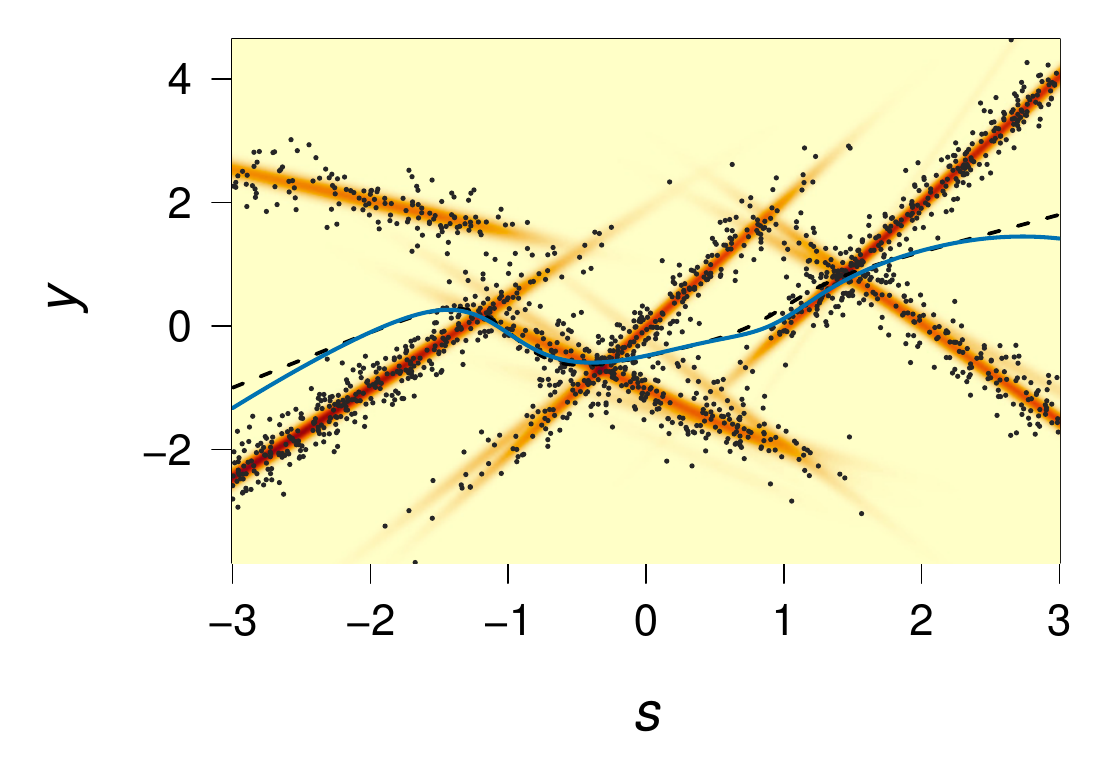}\caption{smooth Voronoi data}\end{subfigure}\hfill
    \begin{subfigure}[b]{0.32\textwidth}\centering\includegraphics[width=\linewidth]{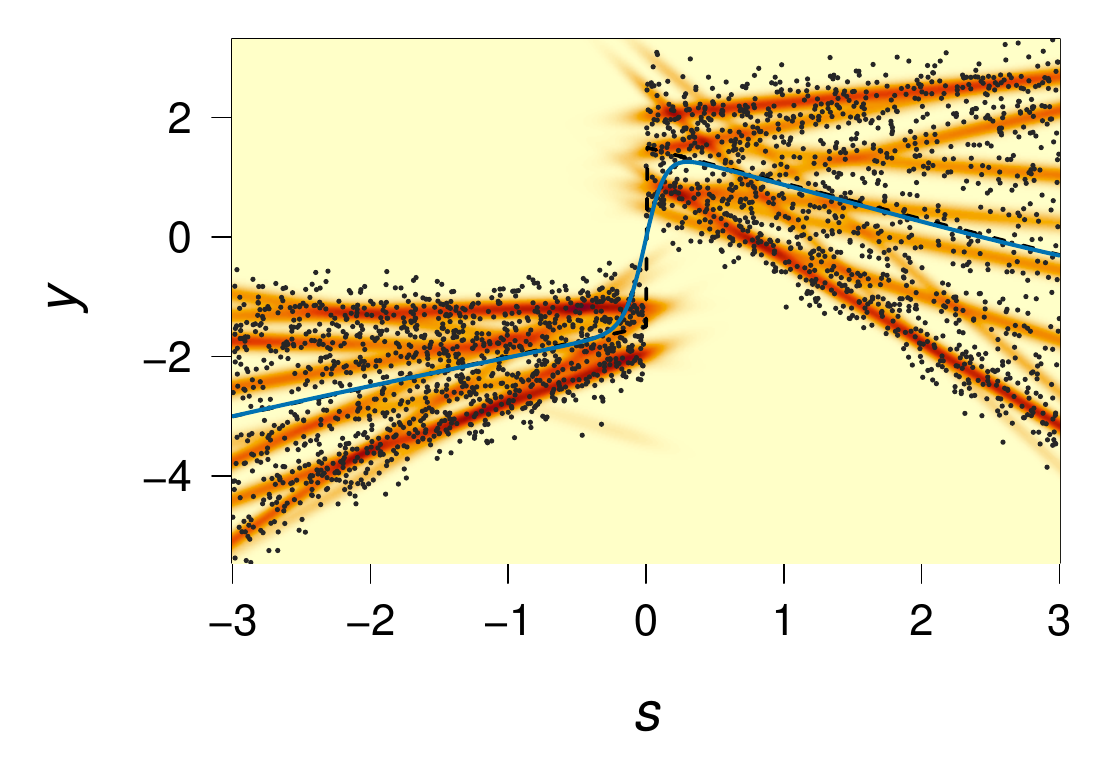}\caption{switching circle data}\end{subfigure}
    \caption{Estimates of Algorithm \ref{alg:exmoe} with $500$ exemplars and updated $\gamma$ and $\beta$ for the three simulated data sets, in the format of Figure \ref{fig:sim_heat}. Figure \ref{fig:sim_alg2_circle_gate} shows the gate-weighted conditional priors of this estimate for the switching circle data in the format of Figure \ref{fig:sim_threshold_circle}.}
    \label{fig:sim_alg2_fits}
\end{figure}

\begin{figure}[htbp]
    \centering
    \begin{subfigure}[b]{0.32\textwidth}\centering\includegraphics[width=\linewidth]{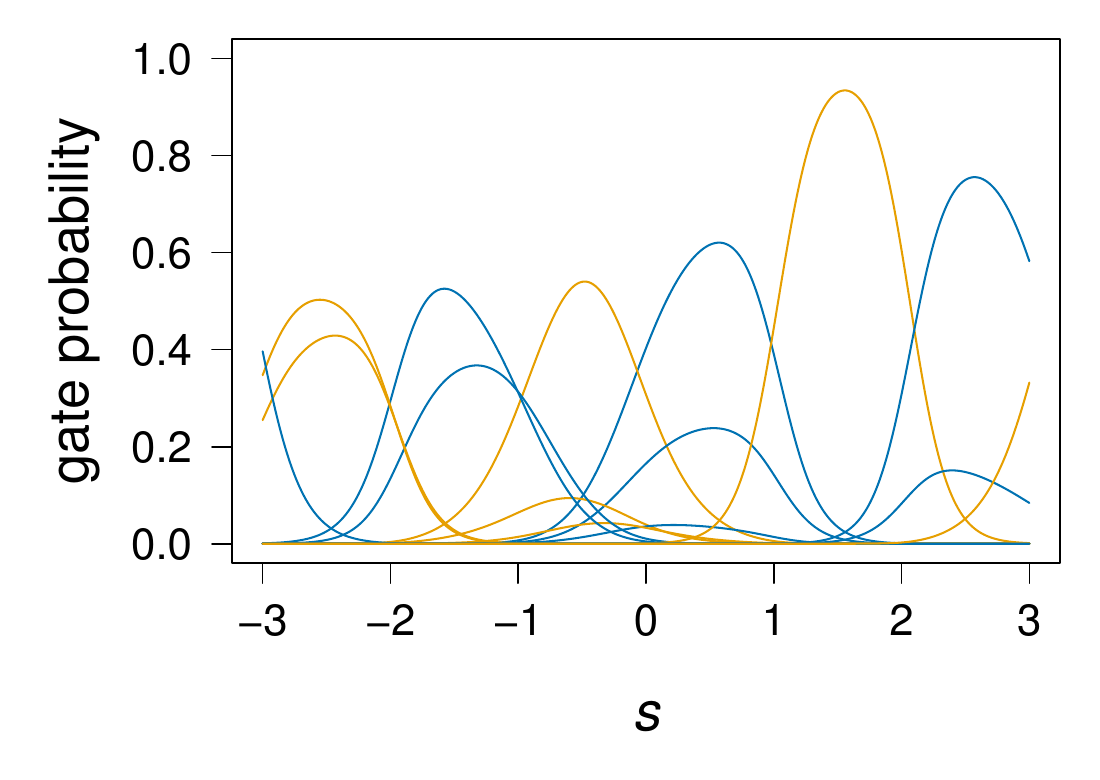}\caption{individual gates}\label{fig:sim_alg2_gates_sin_separate}\end{subfigure}\hfill
    \begin{subfigure}[b]{0.32\textwidth}\centering\includegraphics[width=\linewidth]{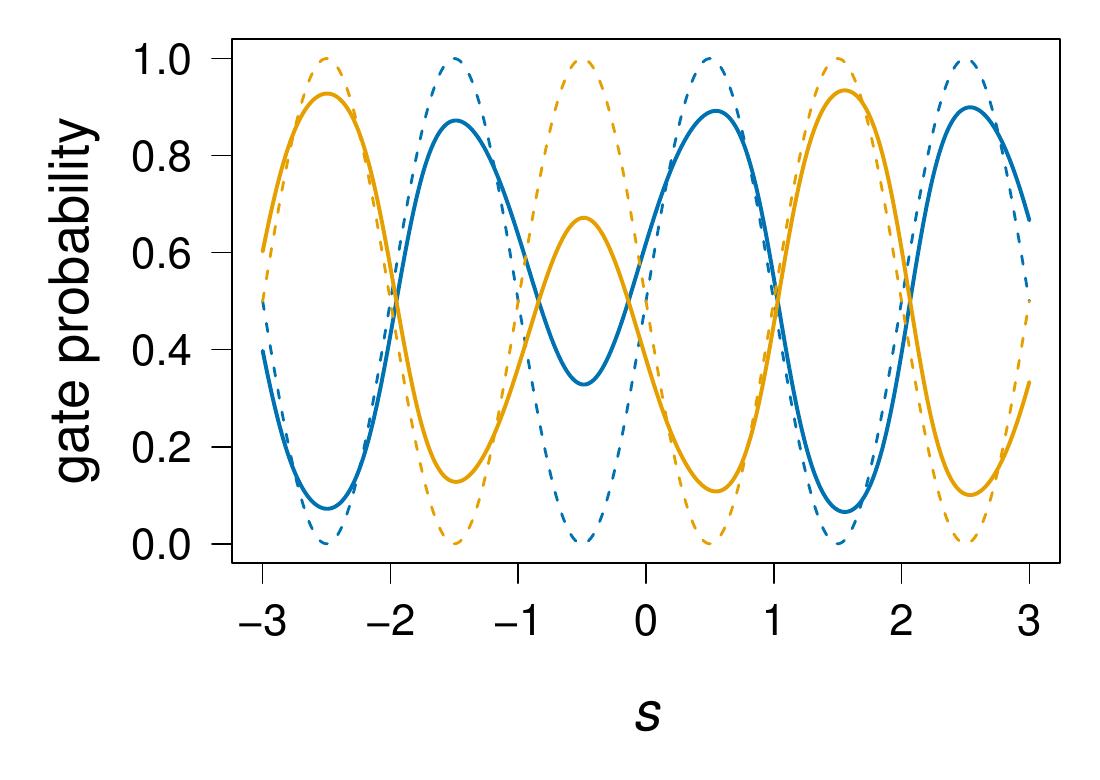}\caption{aggregation by nearest oracle}\label{fig:sim_alg2_gates_sin_aggregated}\end{subfigure}\hfill
    \begin{subfigure}[b]{0.32\textwidth}\centering\includegraphics[width=\linewidth]{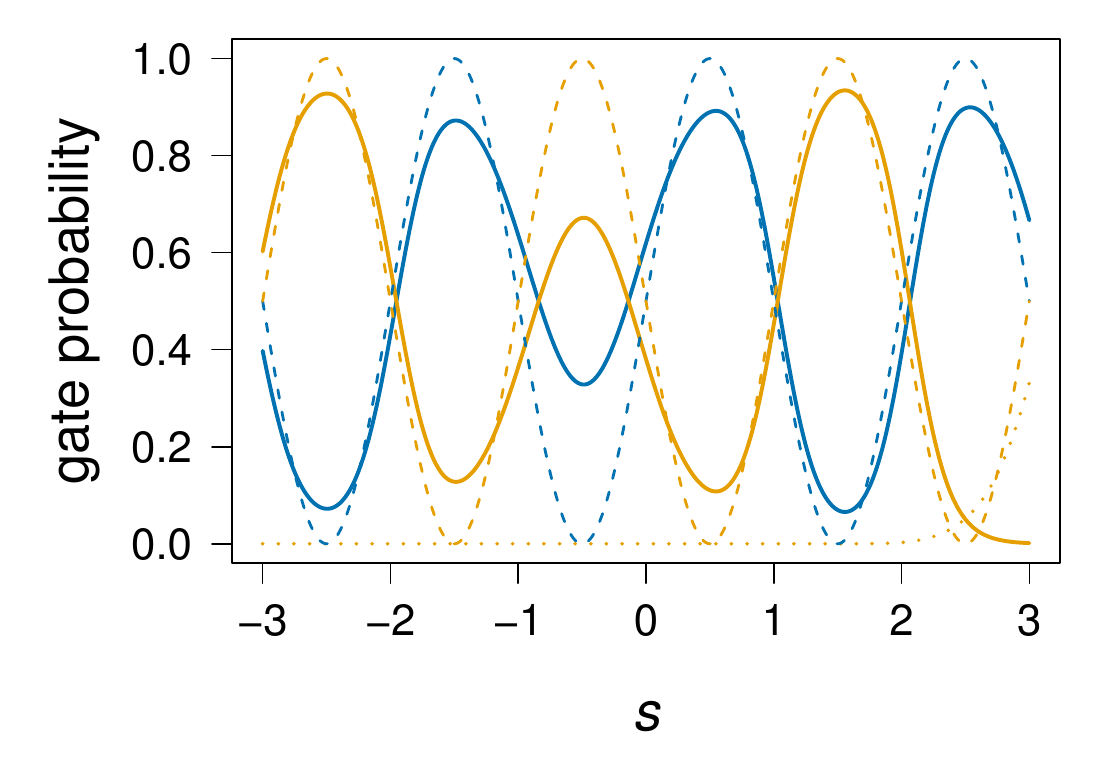}\caption{aggregation by the elbow rule}\label{fig:sim_alg2_gates_sin_elbow}\end{subfigure}
    \caption{Gates of the estimate of Algorithm \ref{alg:exmoe} for the sin data, in the format of Figure \ref{fig:sim_sin_gates}. (a): the individual gates of the $15$ support points. (b): the gates aggregated by the nearest true support point (solid) and the true gates (dashed). (c): the aggregated gates for the $\hat{K}=3$ clusters selected by the elbow method with the weights $\bar{g}_j$, where the third cluster consists of a support point of small weight far from both true coefficients and is drawn dotted.}
    \label{fig:sim_alg2_gates_sin}
\end{figure}

\begin{figure}[htbp]
    \centering
    \begin{subfigure}[b]{0.32\textwidth}\centering\includegraphics[width=\linewidth]{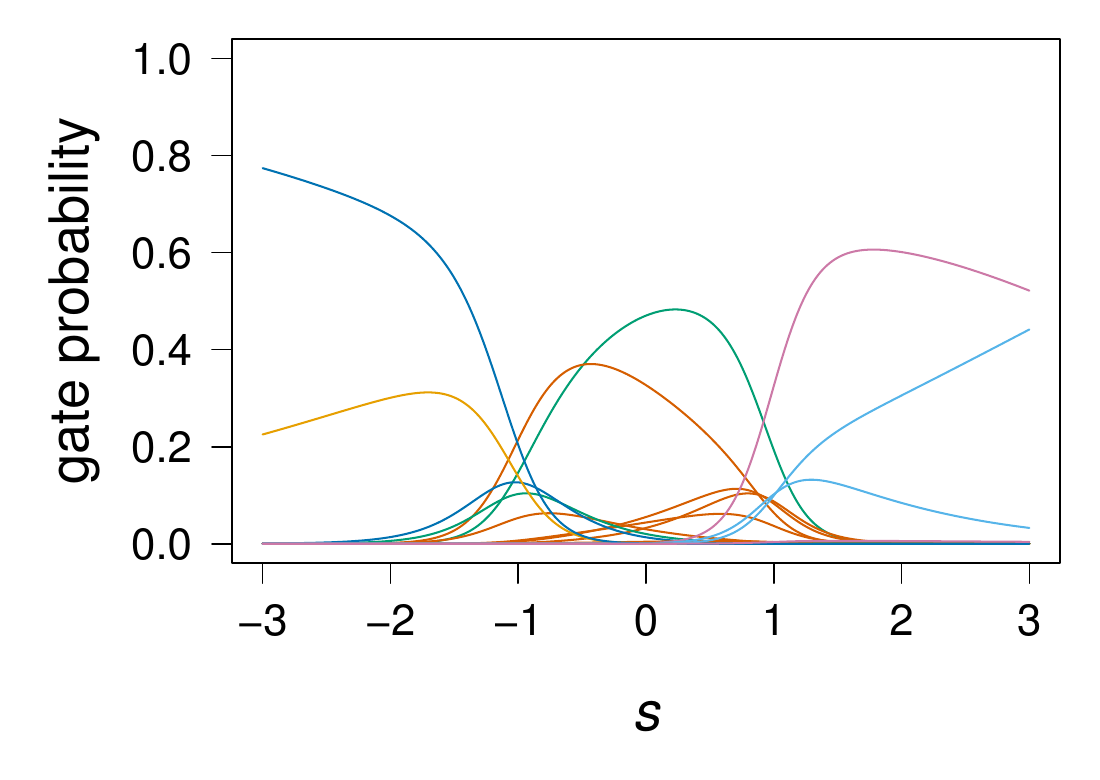}\caption{individual gates}\label{fig:sim_alg2_gates_sv_separate}\end{subfigure}\hfill
    \begin{subfigure}[b]{0.32\textwidth}\centering\includegraphics[width=\linewidth]{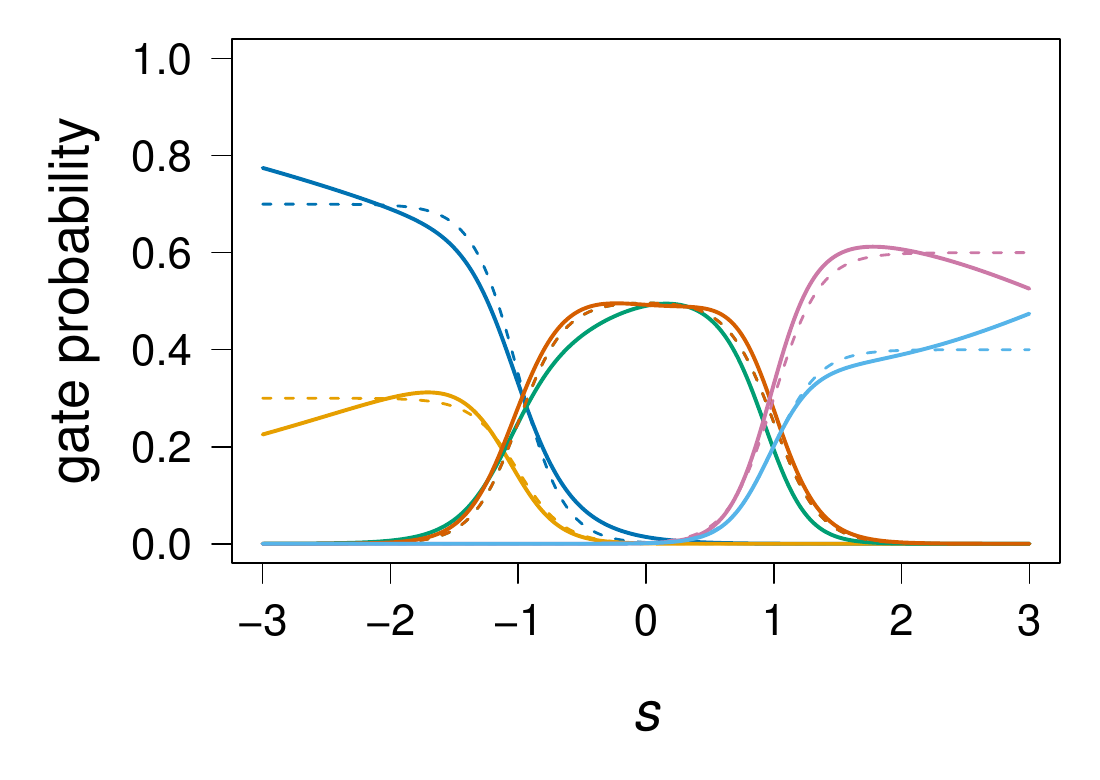}\caption{aggregation by nearest oracle}\label{fig:sim_alg2_gates_sv_aggregated}\end{subfigure}\hfill
    \begin{subfigure}[b]{0.32\textwidth}\centering\includegraphics[width=\linewidth]{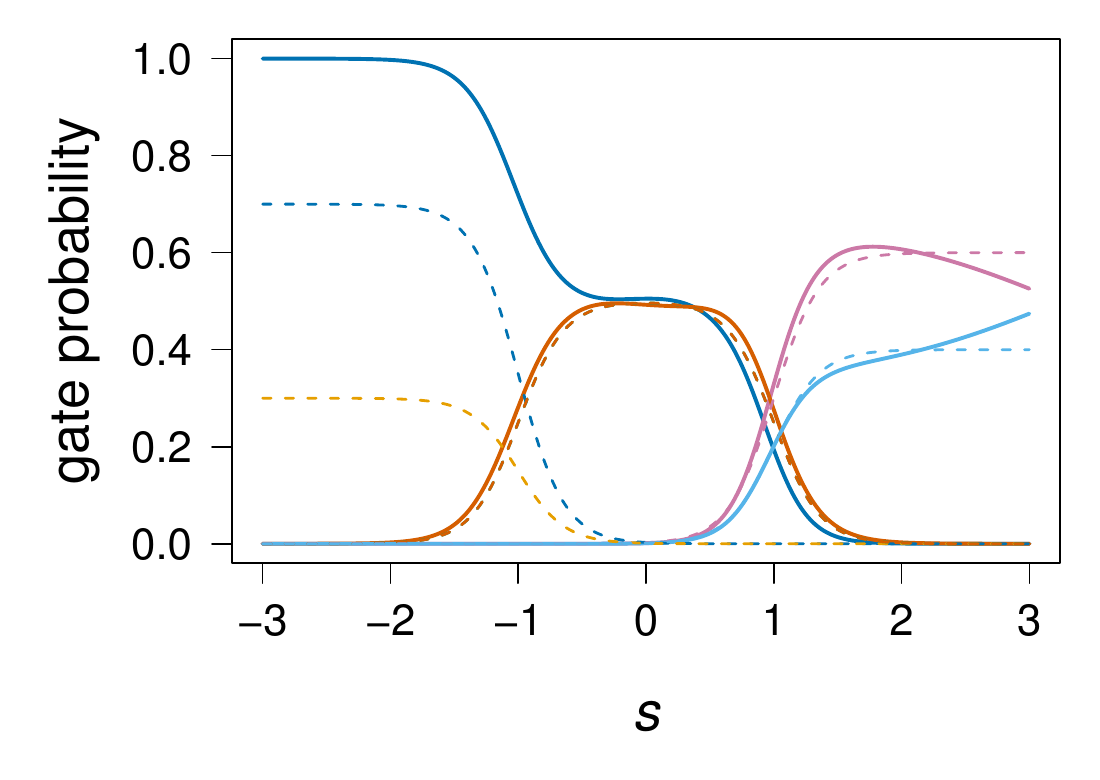}\caption{aggregation by the elbow rule}\label{fig:sim_alg2_gates_sv_elbow}\end{subfigure}
    \caption{Gates of the estimate of Algorithm \ref{alg:exmoe} for the smooth Voronoi data, in the format of Figure \ref{fig:sim_sv_gates}. (a): the individual gates of the $13$ support points. (b): the gates aggregated by the nearest true support point (solid) and the six true gates (dashed). (c): the aggregated gates for the $\hat{K}=4$ clusters selected by the elbow method with the weights $\bar{g}_j$ and the true gates (dashed).}
    \label{fig:sim_alg2_gates_sv}
\end{figure}

\begin{figure}[htbp]
    \centering
    \begin{subfigure}[b]{0.31\textwidth}\centering\includegraphics[width=\linewidth]{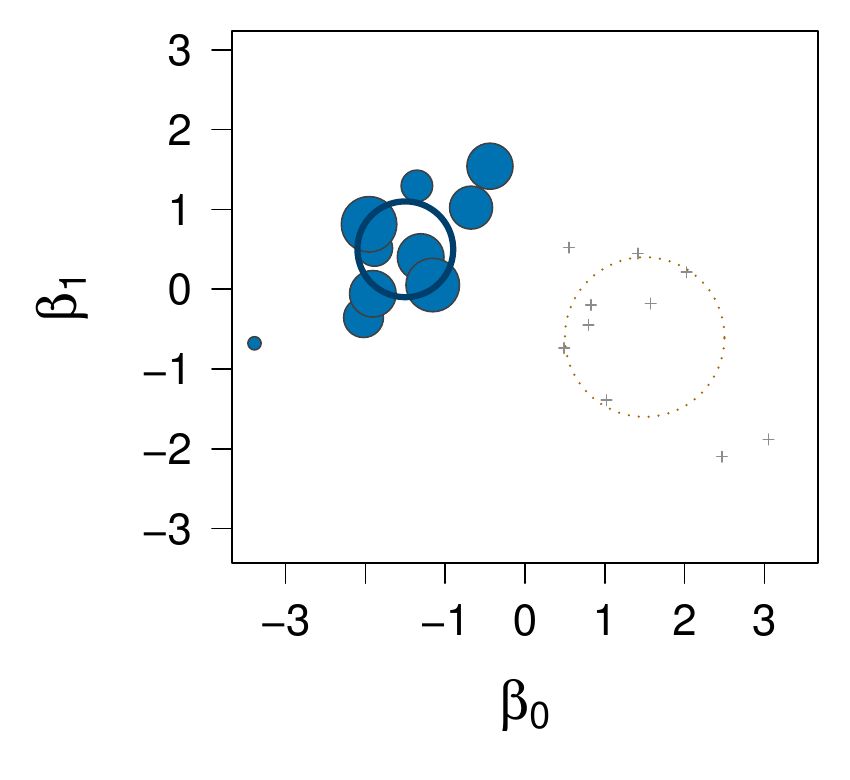}\caption{$s=-2$}\end{subfigure}\hspace{0.04\textwidth}
    \begin{subfigure}[b]{0.31\textwidth}\centering\includegraphics[width=\linewidth]{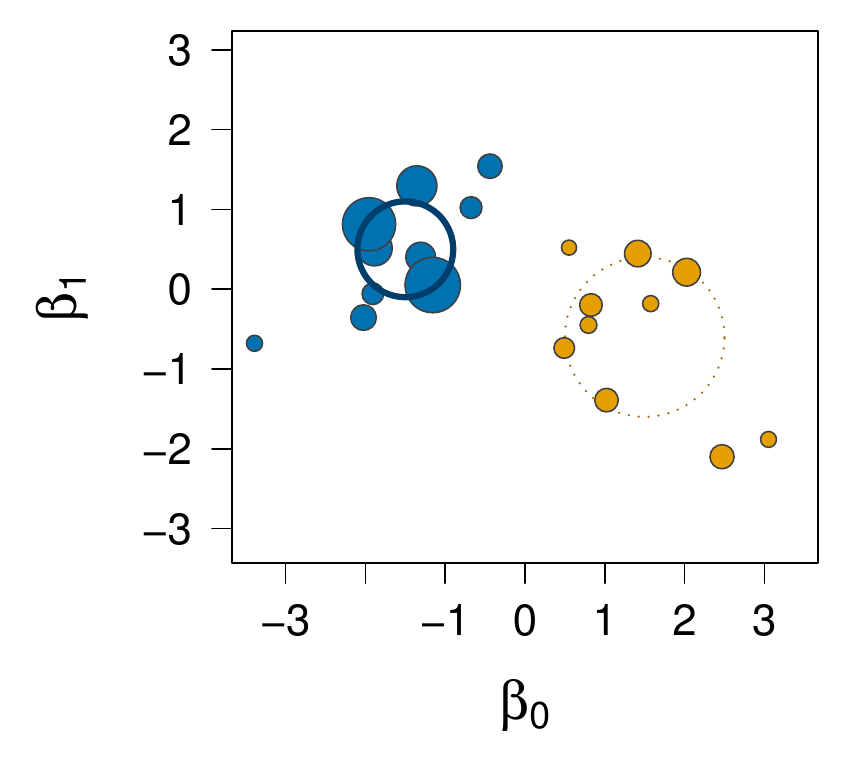}\caption{$s=-0.1$}\end{subfigure}\\[1ex]
    \begin{subfigure}[b]{0.31\textwidth}\centering\includegraphics[width=\linewidth]{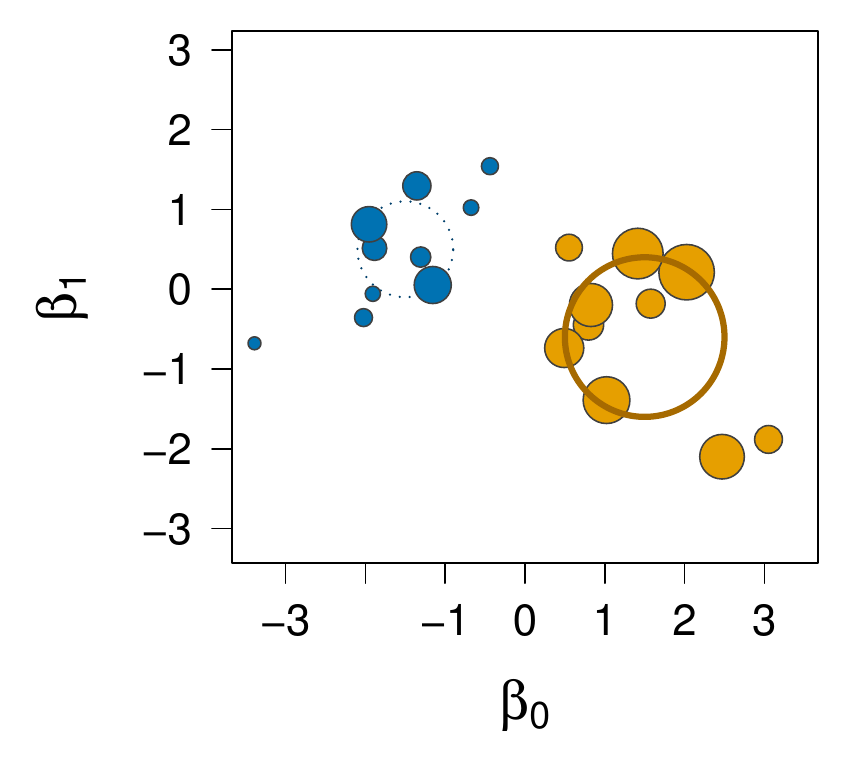}\caption{$s=0.1$}\end{subfigure}\hspace{0.04\textwidth}
    \begin{subfigure}[b]{0.31\textwidth}\centering\includegraphics[width=\linewidth]{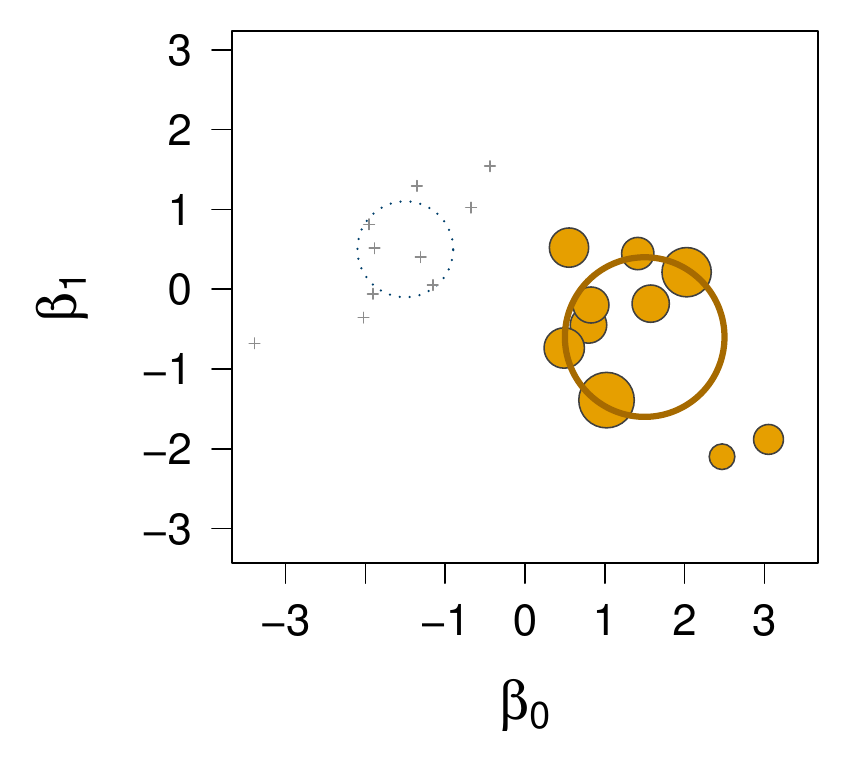}\caption{$s=2$}\end{subfigure}
    \caption{Switching structure of the gates for the switching circle data under Algorithm \ref{alg:exmoe} with $500$ exemplars and updated $\gamma$ and $\beta$, in the format of Figure \ref{fig:sim_threshold_circle}. The gate mass of the support points closer to the active circle is $1.00$, $0.80$, $0.79$, and $1.00$ at $s=-2$, $-0.1$, $0.1$, and $2$.}
    \label{fig:sim_alg2_circle_gate}
\end{figure}

\section{Kneedle rule of the elbow method}
\label{sec:keedle}
To determine the number of clusters $K$ by the elbow method, we use the following procedure.

Let $K_{\max}$ be the maximum number of clusters to be searched. For each sequence of partitions $K=1,\ldots,K_{\max}$, the weighted within-cluster sum of squares calculated
\begin{equation}
    \label{def:wss}
    \mathrm{WSS}(K)\coloneq\sum_{k=1}^{K}\sum_{j\in\mathcal{C}_k}\bar{g}_j\|\hat{\beta}_j-\bar{\beta}_k\|_2^{2},\qquad K=1,\ldots,K_{\max},
\end{equation}
where $\bar{\beta}_k$ is the weighted mean of the regression coefficients in the $k$-th cluster defined in \eqref{def:aggregated_gate}. Since the partitions are nested, $\mathrm{WSS}(K)$ is nonincreasing in $K$. The elbow is computed by the simplest form of the Kneedle rule of \cite{Satopaa2011-kn}:
\begin{equation}
    \label{def:elbow}
    \hat{K}_{\mathrm{elbow}}\coloneq\argmax_{1\leq K\leq K_{\max}}\left\{\left(1-\frac{K-1}{K_{\max}-1}\right)-\frac{\mathrm{WSS}(K)-\mathrm{WSS}(K_{\max})}{\mathrm{WSS}(1)-\mathrm{WSS}(K_{\max})}\right\},
\end{equation}
and $\hat{K}$ is chosen as the number of the aggregated linear relationships.

\section{Fused slope structure of the softmax gates}
\label{app:fused}

This appendix describes the structure of the extended softmax gates that is used in the proof of Theorem \ref{thm:uat_gate}. The following condition is necessary to approximate a step function in the simplex space by using a softmax function.

\begin{condition}[Fused slope condition]
    \label{cond:fusion}
    For a given positive integer $C$, the index set $\{1,\ldots,J\}$ of the softmax gates $\{g_j(z;\alpha,\gamma):j=1,\ldots,J\}$ is partitioned into $L$ groups $I_1,\ldots,I_L$ with $|I_l|=C$ for $l=1,\ldots,L$, so that $J=LC$, and the slope parameters are shared within each group, that is,
    \begin{equation}
        \gamma_j=\gamma_{j^{\prime}}\quad\text{for all}\;j,j^{\prime}\in I_l,\qquad l=1,\ldots,L.
    \end{equation}
\end{condition}

Under Condition \ref{cond:fusion}, 
we reparametrize $\alpha$ in the softmax gates as follows.
Write the common slope of the group $I_l$ as
\begin{equation}
    \frac{\xi_l}{\tau}\coloneq\gamma_j\quad\text{for}\;j\in I_l,\qquad l=1,\ldots,L,
\end{equation}
where $\xi_l\in\mathbb{R}^{d}$ and $\tau>0$ is a temperature parameter. 
For $\xi_{10},\ldots,\xi_{L0}\in\mathbb{R}$ and $w=(w_1,\ldots,w_{J})^{\top}\in\Delta^{J-1}$ with $w_j>0$, define
\begin{equation}
    \label{eq:alpha_reparametrization}
    \alpha_j\coloneq\frac{w_j\,e^{\xi_{l0}/\tau}}{\sum_{l^{\prime}=1}^{L}\bigl(\sum_{k\in I_{l^{\prime}}}w_k\bigr)e^{\xi_{l^{\prime}0}/\tau}},\qquad j\in I_l,\quad l=1,\ldots,L.
\end{equation}
Then $\alpha\in\Delta^{J-1}$ with $\alpha_j>0$, so that $(\alpha,\gamma)$ is an admissible parameter of the softmax gates. 
Conversely, every admissible $(\alpha,\gamma)$ satisfying Condition \ref{cond:fusion} with $\alpha_j>0$ is obtained from \eqref{eq:alpha_reparametrization} by taking $\xi_{l0}=0$ and $w=\alpha$.

Let us explain the limit behaviour with respect to $\tau$.
Writing the affine score of the group $l$ as
\begin{equation}
    \label{def:affine_score}
    h_l(z)\coloneq z^{\top}\xi_l+\xi_{l0},\qquad l=1,\ldots,L,
\end{equation}
the $j$-th softmax gate with $j\in I_l$ is expressed as
\begin{equation}
    \label{eq:fused_gate}
    \begin{split}
        g_j(z;\alpha,\gamma) & = \frac{\alpha_j e^{z^{\top}\gamma_j}}{\sum_{k=1}^{J}\alpha_k e^{z^{\top}\gamma_k}}
        = \frac{w_j\exp\{h_l(z)/\tau\}}{\sum_{l^{\prime}=1}^{L}\bigl(\sum_{k\in I_{l^{\prime}}}w_k\bigr)\exp\{h_{l^{\prime}}(z)/\tau\}} \\
                             & = \frac{w_j}{\bigl(\sum_{k\in I_l}w_k\bigr)+\sum_{l^{\prime}\neq l}\bigl(\sum_{k\in I_{l^{\prime}}}w_k\bigr)\exp\{(h_{l^{\prime}}(z)-h_l(z))/\tau\}},
    \end{split}
\end{equation}
where the normalizing constant of \eqref{eq:alpha_reparametrization} cancels between the numerator and the denominator. In the limit $\tau\rightarrow +0$, the gate converges to
\begin{equation}
    \label{eq:fused_gate_limit}
    g_j^{+0}(z)=\left\{
    \begin{array}{ll}
        \dfrac{w_j}{\sum_{k\in I_l}w_k},                                                          & h_l(z)>h_{l^{\prime}}(z)\;\text{for all}\;l^{\prime}\neq l, \\[2.5ex]
        \dfrac{w_j}{\sum_{l^{\prime}\in\mathcal{I}(z)}\sum_{k\in I_{l^{\prime}}}w_k},               & l\in\mathcal{I}(z)\;\text{and}\;|\mathcal{I}(z)|\geq2,       \\[2.5ex]
        0,                                                                                          & l\notin\mathcal{I}(z),
    \end{array}\right.
    \qquad j\in I_l,
\end{equation}
where $\mathcal{I}(z)\coloneq\argmax_{l^{\prime}=1,\ldots,L}h_{l^{\prime}}(z)$ is the set of groups attaining the maximal score at $z$. Thus, the temperature limit of the aggregated gate function $g^{\dagger}_c(z;\alpha,\gamma)$ behaves as a step function on the cells $\{z:h_l(z)>\max_{l^{\prime}\neq l}h_{l^{\prime}}(z)\}$, which can take an arbitrary value on the simplex in each cell, and the within-group weights $w_j$ determine these values.

\section{Proofs}
\label{app:proofs}

\paragraph{Proof of Theorem \ref{thm:necessary condition}}

\begin{proof}
    We first show (i). Since $\hat{H}$ is a global maximizer of the marginal log-likelihood $\ell(H)$, for any $H\in \mathcal{P}(B\times\Gamma)$ we have
    \begin{equation}
        \begin{split}
                    &  \ell(\hat{H})\geq \ell(\hat{H}_{\epsilon})=\ell((1-\epsilon)\hat{H}+\epsilon H) \\
                    \Leftrightarrow \quad & 0\geq \ell((1-\epsilon)\hat{H}+\epsilon H) - \ell(\hat{H}) \\
                    \Leftrightarrow \quad & 0\geq \frac{1}{\epsilon}[\ell((1-\epsilon)\hat{H}+\epsilon H) - \ell(\hat{H})].
        \end{split}
    \end{equation}
    Taking $H=\delta_{(\beta,\gamma)}$ and letting $\epsilon\to+0$, we obtain $D(\hat{H};(\beta,\gamma))\leq 0$ for every $(\beta,\gamma)\in B\times\Gamma$, which proves (i).

    Next, we show (ii). Integrating $D(\hat{H};(\beta,\gamma))$ with respect to $\hat{H}$ and using the definition \eqref{def:Lambda_Xi_tau}, we obtain
    \begin{equation}
        \label{eq:proof_necessary_integral}
        \int_{B\times\Gamma}D(\hat{H};(\beta,\gamma))\dd\hat{H}(\beta,\gamma)
        =\sum_{i=1}^{n}\left(\frac{\int_{B\times\Gamma} v_i\dd\hat{H}}{\Xi_i(\hat{H})}-\frac{\int_{B\times\Gamma} u_i\dd\hat{H}}{\Lambda_i(\hat{H})}\right)
        =\sum_{i=1}^{n}(1-1)=0.
    \end{equation}
    By (i), the integrand in \eqref{eq:proof_necessary_integral} is not positive on $B\times\Gamma$, and thus \eqref{eq:proof_necessary_integral} implies that $D(\hat{H};(\beta,\gamma))=0$ for $\hat{H}$-almost every $(\beta,\gamma)$.

    To extend this equality to the whole support, note that $(\beta,\gamma)\mapsto D(\hat{H};(\beta,\gamma))$ is continuous on $B\times\Gamma$, because $u_i(\gamma)=e^{\gamma^{\top}z_i}$ and $v_i(\beta,\gamma)=e^{\gamma^{\top}z_i}\psi(y_i;\theta(x_i^{\top}\beta),\phi)$ are continuous in $(\beta,\gamma)$ and the denominators $\Xi_i(\hat{H})$ and $\Lambda_i(\hat{H})$ do not depend on $(\beta,\gamma)$.
    Suppose that $D(\hat{H};(\beta_0,\gamma_0))<0$ for some $(\beta_0,\gamma_0)\in\mathrm{supp}(\hat{H})$. By continuity, $D(\hat{H};\cdot)<0$ on an open neighborhood $U$ of $(\beta_0,\gamma_0)$, and $\hat{H}(U)>0$ by the definition of the support. This contradicts the fact that $D(\hat{H};\cdot)=0$ $\hat{H}$-almost everywhere.
    Hence, we obtain $D(\hat{H};(\beta,\gamma))=0$ for all $(\beta,\gamma)\in\mathrm{supp}(\hat{H})$, which proves (ii).
\end{proof}

\paragraph{Proof of Proposition \ref{prop:necessary_condition}}

\begin{proof}
    By Theorem \ref{thm:necessary condition}, 
    we have $D(\hat{H};\cdot)\leq 0$ on $B\times\Gamma$ and $D(\hat{H};\cdot)=0$ on $\mathrm{supp}(\hat{H})$. Hence every $(\beta,\gamma)\in\mathrm{supp}(\hat{H})$ is a global maximizer of $D(\hat{H};\cdot)$ over $B\times\Gamma$.
    Under the assumptions on $b$ and $g$, the map $(\beta,\gamma)\mapsto D(\hat{H};(\beta,\gamma))$ is differentiable, and thus its gradient vanishes at every maximizer in the interior of $B\times\Gamma$.

    The explicit forms follow from \eqref{cal:directional_derivative} together with $\nabla_{\gamma}u_i(\gamma)=u_i(\gamma)z_i$, $\nabla_{\gamma}v_i(\beta,\gamma)=v_i(\beta,\gamma)z_i$, and, by the chain rule applied to $\theta(x_i^{\top}\beta)=(b')^{-1}(g^{-1}(x_i^{\top}\beta))$,
    \begin{equation}
        \nabla_{\beta}v_i(\beta,\gamma)
        =v_i(\beta,\gamma)\,\frac{y_i-b'(\theta(x_i^{\top}\beta))}{a(\phi)}\,\nabla_{\beta}\theta(x_i^{\top}\beta)
        =v_i(\beta,\gamma)\,\frac{y_i-\mu_i(\beta)}{a(\phi)\,V(\mu_i(\beta))\,g'(\mu_i(\beta))}\,x_i,
    \end{equation}
    where we used $\nabla_{\beta}\theta(x_i^{\top}\beta)=x_i/\{b''(\theta)\,g'(\mu)\}$ with $\theta=\theta(x_i^{\top}\beta)$ and $\mu=\mu_i(\beta)$, and $b''(\theta)=V(\mu)$.
    Thus we conclude the proof.
\end{proof}

\paragraph{Proof of Proposition \ref{prop:sparse_representation}}

\begin{proof}
    Write $s_i(\beta)\coloneq(y_i-\mu_i(\beta))x_i\in\mathbb{R}^{p}$ and $a_i\coloneq v_i(\beta,\gamma)/\{\Xi_i(\hat{H})\,a(\phi)\,V(\mu_i(\beta))\,g'(\mu_i(\beta))\}$.
    Then Equation \eqref{eq:gradient_condition_beta} becomes $\sum_{i=1}^{n}a_i s_i(\beta)=0$.
    Observe that all the $a_i$ have the same sign, because $v_i$, $\Xi_i(\hat{H})$, and $V$ are positive and $g'$ does not change sign.
    Then dividing $\sum_{i=1}^{n}a_i s_i(\beta)=0$ by $\sum_{i=1}^{n}a_i$, we see that $0$ should be a convex combination of the $n$ points $s_1(\beta),\ldots,s_n(\beta)$ in $\mathbb{R}^{p}$.
    By Carath\'{e}odory's theorem, $0$ is already a convex combination of at most $p+1$ of these points. The coefficients of this combination form the probability vector $w$, and $\sum_{i=1}^{n}w_i s_i(\beta)=0$ means $\beta\in S_{\beta}(w)$, which completes the proof.
\end{proof}

\paragraph{Proof of Proposition \ref{prop:covariate_representation}}

\begin{proof}
    Consider the first claim.
    Since $\Gamma$ is bounded, the denominator of $h$ in \eqref{def:covariate_construction} satisfies \[
    0<\int_{B\times\Gamma} e^{z^{\top}\gamma}\dd H(\beta,\gamma)
    \leq\exp(\|z\|\sup_{\gamma\in\Gamma}\|\gamma\|)<\infty\quad \text{for every}\quad z\in\mathbb{R}^{d},\]
    so $h$ is a well-defined nonnegative function and $c(\gamma)\in(0,\infty]$. By Tonelli's theorem, we have
    \begin{equation}
        \int_{B\times\Gamma} c(\gamma)\dd H(\beta,\gamma)
        =\int h(z)\int_{B\times\Gamma} e^{z^{\top}\gamma}\dd H(\beta,\gamma)\dd z
        =\int p^{\ast}(z)\dd z=1.
    \end{equation}
    Hence we ensure that $c(\gamma)<\infty$ for $H$-almost every $(\beta,\gamma)$, and thus $\tilde{H}$ in \eqref{def:covariate_construction} is a probability measure on $B\times\Gamma$. For every $\gamma$ with $c(\gamma)<\infty$, 
    letting $p(z|\gamma)=h(z)e^{z^{\top}\gamma}/c(\gamma)$ gives a probability density, so the covariate distribution \eqref{def:explanatory distribution} is well-defined for $\tilde{H}$-almost every $\gamma$. From the definition of marginal density, we have
    \begin{equation}
        \int_{B\times\Gamma} p(z|\gamma)\dd \tilde{H}(\beta,\gamma)
        =\int_{B\times\Gamma}\frac{h(z)e^{z^{\top}\gamma}}{c(\gamma)}\,c(\gamma)\dd H(\beta,\gamma)
        =h(z)\int_{B\times\Gamma} e^{z^{\top}\gamma}\dd H(\beta,\gamma)=p^{\ast}(z)
    \end{equation}
    for every $z\in\mathbb{R}^{d}$, which is the first claim.

    Consider the second claim. 
    Fix $z$ with $p^{\ast}(z)>0$ so that $h(z)>0$. Substituting $d\tilde{H}=c(\gamma)\,dH$ into \eqref{eq:posterior_interpretation_beta_gamma} and using the marginal density just computed as the denominator give
    \begin{equation}
        \dd G(\beta,\gamma|z)
        =\frac{h(z)e^{z^{\top}\gamma}\dd H(\beta,\gamma)}{p^{\ast}(z)}
        =\frac{e^{z^{\top}\gamma}\dd H(\beta,\gamma)}{\int_{B\times\Gamma} e^{z^{\top}\gamma}\dd H(\beta,\gamma)},
    \end{equation}
    and integrating this over $\gamma\in\Gamma$ yields \eqref{def:infinite_soft_max_prior}, that is, $\dd G(\beta|z)=\dd G_{H}(\beta|z)$. Since $\{z:p^{\ast}(z)=0\}$ is $p^{\ast}$-null, this proves the second claim.
\end{proof}

\paragraph{Proof of Proposition \ref{prop:lipschitz_gamma}}

\begin{proof}
    For a Borel set $A\subset B$, write
    \begin{equation}
        N_H(A,z)\coloneq\int_{A\times\Gamma}e^{\gamma^{\top}z}\dd H(\beta,\gamma),\qquad
        \Lambda_H(z)\coloneq N_H(B,z)=\int_{B\times\Gamma}e^{\gamma^{\top}z}\dd H(\beta,\gamma),
    \end{equation}
    so that $\Lambda_H(z)$ is the denominator of \eqref{def:infinite_soft_max_prior} and $G_H(A|z)=N_H(A,z)/\Lambda_H(z)$.

    We first show that, for every Borel set $A\subset B$,
    \begin{equation}
        \label{eq:lipschitz_ratio}
        e^{-2R_{\Gamma}\|z-z'\|_2}\,G_H(A|z')\leq G_H(A|z)\leq e^{2R_{\Gamma}\|z-z'\|_2}\,G_H(A|z').
    \end{equation}
    For every $(\beta,\gamma)\in B\times\Gamma$, the Cauchy--Schwarz inequality gives 
    \[|\gamma^{\top}(z-z')|\leq\|\gamma\|_2\|z-z'\|_2\leq R_{\Gamma}\|z-z'\|_2,\]
    and hence we get
    \begin{equation}
        e^{-R_{\Gamma}\|z-z'\|_2}\,e^{\gamma^{\top}z'}\leq e^{\gamma^{\top}z}\leq e^{R_{\Gamma}\|z-z'\|_2}\,e^{\gamma^{\top}z'}.
    \end{equation}
    Integrating this over $A\times\Gamma$ with respect to $H$ yields
    \begin{equation}
        e^{-R_{\Gamma}\|z-z'\|_2}\,N_H(A,z')\leq N_H(A,z)\leq e^{R_{\Gamma}\|z-z'\|_2}\,N_H(A,z'),
    \end{equation}
    and taking $A=B$ yields the same bounds between $\Lambda_H(z)$ and $\Lambda_H(z')$. Dividing the bound on $N_H(A,z)$ by the reversed bound on $\Lambda_H(z)$ yields \eqref{eq:lipschitz_ratio}.

    We next employ \eqref{eq:lipschitz_ratio} to bound the total variation distance.
    Write $r\coloneq 2R_{\Gamma}\|z-z'\|_2$. By \eqref{eq:lipschitz_ratio}, 
    we have
    \[G_H(A|z')-G_H(A|z)\leq(1-e^{-r})G_H(A|z')\leq 1-e^{-r}\] for every Borel set $A$, and the same bound holds with $z$ and $z'$ exchanged, which proves the bound on the total variation distance.

    We lastly bound the $1$-Wasserstein distance.
    Let $\mu\coloneq G_H(\cdot|z)$, $\nu\coloneq G_H(\cdot|z')$, and $\epsilon\coloneq\|\mu-\nu\|_{\mathrm{TV}}$. By \eqref{eq:lipschitz_ratio}, $\mu$ and $\nu$ are mutually absolutely continuous, so that $\kappa\coloneq\min\{d\mu/d\nu,1\}\,\nu$, the common part of $\mu$ and $\nu$, is well-defined and has total mass $1-\epsilon$. If $\epsilon=0$, then $\mu=\nu$ and there is nothing to prove. Otherwise, the coupling that leaves $\kappa$ on the diagonal and couples the remainders $\mu-\kappa$ and $\nu-\kappa$, each of mass $\epsilon$, by their product divided by $\epsilon$ transports mass $\epsilon$ over a distance at most $\mathrm{diam}(B)$, so that we have
    \[W_1(\mu,\nu)\leq\mathrm{diam}(B)\,\epsilon\leq\mathrm{diam}(B)(1-e^{-r})
    \leq\mathrm{diam}(B) r
    =2R_{\Gamma}\mathrm{diam}(B) \|z-z'\|_{2},
    \]
    which completes the proof.
\end{proof}

\paragraph{Proof of Theorem \ref{thm:uat_gate}}

\begin{proof}
    The proof of Theorem \ref{thm:uat_gate} consists of the following steps:
    \begin{enumerate}
    \item First, we divide the covariate space $\mathcal{Z}$ into finitely many Voronoi cells by using a $\delta$-net. These cells are represented by the affine scores \eqref{def:affine_score}, whose slopes are realized by the fused slopes of Condition \ref{cond:fusion} and whose intercepts are realized by the weights $\alpha$ through the reparametrization \eqref{eq:alpha_reparametrization}. \item Second, we show that the aggregated softmax gates $g^{\dagger}_c$ uniformly approximate a Voronoi cell-wise step function that takes a constant arbitrary probability vector in the interior of each Voronoi cell. 
    \item Finally, by refining the $\delta$-net to reduce the diameter of the Voronoi cells, the uniformly continuous function $\pi^{\ast}(z)$ is approximated by a Voronoi cell-wise step function. As a result, we can approximate the function $\pi^{\ast}(z)$ through a Voronoi cell-wise step function and its approximation by the aggregated softmax gates.
    \end{enumerate}

    \paragraph{Step 1.}
    Because the covariate space $\mathcal{Z}$ is a compact set,
    we can take a $\delta$-net $\{z_1,\ldots,z_L\}\subset \mathcal{Z}$ of distinct points so that $\mathcal{Z}\subset\cup_{l=1}^{L}\mathrm{Ball}(z_l,\delta)$. 
    Then, we partition the covariate space $\mathcal{Z}$ into $L$ Voronoi cells using the $\delta$-net $\{z_l\}_{l=1,\ldots,L}$ defined as
    \begin{equation}
        \label{def:voronoi_cell}
        V_l \coloneq\left\{z\in\mathcal{Z}: \|z-z_l\|^{2}\leq \|z-z_{l^{\prime}}\|^{2},\quad \text{for all}\;l^{\prime}\neq l \right\},\quad l=1,\ldots, L.
    \end{equation}
    Expanding the squared norms, we can write each cell as
    \begin{equation}
        V_l = \left\{z\in\mathcal{Z}: 2z^{\top}z_l-\|z_l\|^2\geq 2z^{\top}z_{l^{\prime}}-\|z_{l^{\prime}}\|^2,\quad \text{for all}\;l^{\prime}\neq l \right\},\quad l=1,\ldots, L.
    \end{equation}
    By setting $\xi_l\coloneq 2z_l$ and $\xi_{l0}\coloneq-\|z_l\|^2$ in the affine score \eqref{def:affine_score}, that is, $h_l(z)=2z^{\top}z_l-\|z_l\|^{2}$, we can express the Voronoi cells \eqref{def:voronoi_cell} by affine Voronoi cells
    \begin{equation}
        \label{def:affine_voronoi_cell}
        V_l = \left\{z\in\mathcal{Z}: h_l(z)\geq h_{l^{\prime}}(z)\quad  \text{for all}\;l^{\prime}\neq l \right\},\quad l=1,\ldots, L.
    \end{equation}
    The slopes $\xi_l$ in \eqref{def:affine_voronoi_cell} are realized by the softmax gates through the fused slopes \[\gamma_j\coloneq\xi_l/\tau=2z_l/\tau
    \quad\text{for}\quad j\in I_l\]
    that satisfy Condition \ref{cond:fusion}, and the intercepts $\xi_{l0}$ are realized through the weights $\alpha$ by reparametrization \eqref{eq:alpha_reparametrization}. 
    Hence the partition $\{V_l\}$ is realizable by the softmax gates. 
    
    \paragraph{Step 2.}    
    We will show that $g^{\dagger}_{c}$ uniformly converges to an arbitrary probability vector in the interior of each affine Voronoi cell \eqref{def:affine_voronoi_cell}.
    
    Let a gap function for $z\in \mathcal{Z}$ be
    \begin{equation}
        \Delta_l(z)\coloneq \min_{l^{\prime}\neq l}\{h_l(z)-h_{l^{\prime}}(z)\},\qquad l=1,\ldots,L,
    \end{equation}
    where this is positive in the interior of $V_l$.
    We then construct the index subsets $\mathcal{S}_{c}$. Following the fused slope structure of Appendix \ref{app:fused}, we take $J=LC$ gates and partition the index set $\{1,\ldots,J\}$ into the $L$ groups $I_1,\ldots,I_L$ with $|I_l|=C$ for each $l$, so that Condition \ref{cond:fusion} holds with the slopes assigned in Step 1. For each $l=1,\ldots,L$, fix an arbitrary ordering of the elements of $I_l$ and write
    \begin{equation}
        I_l = \{j_{l,1},j_{l,2},\ldots,j_{l,C}\}.
    \end{equation}
    Define
    \begin{equation}
        \mathcal{S}_{c}\coloneq\{j_{l,c}: l=1,\ldots,L\},\quad c=1,\ldots,C.
    \end{equation}
    By construction, $|\mathcal{S}_{c}|=L$, the sets $\mathcal{S}_{1},\ldots,\mathcal{S}_{C}$ are pairwise disjoint with $\cup_{c=1}^{C}\mathcal{S}_{c}=\{1,\ldots,J\}$, and each $\mathcal{S}_{c}$ intersects each $I_l$ in exactly one element, namely $j_{l,c}$.

    With this notation and the fused expression \eqref{eq:fused_gate}, the $c$-th element of the vectorized and aggregated softmax gate $g^{\dagger}_c$ is written, for a fixed $l$, as
    \begin{equation}
        \label{eq:aggregated_gate_fused}
        \begin{split}
            g^{\dagger}_{c}(z;\alpha,\gamma) & = \frac{\sum_{j\in\mathcal{S}_c}\alpha_{j}e^{z^{\top}\gamma_j}}{\sum_{k=1}^{J}\alpha_{k}e^{z^{\top}\gamma_{k}}}
            = \frac{\sum_{l^{\prime}=1}^{L}w_{j_{l^{\prime},c}}\exp\{h_{l^{\prime}}(z)/\tau\}}{\sum_{l^{\prime}=1}^{L}(\sum_{k\in I_{l^{\prime}}}w_{k})\exp\{h_{l^{\prime}}(z)/\tau\}} \\
                                                     & = \frac{w_{j_{l,c}}+\sum_{l^{\prime}\neq l}w_{j_{l^{\prime},c}}\exp\{(h_{l^{\prime}}(z)-h_{l}(z))/\tau\}}
                                                     {(\sum_{k\in I_{l}}w_{k})+ \sum_{l^{\prime}\neq l}(\sum_{k\in I_{l^{\prime}}}w_{k})\exp\left\{(h_{l^{\prime}}(z)-h_{l}(z))/\tau\right\}}
                                                     =\frac{N_{lc}(z)}{D_{l}(z)}
                                                     ,
        \end{split}
    \end{equation}
    where 
    \begin{equation}
    \begin{split}
    N_{lc}(z)&:=w_{j_{l,c}}+\sum_{l^{\prime}\neq l}w_{j_{l^{\prime},c}}\exp\{(h_{l^{\prime}}(z)-h_{l}(z))/\tau\}, \quad\text{and}\\
    D_{l}(z)&:=(\sum_{k\in I_{l}}w_{k})+ \sum_{l^{\prime}\neq l}(\sum_{k\in I_{l^{\prime}}}w_{k})\exp\left\{(h_{l^{\prime}}(z)-h_{l}(z))/\tau\right\}.
    \end{split}
    \end{equation}

    We proceed to bounding $\pi_l-g^{\dagger}(z;\alpha,\gamma)$.
    Let $a_{lc}\coloneq w_{j_{l,c}}$, $s_l\coloneq\sum_{k\in I_l}w_k$, and $e_l(z)\coloneq\exp(-\Delta_l(z)/\tau)$. Since $h_{l^{\prime}}(z)-h_l(z)\leq-\Delta_l(z)$ for every $l^{\prime}\neq l$ and $\sum_{j=1}^{J}w_j=1$, the numerator $N_{lc}(z)$ and the denominator $D_l(z)$ of the last expression in \eqref{eq:aggregated_gate_fused} satisfy
    \begin{equation}
        \label{eq:num_den_bounds}
        a_{lc}\leq N_{lc}(z)\leq a_{lc}+e_l(z),\qquad s_l\leq D_l(z)\leq s_l+e_l(z).
    \end{equation}
    Now, for arbitrary probability vectors $\pi_l\coloneq(\pi_{l1},\ldots, \pi_{lC})^{\top}\in\mathring\Delta^{C-1}$, $l=1,\ldots,L$, we choose the within-group weights as
    \begin{equation}
        \label{eq:w_choice}
        w_{j_{l,c}}\coloneq\frac{\pi_{lc}}{L},\qquad c=1,\ldots,C,\quad l=1,\ldots,L.
    \end{equation}
    Then $w_j>0$ for every $j$, $\sum_{j=1}^{J}w_j=\sum_{l=1}^{L}\sum_{c=1}^{C}\pi_{lc}/L=1$, $s_l=1/L$, and $a_{lc}/s_l=\pi_{lc}$, so that $w\in\Delta^{J-1}$ is admissible in \eqref{eq:alpha_reparametrization}. By \eqref{eq:num_den_bounds}, 
    we have
    \[\pi_{lc}D_l(z)-N_{lc}(z)\leq \pi_{lc}(s_l+e_l(z))-a_{lc}=\pi_{lc}e_l(z)\leq e_l(z)\]
    and 
    \[\pi_{lc}D_l(z)-N_{lc}(z)\geq \pi_{lc}s_l-a_{lc}-e_l(z)=-e_l(z).
    \]
    Therefore, for every $c=1,\ldots,C$ and $z\in\mathcal{Z}$, we have
    \begin{equation}
        \label{eq:gate_pointwise_convergence}
        |\pi_{lc}- g^{\dagger}_c(z;\alpha,\gamma)|
        = \frac{|\pi_{lc}D_l(z)-N_{lc}(z)|}{D_l(z)}
        \leq \frac{e_l(z)}{s_l}
        = L\exp\left(-\frac{\Delta_l(z)}{\tau}\right),
    \end{equation}
    and hence $\|\pi_l-g^{\dagger}(z;\alpha,\gamma)\|_{\infty}\leq L\exp(-\Delta_l(z)/\tau)$. For $z$ in the interior of $V_l$, we have $\Delta_l(z)>0$ and the right-hand side vanishes as $\tau\rightarrow+0$, which gives pointwise convergence.

    We will see uniform convergence in the interior of $V_{l}$, and evaluate $\Delta_{l}(z)$ from below.
    We define the boundary hyperplanes of the affine Voronoi cells as
    \begin{equation}
        \mathcal{N}_{l,l^{\prime}}\coloneq \left\{z\in\mathbb{R}^{d}: h_l(z)=h_{l^{\prime}}(z)\right\}, \quad l,l^{\prime}=1,\ldots,L,\;l\neq l^{\prime},
    \end{equation}
    and the union of all boundary hyperplanes as $\mathcal{N}=\cup_{l\neq l^{\prime}}\mathcal{N}_{l,l^{\prime}}$. Let the non-near-boundary subset be $\mathcal{Z}^{(\eta)}\coloneq\{z\in\mathcal{Z}: d(z,\mathcal{N})\geq \eta \}$ with $\eta>0$, where $d(z,\mathcal{N})$ denotes the Euclidean distance from $z$ to $\mathcal{N}$.
    Then, for $z\in V_{l}$ and $l^{\prime}\neq l$, we have
    \begin{equation}
        \begin{split}
            h_l(z)-h_{l^{\prime}}(z) & = \|z-z_{l^{\prime}}\|^{2}-\|z-z_{l}\|^2                               \\
                                     & = (\|z-z_{l^{\prime}}\|+\|z-z_{l}\|)(\|z-z_{l^{\prime}}\|-\|z-z_{l}\|) \\
                                     & = 2\|z_{l}-z_{l^{\prime}}\|\,d(z,\mathcal{N}_{l,l^{\prime}})             \\
                                     & \geq 2\|z_l-z_{l^{\prime}}\|\, d(z,\mathcal{N}).
        \end{split}
    \end{equation}
    In the third equality, we used the formula for the distance from a point to the hyperplane $\mathcal{N}_{l,l^{\prime}}$, together with $\xi_l-\xi_{l^{\prime}}=2(z_l-z_{l^{\prime}})$ and $\|z-z_{l}\|\leq\|z-z_{l^{\prime}}\|$ for $z\in V_l$, namely
    \begin{equation}
        \begin{split}
            d(z,\mathcal{N}_{l,l^{\prime}}) & = \frac{|h_l(z)-h_{l^{\prime}}(z)|}{\|\xi_{l}-\xi_{l^{\prime}} \|}
            =\frac{(\|z-z_{l^{\prime}}\|+\|z-z_{l}\|)\cdot(\|z-z_{l^{\prime}}\|-\|z-z_{l}\|)}{2\|z_{l^{\prime}}-z_{l}\|}.
        \end{split}
    \end{equation}
    Let $d_{\min}\coloneq\min_{l\neq l^{\prime}}\|z_l-z_{l^{\prime}}\|$, where it is positive because the points of the $\delta$-net are distinct.
    Therefore, we have $\Delta_{l}(z)\geq 2d_{\min}\eta$ for all $z\in\mathcal{Z}^{(\eta)}\cap V_l$ and all $l$, 
    and the inequality \eqref{eq:gate_pointwise_convergence} yields
    \begin{equation}
        \sup_{z\in \mathcal{Z}^{(\eta)}\cap V_{l}}\|\pi_{l}-g^{\dagger}(z;\alpha,\gamma)\|_{\infty} \leq L\exp\left(-\frac{2d_{\min}\eta}{\tau}\right),\quad l=1,\ldots,L,
    \end{equation}
    whose right-hand side does not depend on $z$ or $l$. As a result, for fixed $\delta$, $\eta$ and
    all $\epsilon>0$, there exists $\tau>0$ such that
    \begin{equation}
        \label{eq:uniform_cell}
        \sup_{z\in \mathcal{Z}^{(\eta)}\cap V_{l}}\|\pi_{l}-g^{\dagger}(z;\alpha,\gamma)\|_{\infty} <\frac{\epsilon}{2},\quad l=1,\cdots,L.
    \end{equation}

    \paragraph{Step 3.}

    Since $\pi^{\ast}(z)$ is a continuous function on the compact set $\mathcal{Z}$, it is also a uniformly continuous function. Thus, for all $\epsilon>0$, there exists $\delta^{\prime}>0$ such that
    \begin{equation}
        \label{eq:unif_cont}
        \sup_{\|z-z^{\prime}\|<\delta^{\prime}}\|\pi^{\ast}(z)-\pi^{\ast}(z^{\prime})\|_{\infty} <\frac{\epsilon}{2}.
    \end{equation}
    Thus, if we take $\delta>0$ such that $\mathrm{diam}(V_l)<\delta^{\prime}$,
    we have
    \begin{equation}
        \label{eq:within_cell_target}
        \sup_{z\in V_l}\|\pi^{\ast}(z)-\pi^{\ast}(z_{l})\|_{\infty} <\frac{\epsilon}{2}.
    \end{equation}
    Note that $\mathrm{diam}(V_l)\leq2\delta$, since any $z\in V_l$ satisfies $\|z-z_l \|\leq\delta$. Setting $\pi_l\coloneq\pi^{\ast}(z_l)\in\mathring\Delta^{C-1}$ in \eqref{eq:w_choice} and \eqref{eq:uniform_cell}, we have
    \begin{equation}
        \begin{split}
            \sup_{z\in V_l\cap\mathcal{Z}^{(\eta)}} \|\pi^{\ast}(z)-g^{\dagger}(z;\alpha,\gamma) \|_{\infty} & \leq \sup_{z\in V_l\cap\mathcal{Z}^{(\eta)}} \|\pi^{\ast}(z)- \pi^{\ast}(z_l) \|_{\infty} + \sup_{z\in V_l\cap\mathcal{Z}^{(\eta)}} \|\pi^{\ast}(z_l) - g^{\dagger}(z;\alpha,\gamma) \|_{\infty} \\
                                                                                                             & < \frac{\epsilon}{2}+ \frac{\epsilon}{2}=\epsilon.
        \end{split}
    \end{equation}
    Splitting the integral over $\mathcal{Z}$ into the integrals over $\mathcal{Z}^{(\eta)}$ and over $\mathcal{Z}\backslash\mathcal{Z}^{(\eta)}$, we obtain
    \begin{equation}
        \label{eq:integral_decomp}
        \begin{split}
            \int_{\mathcal{Z}}\|g^{\dagger}(z;\alpha,\gamma)-\pi^{\ast}(z)\|_\infty^{q}\dd z & = \int_{\mathcal{Z}^{(\eta)}}\|g^{\dagger}(z;\alpha,\gamma)-\pi^{\ast}(z)\|_\infty^{q}\dd z+\int_{\mathcal{Z}\backslash\mathcal{Z}^{(\eta)}}\|g^{\dagger}(z;\alpha,\gamma)-\pi^{\ast}(z)\|_\infty^{q}\dd z \\
                                                                                          & < \epsilon^{q}\mathrm{Leb}(\mathcal{Z}^{(\eta)}) + \mathrm{Leb}(\mathcal{Z}\backslash \mathcal{Z}^{(\eta)}),
        \end{split}
    \end{equation}
    where  $\mathrm{Leb}(\cdot)$ denotes the Lebesgue measure.
    In this inequality, we used the trivial bound 
    \[\sup_{z\in\mathcal{Z}} \|g^{\dagger}(z;\alpha,\gamma)-\pi^{\ast}(z)\|_{\infty}\leq1 \]
    on $\mathcal{Z}\backslash\mathcal{Z}^{(\eta)}$.

    We now bound each term in \eqref{eq:integral_decomp} and specify the order in which the parameters $\epsilon$, $\delta$, $\eta$, and $\tau$ are chosen. For the first term, $\mathrm{Leb}(\mathcal{Z}^{(\eta)})\leq\mathrm{Leb}(\mathcal{Z})<\infty$ holds, since $\mathcal{Z}$ is bounded. For the second term, the boundary $\mathcal{N}$ is contained in the union of at most $\binom{L}{2}$ affine hyperplanes, and the $\eta$-neighborhood of each hyperplane intersected with $\mathrm{ball}(0,\rho)$ has Lebesgue measure at most $2\eta\,\omega_{d-1}\rho^{d-1}$, where $\omega_{d-1}$ denotes the volume of the unit ball in $\mathbb{R}^{d-1}$. Hence
    \begin{equation}
        \label{eq:boundary_measure}
        \mathrm{Leb}(\mathcal{Z}\backslash\mathcal{Z}^{(\eta)})\leq \binom{L}{2}\cdot2\eta\, \omega_{d-1} \rho^{d-1}\leq \eta L^{2}\omega_{d-1}\rho^{d-1},
    \end{equation}
    since $2\binom{L}{2}=L(L-1)\leq L^{2}$.

    Given a target accuracy $\epsilon_0>0$, the parameters are chosen in the following order. First, choose $\epsilon>0$ so that $\epsilon^{q}\mathrm{Leb}(\mathcal{Z})<\epsilon_0^{q}/2$. Second, this choice of $\epsilon$ determines $\delta^{\prime}>0$ through the uniform continuity bound \eqref{eq:unif_cont}, and hence determines $\delta$ and the net size $L=L(\delta)$ through \eqref{eq:within_cell_target}. Third, with $L$ fixed, choose $\eta>0$ so that $\eta L^{2}\omega_{d-1}\rho^{d-1}<\epsilon_0^{q}/2$, which by \eqref{eq:boundary_measure} guarantees $\mathrm{Leb}(\mathcal{Z}\backslash\mathcal{Z}^{(\eta)})<\epsilon_0^{q}/2$. Finally, with $\delta$ and $\eta$ fixed, choose $\tau>0$ so that the uniform bound \eqref{eq:uniform_cell} holds with the chosen $\epsilon$. Substituting these bounds into \eqref{eq:integral_decomp} yields
    \begin{equation}
        \int_{\mathcal{Z}}\|g^{\dagger}(z;\alpha,\gamma)-\pi^{\ast}(z)\|_\infty^{q}\dd z < \frac{\epsilon_0^{q}}{2}+\frac{\epsilon_0^{q}}{2}=\epsilon_0^{q},
    \end{equation}
    and taking the $q$-th root gives
    \begin{equation}
        \begin{split}
            \left(\int_{\mathcal{Z}}\|g^{\dagger}(z;\alpha,\gamma)-\pi^{\ast}(z)\|_\infty^{q}\dd z\right)^{1/q} <  \epsilon_0.
        \end{split}
    \end{equation}
    The parameters of the theorem are the slopes $\gamma_j=2z_l/\tau$ for $j\in I_l$, which satisfy Condition \ref{cond:fusion} and carry no intercept, and the weights $\alpha\in\Delta^{J-1}$ given by \eqref{eq:alpha_reparametrization} with the within-group weights \eqref{eq:w_choice}, the intercepts $\xi_{l0}=-\|z_l\|^{2}$, and the chosen $\tau$.
    This completes the proof of Theorem \ref{thm:uat_gate}.
\end{proof}

\end{document}